\documentclass[english,compsoc]{IEEEtran}
\usepackage[T1]{fontenc}
\usepackage[latin9]{inputenc}
\usepackage{amsmath}
\usepackage{amssymb}
\usepackage{graphicx}
\usepackage{textcomp}

\newcommand{\subparagraph}{}
\usepackage{titlesec}

\usepackage{url}

\makeatletter

\newcommand{\lyxdot}{.}

\IEEEoverridecommandlockouts

\usepackage{algorithmic}
\usepackage{amsthm}
\usepackage{subfigure}
\usepackage{float}
\usepackage{amsmath}
\usepackage{graphicx}
\usepackage{setspace}
\usepackage{bbm}
\usepackage{multirow}

\floatstyle{ruled}
\newfloat{algorithm}{tbp}{loa}
\providecommand{\algorithmname}{Algorithm}
\floatname{algorithm}{\protect\algorithmname}

\newtheorem{proposition}{Proposition}
\newtheorem{corollary}{Corollary}

\newtheorem{lemma}{Lemma}
\newtheorem{definition}{Definition}

\newtheorem{assumption}{Assumption}

\usepackage{babel}
\usepackage{upquote}
\usepackage{booktabs}\usepackage{bm}\usepackage{cite}

\usepackage{tabularx}

\usepackage{comment}

\usepackage{enumitem}
\setlist{noitemsep}

\titlespacing\section{0pt}{10pt plus 4pt minus 2pt}{3pt plus 2pt minus 2pt}
\titlespacing\subsection{0pt}{10pt plus 4pt minus 2pt}{3pt plus 2pt minus 2pt}
\titlespacing\subsubsection{0pt}{7pt plus 4pt minus 2pt}{3pt plus 2pt minus 2pt}

\allowdisplaybreaks[1]

\makeatother

\usepackage{babel}
\begin{document}

\title{Dynamic Service Placement for Mobile Micro-Clouds with Predicted Future Costs}

\author{
\IEEEauthorblockN{Shiqiang Wang, Rahul Urgaonkar, Ting He, Kevin Chan, Murtaza Zafer, and Kin K. Leung}

\IEEEcompsocitemizethanks{
\vspace{-0.05in}
\IEEEcompsocthanksitem{S. Wang is with IBM T. J. Watson Research Center, Yorktown Heights, NY, United States. Email: wangshiq@us.ibm.com}

\IEEEcompsocthanksitem{R. Urgaonkar is with Amazon Inc., Seattle, WA, United States. Email: rahul.urgaonkar@gmail.com}

\IEEEcompsocthanksitem{T. He is with the School of Electrical Engineering and Computer Science, Pennsylvania State University, University Park, PA, United States. Email: tzh58@psu.edu}

\IEEEcompsocthanksitem{K. Chan is with the Army Research Laboratory, Adelphi, MD, United States. Email: kevin.s.chan.civ@mail.mil}

\IEEEcompsocthanksitem{M. Zafer is with Nyansa Inc., Palo Alto, CA, United States. Email: murtaza.zafer.us@ieee.org}

\IEEEcompsocthanksitem{K. K. Leung is with the Department of Electrical and Electronic Engineering, Imperial College London, United Kingdom. Email: kin.leung@imperial.ac.uk}

\IEEEcompsocthanksitem{A preliminary version of this paper was presented at IEEE ICC 2015 \cite{wang2015dynamic}.}

\IEEEcompsocthanksitem{This is the author's version of the paper accepted for publication in the IEEE Transactions on Parallel and Distributed Systems,  DOI: 10.1109/TPDS.2016.2604814.  \newline
\textcopyright 2016 IEEE. Personal use of this material is permitted. Permission from IEEE must be obtained for all other uses, in any current or future media, including reprinting/republishing this material for advertising or promotional purposes, creating new collective works, for resale or redistribution to servers or lists, or reuse of any copyrighted component of this work in other works.
}

}

}

\IEEEaftertitletext{\vspace{-0.4in}}

\IEEEtitleabstractindextext{
\begin{abstract}
Mobile micro-clouds are promising for enabling performance-critical
cloud applications. However, one challenge therein is the dynamics
at the network edge. In this paper, we study how to place service
instances to cope with these dynamics, where multiple users and service
instances coexist in the system. Our goal is to find the optimal placement
(configuration) of instances to minimize the average cost over time,
leveraging the ability of predicting future cost parameters with known
accuracy. We first propose an offline algorithm that solves for the
optimal configuration in a specific look-ahead time-window. Then,
we propose an online approximation algorithm with polynomial time-complexity
to find the placement in real-time whenever an instance arrives. We
analytically show that the online algorithm is $O(1)$-competitive
for a broad family of cost functions. Afterwards, the impact of prediction
errors is considered and a method for finding the optimal look-ahead
window size is proposed, which minimizes an upper bound of the average
actual cost. The effectiveness of the proposed approach is evaluated
by simulations with both synthetic and real-world (San Francisco taxi)
user-mobility traces. The theoretical methodology used in this paper can potentially
be applied to a larger class of dynamic resource allocation problems. \end{abstract}
\begin{IEEEkeywords}
Cloud computing, fog/edge computing, online approximation
algorithm, optimization, resource allocation, wireless networks
\end{IEEEkeywords}
}

\maketitle

\setstretch{0.98}

\section{Introduction}

Many emerging applications, such as video streaming, real-time face/object
recognition, require high data processing capability. However, portable
devices (e.g. smartphones) are generally limited by their size and
battery life, which makes them incapable of performing complex computational
tasks. A remedy for this is to utilize cloud computing techniques,
where the cloud performs the computation for its users. In the traditional
setting, cloud services are provided by centralized data-centers that
may be located \emph{far away} from end-users, which can be inefficient
because users may experience long latency and poor connectivity due
to long-distance communication \cite{cloudletCMUMobiCASE}. The newly
emerging idea of mobile micro-clouds (MMCs) is to place the cloud
\emph{closer} to end-users, so that users can have fast and reliable
access to services. A small-sized server cluster hosting an MMC is
directly connected to a network component at the network edge. For
example, it can be connected to the wireless basestation, as proposed
in \cite{IBMWhitepaper} and \cite{CommMagEdgeComput}, providing
cloud services to users that are either connected to the basestation
or are within a reasonable distance from it. It can also be connected
to other network entities that are in close proximity to users. Fig.
\ref{fig:scenario} shows an application scenario where MMCs coexist
with a backend cloud. MMCs can be used for many applications that
require high reliability or high data processing capability \cite{cloudletCMUMobiCASE}.
Similar concepts include cloudlet \cite{cloudletCMUMobiCASE},
follow me cloud \cite{FollowMeMagazine}, fog computing, edge computing \cite{Edge-as-a-service},
small cell cloud \cite{becvar2014pimrc}, etc. We use the term MMC
in this paper. 

\begin{figure}
\center{\includegraphics[width=0.8\linewidth]{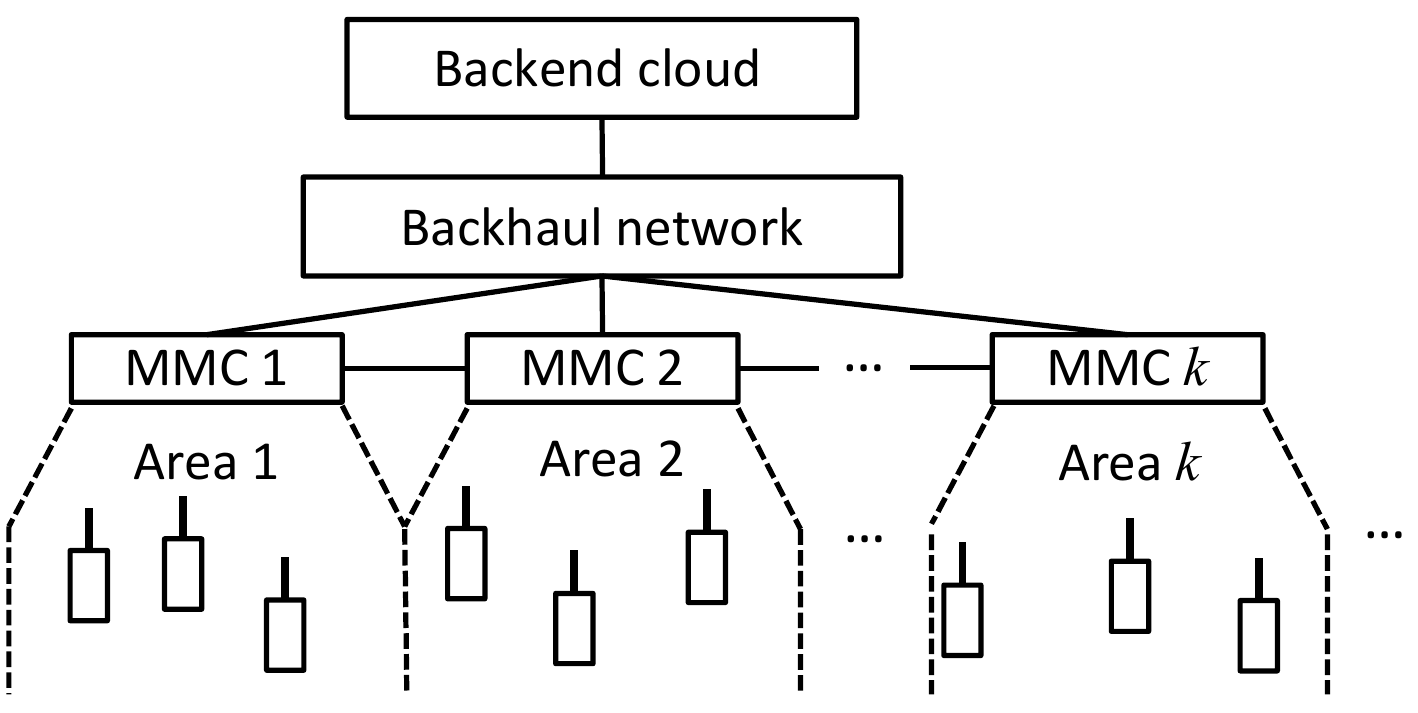}}

\protect\caption{Application scenario.}
\label{fig:scenario}
\end{figure}

One important issue in MMCs is to decide which MMC should perform
the computation for a particular user or a set of users, taking into
account user mobility and other dynamic changes in the network. Providing
a service to a user (or a set of users) requires starting a service
instance, which can be run either in the backend cloud or in one of
the MMCs, and the question is how to choose the optimal location to
run the service instance. Besides, users may move across different
geographical areas due to mobility, thus another question is whether
we should migrate the service instance from one cloud (which can be
either an MMC or the backend cloud) to another cloud when the user
location or network condition changes. For every cloud, there is a
cost\footnote{The term ``cost'' in this paper is an abstract notion that can stand for monetary cost, service access latency of users, service interruption time, amount of transmission/processing resource consumption, etc.} associated with running the service instance in it, and there
is also a cost associated with migrating the service instance from
one cloud to another cloud. The placement and migration of service
instances therefore needs to properly take into account this cost.

\subsection{Related Work}

The abovementioned problems are related to application/workload placement
problems in cloud environments. Although existing work has studied
such problems under complex network topologies \cite{VNEmbeddingSurvey,refViNEYard},
they mainly focused on relatively static network conditions and fixed resource
demands in a data-center environment. The presence of dynamically changing resource availability
that is related to user mobility in an MMC environment has not been sufficiently considered. 

When user mobility exists, it is necessary to consider real-time (live)
migration of service instances. For example, it can be beneficial
to migrate the instance to a location closer to the user. Only a few
existing papers in the literature have studied this problem \cite{MDPFollowMeICC2014,wang2014milcom,wang2015IFIPNetworking}.
The main approach in \cite{MDPFollowMeICC2014,wang2014milcom,wang2015IFIPNetworking}
is to formulate the mobility-driven service instance migration problem
as a Markov decision process (MDP). Such a formulation is suitable
where the user mobility follows or can be approximated by a mobility
model that can be described by a Markov chain. However, there are
cases where the Markovian assumption is not valid \cite{SrivatsaNonMarkovianMobility}.
Besides, \cite{MDPFollowMeICC2014,wang2014milcom,wang2015IFIPNetworking}
either do not explicitly or only heuristically consider multiple users
and service instances, and they assume specific structures of the
cost function that are related to the locations of users and service
instances. Such cost structures may be inapplicable when the load
on different MMCs are imbalanced or when we consider the backend cloud
as a placement option. In addition, the existing MDP-based approaches
mainly consider service instances that are constantly running in the
cloud system; they do not consider instances that may arrive to and
depart from the system over time.

Systems with online (and usually unpredictable) arrivals and departures
have been studied in the field of online approximation algorithms
\cite{bookOnlineComputation,krumke2006online}. The goal is to design
efficient algorithms (usually with polynomial time-complexity) that
have reasonable competitive ratios%
\footnote{We define the \emph{competitive ratio} as the maximum ratio of the
cost from the online approximation algorithm to the true optimal cost
from offline placement. %
}. However, most existing work focus on problems that can be formulated
as integer \emph{linear} programs. Problems that have \emph{convex
but non-linear} objective functions have attracted attention only
very recently \cite{AzarOnlineConvexObj14,BuchbinderOnlinePackingCovering},
where the focus is on online covering problems in which new constraints
arrive over time. Our problem is different from the existing work
in the sense that the online arrivals in our problem are abstracted
as \emph{change }in constraints (or, with a slightly different but
equivalent formulation, adding new \emph{variables}) instead of adding
new constraints, and we consider the average cost over multiple timeslots.
Meanwhile, online departures are not considered in \cite{AzarOnlineConvexObj14,BuchbinderOnlinePackingCovering}.

Concurrently with the work presented in this paper, we have considered
non-realtime applications in \cite{urgaonkar2015performance}, where
users submit job requests that can be processed after some time. Different
from \cite{urgaonkar2015performance}, we consider users continuously
connected to services in this paper, which is often the case for delay-sensitive
applications (such as live video streaming). The technical approach
in this paper is fundamentally different from that in \cite{urgaonkar2015performance}.
Besides, a Markovian mobility model is still assumed in \cite{urgaonkar2015performance}.

Another related problem is the load balancing in distributed systems, where the goal is to even out the load distribution across machines.
Migration cost, future cost parameter prediction and the impact of prediction error are not considered in load balancing problems \cite{refViNEYard,hsiao2013load,AzarLoadBalancingSurvey,LiOptimalDynamicMobility,LiLoadBalancing}.
We consider all these aspects in this paper, and in addition, we consider a generic cost definition that can be defined to favor load balancing as well as other aspects. 

We also note that existing online algorithms with provable performance
guarantees are often of theoretical nature \cite{bookOnlineComputation,krumke2006online,AzarOnlineConvexObj14,BuchbinderOnlinePackingCovering},
which may not be straightforward to apply in practical systems because these algorithms can be conceptually complex thus difficult to understand. At
the same time, most online algorithms applied in practice are of heuristic
nature without theoretically provable optimality guarantees \cite{VNEmbeddingSurvey}; the performance of such algorithms are usually evaluated under a specific experimentation setting (see references of  \cite{VNEmbeddingSurvey}), thus they may perform
poorly under other settings that possibly occur in practice \cite{chen1999ordinal}. For example, in the machine scheduling problem considered in \cite{Aspnes:1997:ORV:258128.258201}, a greedy algorithm (which is a common heuristic) that works well in some cases does not work well in other cases. We propose a simple and
practically applicable online algorithm with theoretically provable performance
guarantees in this paper, and also verify its performance with simulation
using both synthetic arrivals and real-world user traces.


\subsection{Main Contributions \label{sub:Main-Contributions}}

In this paper, we consider a general setting which allows heterogeneity
in cost values, network structure, and mobility models. We assume
that the cost is related to a finite set of parameters, which can
include the locations and preferences of users, load in the system,
database locations, etc. We focus on the case where there is an underlying
mechanism to predict the future values of these parameters, and also
assume that the prediction mechanism provides the most likely future
values and an upper bound on possible deviation of the actual value
from the predicted value. Such an assumption is valid for many prediction
methods that provide guarantees on prediction accuracy. Based on the
predicted parameters, the (predicted) future costs of each configuration
can be found, in which each configuration represents one particular
placement sequence of service instances. 

With the above assumption, we formulate a problem of finding the optimal
configuration of service instances that minimizes the average cost
over time. We define a look-ahead window to specify the amount of
time that we look (predict) into the future. The main contributions
of this paper are summarized as follows:
\begin{enumerate}
\item We first focus on the \emph{offline problem }of service instance placement
using predicted costs within a specific look-ahead window, where the
instance arrivals and departures within this look-ahead window are
assumed to be known beforehand. We show that this problem is equivalent
to a shortest-path problem in a virtual graph formed by all possible
configurations, and propose an algorithm (Algorithm \ref{alg:shortestPath}
in Section \ref{sub:AlgorithmOffline}) to find its optimal solution
using dynamic programming.
\item We note that it is often practically infeasible to know in advance
about when an instance will arrive to or depart from the system. Meanwhile,
Algorithm \ref{alg:shortestPath} may have exponential time-complexity
when there exist multiple instances. Therefore, we propose an \emph{online
approximation algorithm} that finds the placement of a service instance
upon its arrival with polynomial time-complexity. The proposed online
algorithm calls Algorithm \ref{alg:shortestPath} as a subroutine
for each instance upon its arrival. We analytically evaluate the performance
of this online algorithm compared to the optimal offline placement.
The proposed online algorithm is $O(1)$-competitive%
\footnote{We say that an online algorithm is $c$\emph{-competitive }if its
competitive ratio is upper bounded by $c$.%
} for certain types of cost functions (including those which are linear,
polynomial, or in some other specific form), under some mild assumptions.
\item Considering the existence of prediction errors, we propose a method
to find the \emph{optimal look-ahead window size}, such that an upper
bound on the actual placement cost is minimized. 
\item The effectiveness of the proposed approach is evaluated by \emph{simulations
}with both synthetic traces and \emph{real-world mobility traces}
of San Francisco taxis.
\end{enumerate}

The remainder of this paper is organized as follows. The problem formulation
is described in Section \ref{sec:problemFormulation}. Section \ref{sec:optSolutionGivenLookAhead}
proposes an offline algorithm to find the optimal sequence of service
instance placement with given look-ahead window size. The online placement
algorithm and its performance analysis are presented in Section \ref{sec:onlinePlacement}.
Section \ref{sec:findLookAhead} proposes a method to find the optimal
look-ahead window size. Section \ref{sec:simulation} presents the
simulation results and Section \ref{sec:Conclusions} draws conclusions.

\section{Problem Formulation}

\label{sec:problemFormulation}

We consider a cloud computing system as shown in Fig.~\ref{fig:scenario},
where the clouds are indexed by $k\in\{1,2,...,K\}$. Each cloud $k$
can be either an MMC or a backend cloud. All MMCs together with the
backend cloud can host service instances that may arrive and leave
the system over time. A service instance is a process that is executed
for a particular task of a cloud service. Each service instance \emph{may
serve one or a group users}, where there usually exists data transfer
between the instance and the users it is serving. A time-slotted system
as shown in Fig. \ref{fig:timing} is considered, in which the actual
physical time interval corresponding to each slot $t=1,2,3,...$ can
be either the same or different. 

We consider a window-based control framework, where every $T$ slots,
a controller performs cost prediction and computes the service instance
configuration for the next $T$ slots. We define these $T$ consecutive
slots as a \emph{look-ahead window}. Service instance placement within
each window is found either at the beginning of the window (in the
offline case) or whenever an instance arrives (in the online case).
We limit ourselves to within one look-ahead window when finding the
configuration. In other words, we do not attempt to find the placement
in the next window until the time for the current window has elapsed
and the next window starts. Our solution can also be extended to a
slot-based control framework where the controller computes the next
$T$-slot configuration at the beginning of every slot, based on predicted
cost parameters for the next $T$ slots. We leave the detailed comparison
of these frameworks and their variations for future work. 

\begin{figure}
\centering \includegraphics[width=1\columnwidth]{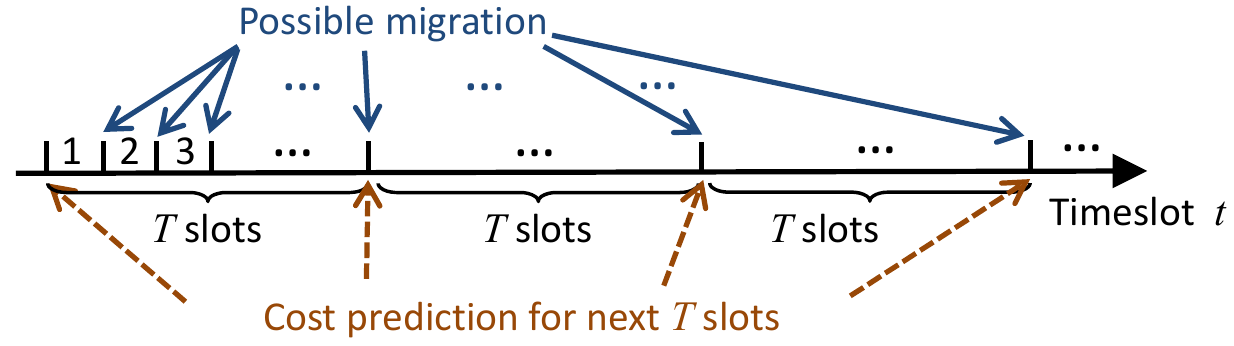}

\protect\caption{Timing of the proposed approach.}
\label{fig:timing}
\end{figure}

\subsection{Definitions }

We introduce some definitions in the following. A summary of main
notations is given in Appendix \ref{app:SummaryNotations}. 

\subsubsection{Service Instances \label{sub:serviceInstanceIndexing}}

We say a service instance \emph{arrives }to the system if it is created,
and we say it \emph{departs }from the system if its operation is finished.
Service instances may arrive and depart over time. We keep an index
counter to assign an index for each new instance. The counter is initialized
to zero when the cloud system starts to operate%
\footnote{This is for ease of presentation. In practice, the index
can be reset when reaching a maximum counter number,
and the definition of service configurations (defined later) can be
easily modified accordingly.%
}. Upon a service instance arrival, we increment the counter by one,
so that if the previously arrived instance has index $i$, a newly
arrived instance will have index $i+1$. With this definition, if
$i<i'$, instance $i$ arrives no later than instance $i'$. A particular
instance $i$ can only arrive to the system once, and we assume that
arrivals always occur at the beginning of a slot and departures always
occur at the end of a slot. For example, consider timeslots $t=1,2,3,4,5$,
instance $i=2$ may arrive to the system at the beginning of slot
$t=2$, and depart from the system at the end of slot $t=4$. At any
timeslot $t$, instance $i$ can have one of the following states:
not arrived, running, or departed. For the above example, instance
$i=2$ has not yet arrived to the system in slot $t=1$, it is running
in slots $t=2,3,4$, and it has already departed in slot $t=5$. Note
that an instance can be running across multiple windows each containing
$T$ slots before it departs.

\subsubsection{Service Configurations \label{sub:Configurations}}

Consider an arbitrary sequence of consecutive timeslots $t\in\{t_{0},t_{0}+1,...,t_{0}+Q-1\}$,
where $Q$ is an integer. For simplicity, assume that the instance with the smallest
index running in slot $t_{0}$ has index $i=1$, and the instance
with the largest index running in \emph{any of} the slots in $\{t_{0},...,t_{0}+Q-1\}$
has index $M$. According to the index assignment discussed in Section
\ref{sub:serviceInstanceIndexing}, there can be at most $M$ instances
running in any slot $t\in\{t_{0},...,t_{0}+Q-1\}$. 

We define a $Q$-by-$M$ matrix denoted by $\boldsymbol{\pi}$, where
its $(q,i)$th ($q\in\{1,...,Q\}$) element $(\boldsymbol{\pi})_{qi}\in\{0,1,2,...,K\}$
denotes the location of service instance $i$ in slot $t_{q}\triangleq t_{0}+q-1$ (
``$\triangleq$'' stands for ``is defined to be equal to'').
We set $(\boldsymbol{\pi})_{qi}$ according to the state of instance
$i$ in slot $t_{q}$, as follows 
\[
(\boldsymbol{\pi})_{qi}=\begin{cases}
0, & \textrm{if }i\textrm{ is not running in slot }t_{q}\\
k, & \textrm{if }i\textrm{ is running in cloud }k\textrm{ in slot }t_{q}
\end{cases}
\]
where instance $i$ is not running if it has not yet arrived or has
already departed. The matrix $\boldsymbol{\pi}$ is called the \emph{configuration}
of instances in slots $\{t_{0},...,t_{0}+Q-1\}$. Throughout this
paper, we use matrix $\boldsymbol{\pi}$ to represent configurations
in different subsets of timeslots. We write $\boldsymbol{\pi}(t_{0},t_{1},...,t_{n})$
to explicitly denote the configuration in slots $\{t_{0},t_{1},...,t_{n}\}$
(we have $Q=t_{n}-t_{0}+1$), and we write $\boldsymbol{\pi}$
for short where the considered slots can be inferred from the context.
For a single slot $t$, $\boldsymbol{\pi}(t)$ becomes
a vector (i.e., $Q=1$). 

\emph{Remark:} The configurations in different slots can
appear either in the same matrix or in different matrices. This means,
from $\boldsymbol{\pi}(t_{0},...,t_{0}+Q-1)$, we can get $\boldsymbol{\pi}(t)$
for any $t\in\{t_{0},...,t_{0}+Q-1\}$, as well as $\boldsymbol{\pi}(t-1,t)$
for any $t\in\{t_{0}+1,...,t_{0}+Q-1\}$, etc., and vice versa. For
the ease of presentation later, we define $(\boldsymbol{\pi}(0))_{i}=0$
for any $i$.


\subsubsection{Costs \label{sub:CostDef}}

The \emph{cost} can stand for different performance-related factors in practice, such as monetary cost (expressed as the price in some currency), service access latency of users (in seconds), service interruption time (in seconds), amount of transmission/processing resource consumption (in the number of bits to transfer, CPU cycles, memory size, etc.), or a combination of these. As long as these aspects can be expressed in some form of a cost function, we can treat them in the same optimization framework, thus we use the generic notion of cost in this paper.

We consider two types of costs. The \emph{local cost} $U(t,\boldsymbol{\pi}(t))$
specifies the cost of data transmission (e.g., between each pair of
user and service instance) and processing in slot $t$ when the configuration
in slot $t$ is $\boldsymbol{\pi}(t)$. Its value can depend on many
factors, including user location, network condition, load of clouds,
etc., as discussed in Section~\ref{sub:Main-Contributions}. When
a service instance is initiated in slot $t$, the local cost in slot
$t$ also includes the cost of initial placement of the corresponding
service instance(s). We then define the \emph{migration cost} $W(t,\boldsymbol{\pi}(t-1),\boldsymbol{\pi}(t))$,
which specifies the cost related to migration between slots $t-1$
and $t$, which respectively have configurations $\boldsymbol{\pi}(t-1)$
and $\boldsymbol{\pi}(t)$. There is no migration cost in the very
first timeslot (start of the system), thus we define $W(1,\cdot,\cdot)=0$.
The sum of local and migration costs in slot $t$ when following configuration
$\boldsymbol{\pi}(t-1,t)$ is given by 
\begin{equation}
C_{\boldsymbol{\pi}(t-1,t)}(t)\triangleq U(t,\boldsymbol{\pi}(t))+W(t,\boldsymbol{\pi}(t-1),\boldsymbol{\pi}(t))\label{eq:costDef}
\end{equation}
The above defined costs are \emph{aggregated costs }for all service
instances in the system during slot $t$. We will give concrete examples
of these costs later in Section \ref{sub:PerfAnalysisDefinitions}.

\subsection{Actual and Predicted Costs}

To distinguish between the actual and predicted cost values, for a
given configuration $\boldsymbol{\pi}$, we let $A_{\boldsymbol{\pi}}(t)$
denote the \emph{actual} value of $C_{\boldsymbol{\pi}}(t)$, and
let $D_{\boldsymbol{\pi}}^{t_{0}}(t)$ denote the \emph{predicted}
(most likely) value of $C_{\boldsymbol{\pi}}(t)$, when cost-parameter
prediction is performed at the beginning of slot $t_{0}$. For completeness
of notations, we define $D_{\boldsymbol{\pi}}^{t_{0}}(t)=A_{\boldsymbol{\pi}}(t)$
for $t<t_{0}$, because at the beginning of $t_{0}$, the costs of
all past timeslots are known. For $t\geq t_{0}$, we assume that the
absolute difference between $A_{\boldsymbol{\pi}}(t)$ and $D_{\boldsymbol{\pi}}^{t_{0}}(t)$
is at most 
\[
\epsilon(\tau)\triangleq\max_{\boldsymbol{\pi},t_{0}}\left|A_{\boldsymbol{\pi}}(t_{0}+\tau)-D_{\boldsymbol{\pi}}^{t_{0}}(t_{0}+\tau)\right|
\]
which represents the maximum error when looking ahead for $\tau$
slots, among all possible configurations $\boldsymbol{\pi}$ (note
that only $\boldsymbol{\pi}(t_{0}+\tau-1)$ and $\boldsymbol{\pi}(t_{0}+\tau)$
are relevant) and all possible prediction time instant $t_{0}$. The
function $\epsilon(\tau)$ is assumed to be non-decreasing with $\tau$,
because we generally cannot have lower error when we look farther
ahead into the future. The specific value of $\epsilon(\tau)$ is
assumed to be provided by the cost prediction module. 

We note that specific methods for predicting future cost parameters
are beyond the scope of this paper, but we anticipate that existing
approaches such as \cite{cloudMonitorSurvey}, \cite{Cho:2011:FMU:2020408.2020579}
and \cite{lacurts2014cicada} can be applied. For example, one simple
approach is to measure cost parameters on the current network condition,
and regard them as parameters for the future cost until the next measurement
is taken. The prediction accuracy in this case is related to how fast
the cost parameters vary, which can be estimated from historical records.
We regard these cost parameters as predictable because they are generally
related to the overall state of the system or historical pattern of
users, which are unlikely to vary significantly from its previous
state or pattern within a short time. This is \emph{different} from
arrivals and departures of instances, which can be spontaneous
and unlikely to follow a predictable pattern.

\subsection{Our Goal}

Our ultimate goal is to find the optimal configuration $\boldsymbol{\pi}^{*}(1,...,\infty)$
that minimizes the \emph{actual} average cost over a sufficiently
long time, i.e.
\begin{equation}
\boldsymbol{\pi}^{*}(1,...,\infty)\!=\!\arg\!\!\min_{\boldsymbol{\pi}(1,...,\infty)}\lim_{T_{\textrm{max}}\rightarrow\infty}\frac{\sum_{t=1}^{T_{\textrm{max}}}A_{\boldsymbol{\pi}(t-1,t)}(t)}{T_{\textrm{max}}}\label{eq:objFunc}
\end{equation}
However, it is impractical to find the optimal solution to (\ref{eq:objFunc}),
because we cannot precisely predict the future costs and also do not
have exact knowledge on instance arrival and departure events in the
future. Therefore, we focus on obtaining an approximate solution to
(\ref{eq:objFunc}) by utilizing \emph{predicted} cost values that
are collected every $T$ slots. 

Now, the service placement problem includes two parts: one is finding
the look-ahead window size $T$, discussed in Section
\ref{sec:findLookAhead}; the other is finding the configuration within
each window, where we consider both offline and online placements, discussed in Sections \ref{sec:optSolutionGivenLookAhead}
(offline placement) and \ref{sec:onlinePlacement} (online placement). The offline placement assumes that at the beginning
of window $T$, we know the exact arrival and departure times of each
instance within the rest of window $T$, whereas the online placement
does not assume this knowledge. We note that the notion of ``offline''
here does \emph{not }imply exact knowledge of future costs. Both offline
and online placements in Sections \ref{sec:optSolutionGivenLookAhead}
and \ref{sec:onlinePlacement} are based on the predicted costs $D_{\boldsymbol{\pi}}^{t_{0}}(t)$,
the actual cost $A_{\boldsymbol{\pi}}(t)$ is considered later in
Section \ref{sec:findLookAhead}.

\section{Offline Service Placement with Given Look-Ahead Window Size}

\label{sec:optSolutionGivenLookAhead}

In this section, we focus on the offline placement problem, where
the arrival and departure times of future instances are assumed to
be exactly known. We denote the configuration found for this problem
by $\boldsymbol{\pi}_{\textrm{off}}$.

\subsection{Procedure \label{sub:offlineHighLevel}}

We start with illustrating the high-level procedure of finding $\boldsymbol{\pi}_{\textrm{off}}$.
When the look-ahead window size $T$ is given, the configuration $\boldsymbol{\pi}_{\textrm{off}}$
is found sequentially for each window (containing timeslots $t_{0},...,t_{0}+T-1$),
by solving the following optimization problem: 
\begin{equation}
\boldsymbol{\pi}_{\textrm{off}}(t_{0},...,t_{0}\!+\! T\!-\!1)=\arg\!\!\min_{\boldsymbol{\pi}(t_{0},...,t_{0}+T-1)}\!\!\sum_{t=t_{0}}^{t_{0}+T-1}D_{\boldsymbol{\pi}(t-1,t)}^{t_{0}}(t)\label{eq:objFuncFrame}
\end{equation}
where $D_{\boldsymbol{\pi}}^{t_{0}}(t)$ can be found based on the
parameters obtained from the cost prediction module. The procedure
is shown in Algorithm \ref{alg:offlineHighLevel}.

In Algorithm \ref{alg:offlineHighLevel}, every time when solving
(\ref{eq:objFuncFrame}), we get the value of $\boldsymbol{\pi}_{\textrm{off}}$
for additional $T$ slots. This is sufficient in practice (compared
to an alternative approach that directly solves for $\boldsymbol{\pi}_{\textrm{off}}$
for all slots) because we only need to know where to place the instances
in the current slot. The value of $D_{\boldsymbol{\pi}(t-1,t)}^{t_{0}}(t)$
in (\ref{eq:objFuncFrame}) depends on the configuration in slot $t_{0}-1$,
i.e. $\boldsymbol{\pi}(t_{0}-1)$, according to (\ref{eq:costDef}).
When $t_{0}=1$, $\boldsymbol{\pi}(t_{0}-1)$ can be regarded as an
arbitrary value, because the migration cost $W(t,\cdot,\cdot)=0$
for $t=1$. 

Intuitively, at the beginning of slot $t_{0}$, (\ref{eq:objFuncFrame})
finds the optimal configuration that minimizes the predicted cost
over the next $T$ slots, given the locations of instances in slot
$t_{0}-1$. We focus on solving (\ref{eq:objFuncFrame}) next.

\begin{algorithm} 
\caption{Procedure of offline service placement} 
\label{alg:offlineHighLevel} 
{\footnotesize
\begin{algorithmic}[1] 
\STATE Initialize $t_{0}=1$
\LOOP
\STATE At the beginning of slot $t_0$, find the solution to (\ref{eq:objFuncFrame})
\STATE Apply placements $\boldsymbol{\pi}_{\textrm{off}}(t_{0},...,t_{0}+T-1)$ in timeslots $t_{0},...,t_{0}+T-1$
\STATE $t_{0}\leftarrow t_{0}+T$
\ENDLOOP
\end{algorithmic} 
}
\end{algorithm}

\subsection{Equivalence to Shortest-Path Problem}

The problem in (\ref{eq:objFuncFrame}) is equivalent to a shortest-path
problem with $D_{\boldsymbol{\pi}(t-1,t)}^{t_{0}}(t)$ as weights,
as shown in Fig.~\ref{fig:shortestPathIllustration}. Each edge represents
one possible combination of configurations in adjacent timeslots,
and the weight on each edge is the predicted cost for such configurations.
The configuration in slot $t_{0}-1$ is always given, and the number
of possible configurations in subsequent timeslots is at most $K^{M}$,
where $M$ is defined as in Section \ref{sub:Configurations} for
the current window $\{t_{0},...,t_{0}+T-1\}$, and we note that depending
on whether the instance is running in the system or not, the number
of possible configurations in a slot is either $K$ or one (for configuration
$0$). Node B is a dummy node to ensure that we find a single shortest
path, and the edges connecting node B have zero weights. It is obvious
that the optimal solution to (\ref{eq:objFuncFrame}) can be found
by taking the shortest (minimum-weighted) path from node $\boldsymbol{\pi}(t_{0}-1)$
to node B in Fig. \ref{fig:shortestPathIllustration}; the nodes
that the shortest path traverses correspond to the optimal solution
$\boldsymbol{\pi}_{\textrm{off}}(t_{0},...,t_{0}+T-1)$ for (\ref{eq:objFuncFrame}).

\begin{figure}
\centering \includegraphics[width=0.9\columnwidth]{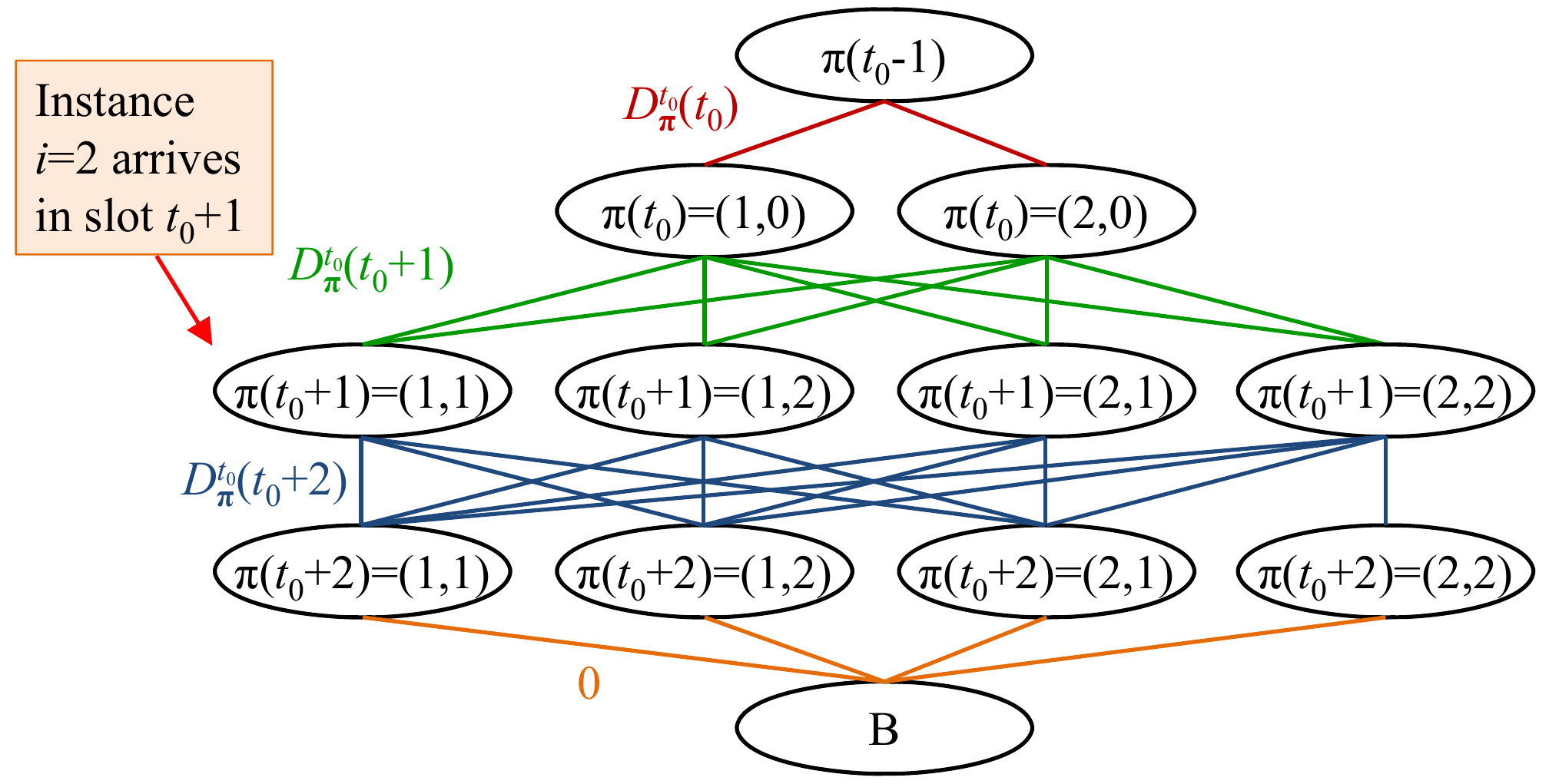}

\protect\caption{Shortest-path formulation with $K=2$, $M=2$, and $T=3$. Instance
$i=1$ is running in all slots, instance $i=2$ arrives at the beginning
of slot $t_{0}+1$ and is running in slots $t_{0}+1$ and $t_{0}+2$.}
\label{fig:shortestPathIllustration}
\end{figure}

\subsection{Algorithm\label{sub:AlgorithmOffline}}

We can solve the abovementioned shortest-path problem by means of
dynamic programming \cite{powell2007approximate}. The algorithm is
shown in Algorithm~\ref{alg:shortestPath}, where we use $U_{p}(t,\mathbf{m})$
and $W_{p}(t,\mathbf{n},\mathbf{m})$ to respectively denote the predicted
local and migration costs when $\boldsymbol{\pi}(t)=\mathbf{m}$
and $\boldsymbol{\pi}(t-1)=\mathbf{n}$. 

In the algorithm, Lines \ref{alg:shortestPath:iteration:start}--\ref{alg:shortestPath:iteration:end}
iteratively find the shortest path (minimum objective function) for
each timeslot. The iteration starts from the second level of the virtual graph in Fig.~\ref{fig:shortestPathIllustration}, which contains nodes with $\pi(t_0)$. It iterates through all the subsequent levels that respectively contain nodes with $\pi(t_0+1)$, $\pi(t_0+2)$, etc., excluding the last level with node B.
In each iteration, the optimal solution for every possible
(single-slot) configuration $\mathbf{m}$ is found by solving the
Bellman's equation of the problem (Line \ref{alg:shortestPath:iteration:bellman}). Essentially, the Bellman's equation finds the shortest path between the top node $\pi(t_0 -1)$ and the current node under consideration (e.g., node $\pi(t_0 +1) = (1,1)$ in Fig. \ref{fig:shortestPathIllustration}), by considering all the nodes in the previous level (nodes $\pi(t_0) = (1,0)$ and $\pi(t_0) = (2,0)$ in Fig. \ref{fig:shortestPathIllustration}). The sum weight on the shortest path between each node in the previous level and the top node $\pi(t_0 -1)$ is stored in $\nu_\mathbf{m}$, and the corresponding nodes that this shortest path traverses through is stored in $\boldsymbol{\xi}_\mathbf{m}$. Based on $\nu_\mathbf{m}$ and $\boldsymbol{\xi}_\mathbf{m}$, Line  \ref{alg:shortestPath:iteration:bellman} finds the shortest paths for all nodes in the current level, which can again be used for finding the shortest paths for all nodes in the next level (in the next iteration).

After iterating through all levels, the algorithm has found the shortest paths between top node $\pi(t_0 -1)$ and all nodes in the last level with $\pi(t_0 -T-1)$. Now, Lines \ref{alg:shortestPath:finalOptimal1}
and \ref{alg:shortestPath:finalOptimal2} find the minimum of all these shortest paths, giving the optimal configuration. It is obvious that output
of this algorithm satisfies the Bellman's principle of optimality,
so the result is the shortest path and hence the optimal solution
to (\ref{eq:objFuncFrame}). 

\emph{Complexity:} When the vectors $\boldsymbol{\pi}_{\mathbf{m}}$
and $\boldsymbol{\xi}_{\mathbf{m}}$ are stored as linked-lists, Algorithm~\ref{alg:shortestPath}
has time-complexity $O\left(K^{2M}T\right)$, because the
minimization in Line~\ref{alg:shortestPath:iteration:bellman} requires
enumerating at most $K^{M}$ possible configurations, and
there can be at most $K^{M}T$ possible combinations of values of
$t$ and $\mathbf{m}$.

\begin{algorithm} 
\caption{Algorithm for solving (\ref{eq:objFuncFrame})} 
\label{alg:shortestPath} 
{\footnotesize
\begin{algorithmic}[1] 
\STATE Define variables $\mathbf{m}$ and $\mathbf{n}$ to represent configurations respectively in the current and previous iteration (level of graph)
\STATE Define vectors $\boldsymbol{\pi}_\mathbf{m}$ and $\boldsymbol{\xi}_\mathbf{m}$ for all $\mathbf{m},\mathbf{n}$, where $\boldsymbol{\pi}_\mathbf{m}$ (correspondingly, $\boldsymbol{\xi}_\mathbf{m}$) records the optimal configuration given that the configuration at the current (correspondingly, previous) timeslot of iteration is $\mathbf{m}$
\STATE Define variables $\mu_\mathbf{m}$ and $\nu_\mathbf{m}$ for all $\mathbf{m}$ to record the sum cost values from slot $t_0$ respectively to the current and previous slot of iteration, given that the configuration is $\mathbf{m}$ in the current or previous slot

\STATE Initialize $\mu_\mathbf{m} \leftarrow 0$ and $\boldsymbol{\pi}_\mathbf{m} \leftarrow \emptyset$ for all $\mathbf{m}$

\FOR {$t=t_0,...,t_{0}+T-1$}  \label{alg:shortestPath:iteration:start}
  \FOR {all $\mathbf{m}$}  \label{alg:shortestPath:forall1:start}
    \STATE $\nu_\mathbf{m} \leftarrow \mu_\mathbf{m}$
    \STATE $\boldsymbol{\xi}_\mathbf{m} \leftarrow \boldsymbol{\pi}_\mathbf{m}$
  \ENDFOR \label{alg:shortestPath:forall1:end}
  
  \FOR {all $\mathbf{m}$} \label{alg:shortestPath:forall2:start}
    \STATE $\mathbf{n}^* \leftarrow \arg\min_\mathbf{n} \left\{ \nu_{\mathbf{n}}+U_p(t,\mathbf{m})+W_p(t,\mathbf{n},\mathbf{m}) \right\}$   \label{alg:shortestPath:iteration:bellman}
    \STATE $\boldsymbol{\pi}_\mathbf{m}(t_0,...,t-1) \leftarrow \boldsymbol{\xi}_{\mathbf{n}^*}(t_0,...,t-1)$  \label{alg:shortestPath:iteration:placementUpdt}
    \STATE ${\pi}_\mathbf{m}(t) \leftarrow \mathbf{m}$
    \STATE $\mu_\mathbf{m} \leftarrow \nu_{\mathbf{n}^*}+U_p(t,\mathbf{m})+W_p(t,\mathbf{n}^*,\mathbf{m})$
  \ENDFOR \label{alg:shortestPath:forall2:end}
\ENDFOR   \label{alg:shortestPath:iteration:end}

\STATE $\mathbf{m}^* \leftarrow \arg\min_\mathbf{m} \mu_{\mathbf{m}}$ \label{alg:shortestPath:finalOptimal1}
\STATE $\boldsymbol{\pi}_\textrm{off}(t_0,...,t_{0}+T-1) \leftarrow \boldsymbol{\pi}_{\mathbf{m}^*}(t_0,...,t_{0}+T-1)$ \label{alg:shortestPath:finalOptimal2}

\RETURN $\boldsymbol{\pi}_\textrm{off}(t_0,...,t_{0}+T-1)$  
\end{algorithmic} 
}
\end{algorithm}

\section{Complexity Reduction and Online Service Placement \label{sec:onlinePlacement}}

The complexity of Algorithm \ref{alg:shortestPath} is exponential
in the number of instances $M$, so it is desirable to reduce the
complexity. In this section, we propose a method that can find an
approximate solution to (\ref{eq:objFuncFrame}) and, at the same
time, handle online instance arrivals and departures that are not
known beforehand. We will also show that (\ref{eq:objFuncFrame})
is NP-hard when $M$ is non-constant, which justifies the need to
solve (\ref{eq:objFuncFrame}) approximately in an efficient manner.

\subsection{Procedure \label{sub:ProcedureOnline}}

In the online case, we modify the procedure given in Algorithm~\ref{alg:offlineHighLevel}
so that instances are placed one-by-one, where each  placement greedily
minimizes the objective function given in (\ref{eq:objFuncFrame}),
while the configurations of previously placed instances remain unchanged. 

We assume that each service instance $i$ has a maximum lifetime $T_{\textrm{life}}(i)$,
denoting the maximum number of remaining timeslots (including the
current slot) that the instance remains in the system. The value of
$T_{\textrm{life}}(i)$ may be infinity for instances that can potentially
stay in the system for an arbitrary amount of time. The actual time
that the instance stays in the system may be shorter than $T_{\textrm{life}}(i)$,
but it cannot be longer than $T_{\textrm{life}}(i)$. When an instance
leaves the system before its maximum lifetime has elapsed, we say
that such a service instance departure is \emph{unpredictable}.

We use $\boldsymbol{\pi}_{\textrm{on}}$ to denote the configuration
$\boldsymbol{\pi}$ computed by online placement. The configuration
$\boldsymbol{\pi}_{\textrm{on}}$ is updated every time when an instance
arrives or unpredictably departs. At the beginning of the window (before
any instance has arrived), it is initiated as an all-zero matrix.

For a specific look-ahead window $\{t_{0},...,t_{0}+T-1\}$, when
service instance $i$ arrives in slot $t\in\{t_{0},...,t_{0}+T-1\}$,
we assume that this instance stays in the system until slot $t_{e}=\min\left\{ t+T_{\textrm{life}}(i)-1;t_{0}+T-1\right\} $,
and accordingly update the configuration by 
\begin{align}
 & \boldsymbol{\pi}_{\textrm{on}}(t,...,t_{e})=\arg\min_{\boldsymbol{\pi}(t_{a},...,t_{e})}\sum_{t=t_{a}}^{t_{e}}D_{\boldsymbol{\pi}(t-1,t)}^{t_{0}}(t)\label{eq:objFuncFrameOnline}\\
 & \textrm{s.t.}\:\boldsymbol{\pi}(t,...,t_{e})=\boldsymbol{\pi}_{\textrm{on}}(t,...,t_{e})\textrm{ except for column }i\nonumber 
\end{align}
Note that only the configuration of instance $i$ (which is
assumed to be stored in the $i$th column of $\boldsymbol{\pi}$) is found and updated
in (\ref{eq:objFuncFrameOnline}), the configurations of all other
instances $i'\neq i$ remain unchanged. The solution to (\ref{eq:objFuncFrameOnline})
can still be found with Algorithm \ref{alg:shortestPath}. The only
difference is that vectors $\mathbf{m}$ and $\mathbf{n}$ now become
scalar values within $\{1,...,K\}$, because we only consider the
configuration of a single instance $i$. The complexity in this case
becomes $O(K^{2}T)$. At the beginning of the window, all the instances
that have not departed after slot $t_{0}-1$ are seen as arrivals
in slot $t_{0}$, because we independently consider the placements
in each window of size $T$. When multiple instances arrive simultaneously,
an arbitrary arrival sequence is assigned to them; the instances
are still placed one-by-one by greedily minimizing (\ref{eq:objFuncFrameOnline}). 

When an instance $i$ unpredictably departs at the end of slot $t\in\{t_{0},...,t_{0}+T-1\}$,
we update $\boldsymbol{\pi}_{\textrm{on}}$ such that the $i$th column
of $\boldsymbol{\pi}_{\textrm{on}}(t+1,...,t_{0}+T-1)$ is set to
zero.

The online procedure described above is shown in Algorithm \ref{alg:onlineHighLevel}.
Recall that $\boldsymbol{\pi}_{\textrm{on}}(t,...,t_{e})$ and $\boldsymbol{\pi}_{\textrm{on}}(t+1,...,t_{0}+T-1)$
are both part of a larger configuration matrix $\boldsymbol{\pi}_{\textrm{on}}(t_{0},...,t_{0}+T-1)$
(see Section \ref{sub:Configurations}).

\emph{Complexity:} When placing a total of $M$ instances, for a specific
look-ahead window with size $T$, we can find the configurations of
these $M$ instances with complexity $O(K^{2}TM)$, because (\ref{eq:objFuncFrameOnline})
is solved for $M$ times, each with complexity $O(K^{2}T)$. 

\begin{algorithm} 
\caption{Procedure of online service placement} 
\label{alg:onlineHighLevel} 
{\footnotesize
\begin{algorithmic}[1] 
\STATE Initialize $t_{0}=1$
\LOOP
\STATE Initialize $\boldsymbol{\pi}_{\textrm{on}}(t_0,...,t_0+T-1)$ as an all-zero matrix
\FOR {each timeslot $t=t_0,...,t_0+T-1$}

\FOR {each instance $i$ arriving at the beginning of slot~$t$}
\STATE $t_{e}\leftarrow \min\left\{ t+T_{\textrm{life}}(i)-1;t_{0}+T-1\right\} $
\STATE Update $\boldsymbol{\pi}_{\textrm{on}}(t,...,t_{e})$ with the result from (\ref{eq:objFuncFrameOnline})
\STATE Apply configurations specified in the $i$th column of $\boldsymbol{\pi}_{\textrm{on}}(t,...,t_{e})$ for service instance $i$ in timeslots $t,...,t_{e}$ until instance $i$ departs
\ENDFOR

\FOR {each instance $i$ departing at the end of slot~$t$}
\STATE \label{algOnline:lineSetRemainingToZero} Set the $i$th column of $\boldsymbol{\pi}_{\textrm{on}}(t+1,...,t_0+T-1)$ to zero
\ENDFOR

\ENDFOR
\STATE $t_{0}\leftarrow t_{0}+T$
\ENDLOOP
\end{algorithmic} 
}
\end{algorithm} 

\emph{Remark:} It is important to note that in the above procedure,
the configuration $\boldsymbol{\pi}_{\textrm{on}}$ (and thus the
cost value $D_{\boldsymbol{\pi}(t-1,t)}^{t_{0}}(t)$ for \emph{any}
$t\in\{t_{0},...,t_{0}+T-1\}$) may vary upon instance arrival or
departure. It follows that the $T$-slot sum cost $\sum_{t=t_{0}}^{t_{0}+T-1}D_{\boldsymbol{\pi}_{\textrm{on}}(t-1,t)}^{t_{0}}(t)$
may vary whenever an instance arrives or departs at an arbitrary slot
$t\in\{t_{0},...,t_{0}+T-1\}$, and the value of $\sum_{t=t_{0}}^{t_{0}+T-1}D_{\boldsymbol{\pi}_{\textrm{on}}(t-1,t)}^{t_{0}}(t)$
stands for the predicted sum cost (over the current window containing
$T$ slots) under the \emph{current} configuration, assuming that
no new instance arrives and no instance unpredictably departs in the
future. This variation in configuration and cost upon instance arrival/departure is frequently mentioned in the analysis
presented next.

\subsection{Performance Analysis \label{sub:PerformanceAnalysisOnlineAlg}}

It is clear that for a single look-ahead window, Algorithm~\ref{alg:onlineHighLevel}
has polynomial time-complexity while Algorithm \ref{alg:offlineHighLevel}
has exponential time-complexity. In this subsection, we show the NP-hardness
of the offline service placement problem, and discuss the optimality
gap between the online algorithm and the optimal offline placement.
Note that we only focus on a single look-ahead window in this subsection.
The interplay of multiple look-ahead windows and the impact of the
window size will be considered in Section~\ref{sec:findLookAhead}.

\subsubsection{Definitions \label{sub:PerfAnalysisDefinitions}}

For simplicity, we analyze the performance for a slightly restricted
(but still general) class of cost functions. We introduce some additional
definitions next (see Appendix~\ref{app:SummaryNotations} for a
summary of notations). 

\textbf{Indexing of Instances:} Here, we assume
that the instance with lowest index in the current window $\{t_{0},...,t_{0}+T-1\}$
has index $i=1$, and the last instance that arrives before the \emph{current
time of interest} has index $i=M$, where the current time of interest
can be any time within the current window. With this definition, $M$
does \emph{not} need to be the largest index in window $\{t_{0},...,t_{0}+T-1\}$.
Instead, it can be the index of \emph{any} instance that arrives within
$\{t_{0},...,t_{0}+T-1\}$. The cost of placing up to (and including)
instance $M$ is considered, where some instances $i\leq M$ may have
already departed from the system.

\textbf{Possible Configuration Sequence:} When considering a window
of $T$ slots, we define the set of all possible configurations of
a single instance as a set of $T$-dimensional vectors $\Lambda\triangleq\{\left(\lambda_{1},...,\lambda_{T}\right):\lambda_{n}\in\{0,1,...,K\},\forall n\in\{1,...,T\},$
where $\lambda_{n}$ is non-zero for at most one block of consecutive
values of $n\}$. We also define a vector $\boldsymbol{\lambda}\in\Lambda$
to represent one \emph{possible configuration sequence} of a single
service instance across these $T$ consecutive slots. For any instance
$i$, the $i$th column of configuration matrix $\boldsymbol{\pi}(t_{0},...,t_{0}+T-1)$
is equal to one particular value of $\boldsymbol{\lambda}$. 

We also define a binary variable $x_{i\boldsymbol{\lambda}}$, where
$x_{i\boldsymbol{\lambda}}=1$ if instance $i$ is placed according
to configuration sequence $\boldsymbol{\lambda}$ across slots $\{t_{0},...,t_{0}+T-1\}$
(i.e., the $i$th column of $\boldsymbol{\pi}(t_{0},...,t_{0}+T-1)$
is equal to $\boldsymbol{\lambda}$), and $x_{i\boldsymbol{\lambda}}=0$
otherwise. We always have $\sum_{\boldsymbol{\lambda}\in\Lambda}x_{i\boldsymbol{\lambda}}=1$
for all $i\in\{1,...,M\}$. 

We note that the values of $x_{i\boldsymbol{\lambda}}$ may vary over
time due to arrivals and unpredictable departures of instances, which
can be seen from Algorithm \ref{alg:onlineHighLevel} and by noting
the relationship between $\boldsymbol{\lambda}$ and $\boldsymbol{\pi}$.
Before instance $i$ arrives, $x_{i\boldsymbol{\lambda}_{0}}=1$ for
$\boldsymbol{\lambda}_{0}=[0,...,0]$ which contains all zeros, and
$x_{i\boldsymbol{\lambda}}=0$ for $\boldsymbol{\lambda}\neq\boldsymbol{\lambda}_{0}$.
Upon arrival of instance $i$, we have $x_{i\boldsymbol{\lambda}_{0}}=0$
and $x_{i\boldsymbol{\lambda}_{1}}=1$ for a particular $\boldsymbol{\lambda}_{1}$.
When instance $i$ unpredictably departs at slot $t'$, its configuration
sequence switches from $\boldsymbol{\lambda}_{1}$ to an alternative
(but partly correlated) sequence $\boldsymbol{\lambda}'_{1}$ (i.e.,
$(\boldsymbol{\lambda}'_{1})_{t}=(\boldsymbol{\lambda}_{1})_{t}$
for $t\leq t'$ and $(\boldsymbol{\lambda}'_{1})_{t}=0$ for $t>t'$,
where $(\boldsymbol{\lambda})_{t}$ denotes the $t$th element of
$\boldsymbol{\lambda}$), according to Line \ref{algOnline:lineSetRemainingToZero}
in Algorithm \ref{alg:onlineHighLevel}, after which $x_{i\boldsymbol{\lambda}_{1}}=0$
and $x_{i\boldsymbol{\lambda}'_{1}}=1$.

\textbf{Resource Consumption: }We assume that the costs are related
to the resource consumption, and for the ease of presentation, we
consider two types of resource consumptions. The first type is associated
with serving user requests, i.e., data transmission and processing
when a cloud is running a service instance, which we refer to as the
\emph{local resource consumption}. The second type is associated with
migration, i.e., migrating an instance from one cloud to another
cloud, which we refer to as the \emph{migration resource consumption}. 

If we know that instance $i$ operates under configuration sequence
$\boldsymbol{\lambda}$, then we know whether instance $i$ is placed
on cloud $k$ in slot $t$, for any $k\in\{1,...,K\}$ and $t\in\{t_{0},...,t_{0}+T-1\}$.
We also know whether instance $i$ is migrated from cloud $k$ to
cloud $l$ ($l\in\{1,2,...,K\}$) between slots $t-1$ and $t$. We
use $a_{i\boldsymbol{\lambda}k}(t)\geq0$ to denote the local resource
consumption at cloud $k$ in slot $t$ when instance $i$ is operating
under $\boldsymbol{\lambda}$, where $a_{i\boldsymbol{\lambda}k}(t)=0$
if $(\boldsymbol{\lambda})_{t}\neq k$. We use $b_{i\boldsymbol{\lambda}kl}(t)\geq0$
to denote the\emph{ }migration resource consumption when instance
$i$ operating under $\boldsymbol{\lambda}$ is assigned to cloud
$k$ in slot $t-1$ and to cloud $l$ in slot $t$, where $b_{i\boldsymbol{\lambda}kl}(t)=0$
if $(\boldsymbol{\lambda})_{t-1}\neq k$ or $(\boldsymbol{\lambda})_{t}\neq l$,
and we note that the configuration in slot $t_{0}-1$ (before the
start of the current window) is assumed to be given and thus independent
of $\boldsymbol{\lambda}$. The values of $a_{i\boldsymbol{\lambda}k}(t)$
and $b_{i\boldsymbol{\lambda}kl}(t)$ are either service-specific
parameters that are known beforehand, or they can be found as part
of the cost prediction.

We denote the sum local resource consumption at cloud $k$ by $y_{k}(t)\triangleq\sum_{i=1}^{M}\sum_{\boldsymbol{\lambda}\in\Lambda}a_{i\boldsymbol{\lambda}k}(t)x_{i\boldsymbol{\lambda}}$,
and denote the sum migration resource consumption from cloud $k$
to cloud $l$ by $z_{kl}(t)\triangleq\sum_{i=1}^{M}\sum_{\boldsymbol{\lambda}\in\Lambda}b_{i\boldsymbol{\lambda}kl}(t)x_{i\boldsymbol{\lambda}}$.
We may omit the argument $t$ in the following discussion.

\emph{Remark:} The local and migration resource consumptions defined
above can be related to CPU and communication bandwidth occupation, etc.,
or the sum of them. We only consider these two types of resource consumption
for the ease of presentation. By applying the same theoretical framework,
the performance gap results (presented later) can be extended to incorporate
multiple types of resources and more sophisticated cost functions,
and similar results yield for the general case.

\textbf{Costs: }We refine the costs defined in Section \ref{sub:CostDef}
by considering the cost for each cloud or each pair of clouds. The
local cost at cloud $k$ in timeslot $t$ is denoted by $u_{k,t}\left(y_{k}(t)\right)$.
When an instance is initiated in slot $t$, the local cost in slot
$t$ also includes the cost of initial placement of the corresponding
instance. The migration cost\emph{ }from cloud $k$ to cloud $l$
between slots $t-1$ and $t$ is denoted by\emph{ }$w_{kl,t}\left(y_{k}(t-1),y_{l}(t),z_{kl}(t)\right)$.
Besides $z_{kl}(t)$, the migration cost is also related to $y_{k}(t-1)$
and $y_{l}(t)$, because additional processing may be needed for migration,
and the cost for such processing can be related to the current load
at clouds $k$ and $l$. The functions $u_{k,t}\left(y\right)$ and
$w_{kl,t}\left(y_{k},y_{l},z_{kl}\right)$ can be different for different
slots $t$ and different clouds $k$ and $l$, and they can depend
on many factors, such as network condition, background load of the
cloud, etc. Noting that any constant term added to the cost function
does not affect the optimal configuration, we set $u_{k,t}(0)=0$
and $w_{kl,t}(0,0,0)=0$. We also set $w_{kl,t}(\cdot,\cdot,0)=0$,
because there is no migration cost if we do not migrate. There is
also no migration cost at the start of the first timeslot, thus we
set $w_{kl,t}(\cdot,\cdot,\cdot)=0$ for $t=1$. With these definitions,
the aggregated costs $U(t,\boldsymbol{\pi}(t))$ and $W(t,\boldsymbol{\pi}(t-1),\boldsymbol{\pi}(t))$
can be explicitly expressed as
\begin{align}
U(t,\boldsymbol{\pi}(t)) & \triangleq\sum_{k=1}^{K}u_{k,t}\left(y_{k}(t)\right)\label{eq:localCostBasedOnPiDef}\\
W\!(t,\boldsymbol{\pi}(t\!-\!1),\!\boldsymbol{\pi}(t))\!\! & \triangleq\!\!\sum_{k=1}^{K}\!\sum_{l=1}^{K}\! w_{kl,t}\!\left(y_{k}(t\!-\!1),y_{l}(t),z_{kl}(t)\right)\label{eq:migCostBasedOnPiDef}
\end{align}

We then assume that the following assumption is satisfied for the
cost functions, which holds for a large class of practical cost functions,
such as those related to the delay performance or load balancing \cite{refViNEYard}.

\begin{assumption} \label{condition:costFunc} Both $u_{k,t}(y)$
and $w_{kl,t}(y_{k},y_{l},z_{kl})$ are \emph{convex non-decreasing}
functions of $y$ (or $y_{k},y_{l},z_{kl}$), satisfying: 
\begin{itemize}
\item $\frac{du_{k,t}}{dy}(0)>0$
\item $\frac{\partial w_{kl,t}}{\partial z_{kl}}\left(\cdot,\cdot,0\right)>0$
for $t\geq2$
\end{itemize}
for all $t$, $k$, and $l$ (unless stated otherwise), where $\frac{du_{k,t}}{dy}(0)$
denotes the derivative of $u_{k,t}$ with respect to (w.r.t.) $y$
evaluated at $y=0$, and $\frac{\partial w_{kl,t}}{\partial z_{kl}}\left(\cdot,\cdot,0\right)$
denotes the partial derivative of $w_{kl,t}$ w.r.t. $z_{kl}$ evaluated
at $z_{kl}=0$ and arbitrary $y_{k}$ and $y_{l}$. \end{assumption}

\textbf{Vector Notation:} To simplify the presentation, we use vectors
to denote a collection of variables across multiple clouds, slots,
or configuration sequences. For simplicity, we index each element
in the vector with multiple indexes that are related to the index
of the element, and use the general notion $\left(\mathbf{g}\right)_{h_{1}h_{2}}$
(or $\left(\mathbf{g}\right)_{h_{1}h_{2}h_{3}}$) to denote the $(h_{1},h_{2})$th
(or $(h_{1},h_{2},h_{3})$th) element in an arbitrary vector $\mathbf{g}$.
Because we know the range of each index, multiple indexes can be easily
mapped to a single index. We regard each vector as a \emph{single-indexed}
vector for the purpose of vector concatenation (i.e., joining two
vectors into one vector) and gradient computation later. 

We define vectors $\mathbf{y}$ (with $KT$ elements), $\mathbf{z}$
(with $K^{2}T$ elements), $\mathbf{x}$ (with $MK^{T}$ elements),
$\mathbf{a}_{i\boldsymbol{\lambda}}$ (with $KT$ elements), and $\mathbf{b}_{i\boldsymbol{\lambda}}$
(with $K^{2}T$ elements), for every value of $i\in\{1,2,...,M\}$
and $\boldsymbol{\lambda}\in\Lambda$. Different values of $i$ and
$\boldsymbol{\lambda}$ correspond to different vectors $\mathbf{a}_{i\boldsymbol{\lambda}}$
and $\mathbf{b}_{i\boldsymbol{\lambda}}$. The elements in these vectors
are defined as follows:
\[
\left(\mathbf{y}\right)_{kt}\triangleq y_{k}(t),\:\left(\mathbf{z}\right)_{klt}\triangleq z_{kl}(t),\:\left(\mathbf{x}\right)_{i\boldsymbol{\lambda}}\triangleq x_{i\boldsymbol{\lambda}},
\]
\[
\left(\mathbf{a}_{i\boldsymbol{\lambda}}\right)_{kt}\triangleq a_{i\boldsymbol{\lambda}k}(t),\:\left(\mathbf{b}_{i\boldsymbol{\lambda}}\right)_{klt}\triangleq b_{i\boldsymbol{\lambda}kl}(t)
\]
As discussed earlier in this section, $x_{i\boldsymbol{\lambda}}$
may unpredictably change over time due to arrivals and departures
of service instances. It follows that the vectors $\mathbf{x}$, $\mathbf{y}$,
and $\mathbf{z}$ may vary over time (recall that $\mathbf{y}$ and
$\mathbf{z}$ are dependent on $\mathbf{x}$ by definition). The vectors
$\mathbf{a}_{i\boldsymbol{\lambda}}$ and $\mathbf{b}_{i\boldsymbol{\lambda}}$
are constant. 

\textbf{Alternative Cost Expression:} Using the above definitions,
we can write the sum cost of all $T$ slots as follows 
\begin{align}
 & \widetilde{D}\left(\mathbf{x}\right)\triangleq\widetilde{D}\left(\mathbf{y},\mathbf{z}\right)\triangleq\sum_{t=t_{0}}^{t_{0}+T-1}\Bigg[\sum_{k=1}^{K}u_{k,t}\left(y_{k}(t)\right)\nonumber \\
 & \quad\quad\quad\quad\quad+\sum_{k=1}^{K}\sum_{l=1}^{K}w_{kl,t}\left(y_{k}(t-1),y_{l}(t),z_{kl}(t)\right)\Bigg]\label{eq:costDef_T_xi}
\end{align}
where the cost function $\widetilde{D}(\cdot)$ can be expressed either
in terms of $\mathbf{x}$ or in terms of $\left(\mathbf{y},\mathbf{z}\right)$.
The cost function defined in (\ref{eq:costDef_T_xi}) is equivalent
to $\sum_{t=t_{0}}^{t_{0}+T-1}D_{\boldsymbol{\pi}(t-1,t)}^{t_{0}}(t)$,
readers are also referred to the per-slot cost defined in (\ref{eq:costDef})
for comparison. The value of $\widetilde{D}\left(\mathbf{x}\right)$
or, equivalently, $\widetilde{D}\left(\mathbf{y},\mathbf{z}\right)$
may vary over time due to service arrivals and unpredictable service
instance departures as discussed above.

\subsubsection{Equivalent Problem Formulation }

With the above definitions, the offline service placement problem
in (\ref{eq:objFuncFrame}) can be equivalently formulated as the
following, where our goal is to find the optimal configuration for
all service instances $1,2,...,M$ (we consider the offline case here
where we know when each instance arrives and no instance will unpredictably
leave after they have been placed):
\begin{align}
\min_{\mathbf{x}}\quad & \widetilde{D}\left(\mathbf{x}\right)\label{eq:optWithVectors}\\
s.t.\quad & \sum_{\boldsymbol{\lambda}\in\Lambda_{i}}x_{i\boldsymbol{\lambda}}=1,\forall i\in\{1,2,...,M\}\nonumber \\
 & x_{i\boldsymbol{\lambda}}\in\{0,1\},\forall i\in\{1,2,...,M\},\boldsymbol{\lambda}\in\Lambda_{i}\nonumber 
\end{align}
where $\Lambda_{i}\subseteq\Lambda$ is a subset of feasible configuration
sequences for instance $i$, i.e., sequences that contain those vectors
whose elements are non-zero starting from the slot at which $i$ arrives
and ending at the slot at which $i$ departs, while all other elements
of the vectors are zero.

We now show that (\ref{eq:optWithVectors}), and thus (\ref{eq:objFuncFrame}),
is NP-hard even in the offline case, which further justifies the need
for an approximation algorithm for solving the problem.

\begin{proposition}\textbf{ \label{prop:NPHardness} (NP-Hardness)}
The problem in (\ref{eq:optWithVectors}) in the offline sense, and
thus (\ref{eq:objFuncFrame}), is NP-hard.\end{proposition}\begin{proof}
Problem (\ref{eq:optWithVectors}) can be reduced from the partition
problem, which is known to be NP-complete \cite[Corollary 15.28]{korte2002combinatorial}.
See Appendix \ref{sec:NPHardProof} for details. \end{proof}

An online version of problem (\ref{eq:optWithVectors}) can be constructed
by updating $\Lambda_{i}$ over time. When an arbitrary instance $i$
has not yet arrived, we define $\Lambda_{i}$ as the set containing
an all-zero vector. After instance $i$ arrives, we assume that it
will run in the system until $t_{e}$ (defined in Section \ref{sub:ProcedureOnline}),
and update $\Lambda_{i}$ to conform to the arrival and departure
times of instance $i$ (see above). After instance $i$ departs, $\Lambda_{i}$
can be further updated so that the configurations corresponding to
all remaining slots are zero.

\subsubsection{Performance Gap}

As discussed earlier, Algorithm \ref{alg:onlineHighLevel} solves
(\ref{eq:optWithVectors}) in a greedy manner, where each service
instance $i$ is placed to greedily minimize 
in (\ref{eq:optWithVectors}). In the following, we compare the result
from Algorithm \ref{alg:onlineHighLevel} with the true optimal result,
where the optimal result assumes offline placement. We use $\mathbf{x}$
and $\left(\mathbf{y},\mathbf{z}\right)$ to denote the result from
Algorithm \ref{alg:onlineHighLevel}, and use $\mathbf{x}^{*}$ and
$\left(\mathbf{y}^{*},\mathbf{z}^{*}\right)$ to denote the offline
optimal result to (\ref{eq:optWithVectors}).

\begin{lemma} \textbf{\label{lemma:convexityOfH} (Convexity of $\widetilde{D}(\cdot)$)}
When Assumption \ref{condition:costFunc} is satisfied, the cost function
$\widetilde{D}\left(\mathbf{x}\right)$ or, equivalently, $\widetilde{D}\left(\mathbf{y},\mathbf{z}\right)$
is a non-decreasing convex function w.r.t. $\mathbf{x}$, and it is
also a non-decreasing convex function w.r.t. $\mathbf{y}$ and $\mathbf{z}$.
\end{lemma}\begin{proof} According to Assumption \ref{condition:costFunc},
$u_{k,t}\left(y_{k}(t)\right)$ and $w_{kl,t}\left(y_{k}(t-1),y_{l}(t),z_{kl}(t)\right)$
are non-decreasing convex functions. Because $y_{k}(t)$ and $z_{kl}(t)$
are linear mappings of $x_{i\boldsymbol{\lambda}}$ with non-negative
weights for any $t$, $k$, and $l$, and also because the sum of
non-decreasing convex functions is still a non-decreasing convex function,
the lemma holds \cite[Section 3.2]{boyd2004convex}.\end{proof}

In the following, we use $\nabla_{\mathbf{x}}$ to denote the gradient
w.r.t. \emph{each element} in vector $\mathbf{x}$, i.e., the $(i,\boldsymbol{\lambda})$th
element of $\nabla_{\mathbf{x}}\widetilde{D}(\mathbf{x})$ is $\frac{\partial\widetilde{D}(\mathbf{x})}{\partial x_{i\boldsymbol{\lambda}}}$.
Similarly, we use $\nabla_{\mathbf{y,z}}$ to denote the gradient
w.r.t. each element in vector $(\mathbf{y,z})$, where $(\mathbf{y,z})$
is a vector that concatenates vectors $\mathbf{y}$ and $\mathbf{z}$. 

\begin{proposition}\textbf{ \label{prop:performanceGapResult} (Performance
Gap)} When Assumption \ref{condition:costFunc} is satisfied, we have
\begin{equation}
\widetilde{D}(\mathbf{x})\leq\widetilde{D}(\phi\psi\mathbf{x}^{*})\label{eq:performGap_x}
\end{equation}
or, equivalently, 
\begin{equation}
\widetilde{D}\left(\mathbf{y},\mathbf{z}\right)\leq\widetilde{D}\left(\phi\psi\mathbf{y}^{*},\phi\psi\mathbf{z}^{*}\right)\label{eq:performGap_yz}
\end{equation}
where $\phi$ and $\psi$ are constants satisfying
\begin{equation}
\phi\geq\frac{\nabla_{\mathbf{y,z}}\widetilde{D}\left(\mathbf{y}_{\textrm{max}}+\mathbf{a}_{i\boldsymbol{\lambda}},\mathbf{z}_{\textrm{max}}+\mathbf{b}_{i\boldsymbol{\lambda}}\right)\cdot\left(\mathbf{a}_{i\boldsymbol{\lambda}},\mathbf{b}_{i\boldsymbol{\lambda}}\right)}{\nabla_{\mathbf{y,z}}\widetilde{D}\left(\mathbf{y},\mathbf{z}\right)\cdot\left(\mathbf{a}_{i\boldsymbol{\lambda}},\mathbf{b}_{i\boldsymbol{\lambda}}\right)}\label{eq:alphaDef}
\end{equation}
\begin{equation}
\psi\geq\frac{\nabla_{\mathbf{x}}\widetilde{D}\left(\mathbf{x}\right)\cdot\mathbf{x}}{\widetilde{D}(\mathbf{x})}=\frac{\nabla_{\mathbf{y,z}}\widetilde{D}\left(\mathbf{y},\mathbf{z}\right)\cdot\left(\mathbf{y},\mathbf{z}\right)}{\widetilde{D}\left(\mathbf{y},\mathbf{z}\right)}\label{eq:betaDef}
\end{equation}
for any $i$ and $\boldsymbol{\lambda}\in\Lambda_{i}$, in which $\mathbf{y}_{\textrm{max}}$
and $\mathbf{z}_{\textrm{max}}$ respectively denote the maximum values
of $\mathbf{y}$ and $\mathbf{z}$ (the maximum is taken element-wise)
\emph{after any number of instance arrivals} within slots $\{t_{0},...,t_{0}+T-1\}$
until the current time of interest (at which time the latest arrived
instance has index $M$), $\left(\mathbf{a}_{i\boldsymbol{\lambda}},\mathbf{b}_{i\boldsymbol{\lambda}}\right)$
is a vector that concatenates $\mathbf{a}_{i\boldsymbol{\lambda}}$
and $\mathbf{b}_{i\boldsymbol{\lambda}}$, and ``$\cdot$'' denotes
the dot-product. \end{proposition}

\begin{proof} See Appendix \ref{sec:performanceGapProof}. \end{proof}

\emph{Remark:} We note that according to the definition of $M$ in
Section \ref{sub:PerfAnalysisDefinitions}, the bound given in Proposition~\ref{prop:performanceGapResult}
holds at any time of interest within slots $\{t_{0},...,t_{0}+T-1\}$,
i.e., for any number of instances that has arrived to the system,
where some of them may have already departed.

\subsubsection{Intuitive Explanation to the Constants $\phi$ and $\psi$}

The constants $\phi$ and $\psi$ in Proposition \ref{prop:performanceGapResult}
are related to ``how convex'' the cost function is. In other words,
they are related to how fast the cost of placing a single instance
changes under different amount of existing resource consumption. Figure
\ref{fig:PerformanceGapIllustration} shows an illustrative example,
where we only consider one cloud and one timeslot (i.e., $t=1$, $T=1$,
and $K=1$). In this case, setting $\phi=\frac{d\widetilde{D}}{dy}(y_{\textrm{max}}+a_{\textrm{max}})\big/\frac{d\widetilde{D}}{dy}(y)$
satisfies (\ref{eq:alphaDef}), where $a_{\textrm{max}}$ denotes
the maximum resource consumption of a single instance. Similarly,
setting $\psi=\frac{d\widetilde{D}}{dy}(y)\cdot y\big/\widetilde{D}(y)$
satisfies (\ref{eq:betaDef}). We can see that the values of $\phi$
and $\psi$ need to be larger when the cost function is more convex.
For the general case, there is a weighted sum in both the numerator
and denominator in (\ref{eq:alphaDef}) and (\ref{eq:betaDef}). However,
when we look at a single cloud (for the local cost) or a single pair
of clouds (for the migration cost) in a single timeslot, the above
intuition still applies.

So, why is the optimality gap larger when the cost functions are more
convex, i.e., have a larger second order derivative? We note that
in the greedy assignment procedure in Algorithm \ref{alg:onlineHighLevel},
we choose the configuration of each instance $i$ by minimizing the
cost under the system state at the time when instance $i$ arrives,
where the system state represents the local and migration resource
consumptions as specified by vectors $\mathbf{y}$ and $\mathbf{z}$.
When cost functions are more convex, for an alternative system state
$(\mathbf{y}',\mathbf{z}')$, it is more likely that the placement
of instance $i$ (which was determined at system state $(\mathbf{y},\mathbf{z})$)
becomes far from optimum. This is because if cost functions are more
convex, the cost increase of placing a new instance $i$ (assuming
the same configuration for $i$) varies more when $(\mathbf{y},\mathbf{z})$
changes. This intuition is confirmed by formal results described
next.

\begin{figure}
\centering \includegraphics[width=0.7\columnwidth]{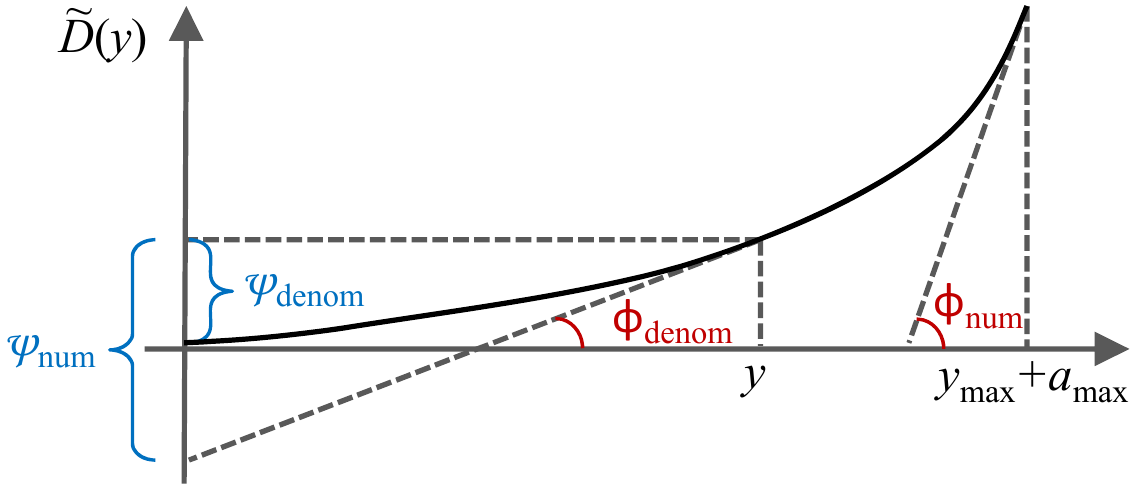}

\protect\caption{Illustration of the performance gap for $t=1$, $T=1$, and $K=1$,
where $a_{\textrm{max}}$ denotes the maximum resource consumption
of a single instance. In this example, (\ref{eq:alphaDef}) becomes
$\phi\geq\frac{\phi_{\textrm{num}}}{\phi_{\textrm{denom}}}$, and
(\ref{eq:betaDef}) becomes $\psi\geq\frac{\psi_{\textrm{num}}}{\psi_{\textrm{denom}}}$.}
\label{fig:PerformanceGapIllustration}
\end{figure}

\subsubsection{Linear Cost Functions \label{sub:LinearCostFunctions}}

Consider linear cost functions in the form of
\begin{align}
u_{k,t}(y) & =\gamma_{k,t}y\label{eq:linearLocalCost}\\
w_{kl,t}\left(y_{k},y_{l},z_{kl}\right) & =\kappa_{kl,t}^{(1)}y_{k}+\kappa_{kl,t}^{(2)}y_{l}+\kappa_{kl,t}^{(3)}z_{kl}\label{eq:linearMigCost}
\end{align}
where the constants $\gamma_{k,t},\kappa_{kl,t}^{(3)}>0$ and $\kappa_{kl,t}^{(1)},\kappa_{kl,t}^{(2)}\geq0$.

\begin{proposition} \label{prop:linearCostCompetitive} When the
cost functions are defined as in (\ref{eq:linearLocalCost}) and (\ref{eq:linearMigCost}),
Algorithm \ref{alg:onlineHighLevel} provides the optimal solution.
\end{proposition} \begin{proof}We have
\[
\nabla_{\mathbf{y,z}}\widetilde{D}\left(\mathbf{y}_{\textrm{max}}+\mathbf{a}_{i\boldsymbol{\lambda}},\mathbf{z}_{\textrm{max}}+\mathbf{b}_{i\boldsymbol{\lambda}}\right)=\nabla_{\mathbf{y,z}}\widetilde{D}\left(\mathbf{y},\mathbf{z}\right)
\]
\[
\nabla_{\mathbf{y,z}}\widetilde{D}\left(\mathbf{y},\mathbf{z}\right)\cdot\left(\mathbf{y},\mathbf{z}\right)=\widetilde{D}\left(\mathbf{y},\mathbf{z}\right)
\]
because the gradient in this case is a constant. Hence, choosing $\phi=\psi=1$
satisfies (\ref{eq:alphaDef}) and (\ref{eq:betaDef}), yielding $\widetilde{D}(\mathbf{x})\leq\widetilde{D}(\mathbf{x}^{*})$
which means that the solution from Algorithm \ref{alg:onlineHighLevel}
is not worse than the optimal solution. \end{proof}

This implies that the greedy service placement is optimal for linear
cost functions, which is intuitive because the previous placements
have no impact on the cost of later placements when the cost function
is linear.

\subsubsection{Polynomial Cost Functions \label{sub:PolynomialCostFunctions}}

Consider polynomial cost functions in the form of
\begin{align}
u_{k,t}(y) & =\sum_{\rho}\gamma_{k,t}^{(\rho)}y^{\rho}\label{eq:polyLocalCost}\\
w_{kl,t}\left(y_{k},y_{l},z_{kl}\right) & =\sum_{\rho_{1}}\sum_{\rho_{2}}\sum_{\rho_{3}}\kappa_{kl,t}^{(\rho_{1},\rho_{2},\rho_{3})}y_{k}^{\rho_{1}}y_{l}^{\rho_{2}}z_{kl}^{\rho_{3}}\label{eq:polyMigCost}
\end{align}
where $\rho,\rho_{1},\rho_{2},\rho_{3}$ are integers satisfying $\rho\geq1$,
$\rho_{1}+\rho_{2}+\rho_{3}\geq1$ and the constants $\gamma_{k,t}^{(\rho)}\geq0$,
$\kappa_{kl,t}^{(\rho_{1},\rho_{2},\rho_{3})}\geq0$. 

We first introduce the following assumption which can be satisfied
in most practical systems with an upper bound on resource consumptions
and departure rates.

\begin{assumption} \label{condition:boundedArrDept} The following
is satisfied:
\begin{itemize}
\item For all $i,\boldsymbol{\lambda},k,l,t$, there exists a constants
$a_{\textrm{max}}$ and $b_{\textrm{max}}$, such that 
$a_{i\boldsymbol{\lambda}k}(t)\leq a_{\textrm{max}}$
and
$b_{i\boldsymbol{\lambda}kl}(t)\leq b_{\textrm{max}}$

\item The number of instances that unpredictably leave the system in each
slot is upper bounded by a constant $B_{d}$. 
\end{itemize}
\end{assumption}

\begin{proposition} \label{prop:polyCostCompetitive} Assume that
the cost functions are defined as in (\ref{eq:polyLocalCost}) and
(\ref{eq:polyMigCost}) while satisfying Assumption \ref{condition:costFunc},
and that Assumption \ref{condition:boundedArrDept} is satisfied.

Let $\Omega$ denote the maximum value of $\rho$ such that $\gamma_{k,t}^{(\rho)}>0$
or $\kappa_{kl,t}^{(\rho_{1},\rho_{2},\rho_{3})}>0$, subject to $\rho_{1}+\rho_{2}+\rho_{3}=\rho$.
The value of $\Omega$ represents the highest order of the polynomial
cost functions. 

Define $\Gamma(\mathcal{I}(M))\triangleq {\widetilde{D}(\mathbf{x}_{\mathcal{I}(M)})} \Big/ {\widetilde{D}(\mathbf{x}_{\mathcal{I}(M)}^{*})}$,
where $\mathcal{I}(M)$ is a problem input%
\footnote{A particular problem input specifies the time each instance arrives/departs as
well as the values of $\mathbf{a}_{i\boldsymbol{\lambda}}$ and $\mathbf{b}_{i\boldsymbol{\lambda}}$
for each $i,\boldsymbol{\lambda}$.%
} containing $M$ instances, and $\mathbf{x}_{\mathcal{I}(M)}$ and
$\mathbf{x}_{\mathcal{I}(M)}^{*}$ are respectively the online and
offline (optimal) results for input $\mathcal{I}(M)$. We say that
Algorithm~\ref{alg:onlineHighLevel} is $c$-competitive in placing
$M$ instances if $\Gamma\triangleq\max_{\mathcal{I}(M)}\Gamma(\mathcal{I}(M))\leq c$
for a given $M$. We have:
\begin{itemize}
\item Algorithm \ref{alg:onlineHighLevel} is $O(1)$-competitive. 
\item In particular, for any $\delta>0$, there exists a sufficiently large
$M$, such that Algorithm \ref{alg:onlineHighLevel} is $\left(\Omega^{\Omega}+\delta\right)$-competitive. 
\end{itemize}
\end{proposition}

\begin{proof} See Appendix \ref{sec:polyCostCompetProof}. \end{proof}

Proposition \ref{prop:polyCostCompetitive} states that the competitive
ratio does not indefinitely increase with increasing number of instances
(specified by $M$). Instead, it approaches a constant value when
$M$ becomes large.

When the cost functions are linear as in (\ref{eq:linearLocalCost})
and (\ref{eq:linearMigCost}), we have $\Omega=1$. In this case,
Proposition \ref{prop:polyCostCompetitive} gives a competitive ratio
upper bound of $1+\delta$ (for sufficiently large $M$) where $\delta>0$
can be arbitrarily small, while Proposition \ref{prop:linearCostCompetitive}
shows that Algorithm \ref{alg:onlineHighLevel} is optimal. This means
that the competitive ratio upper bound given in Proposition \ref{prop:polyCostCompetitive}
is \emph{asymptotically tight} as $M$ goes to infinity.

\subsubsection{Linear Cost at Backend Cloud \label{sub:onlinePerformanceOtherCostFunc}}

Algorithm \ref{alg:onlineHighLevel} is also $O(1)$-competitive for
some more general forms of cost functions. For example, consider a
simple case where there is no migration resource consumption, i.e.
$b_{i\boldsymbol{\lambda}kl}(t)=0$ for all $i,\boldsymbol{\lambda},k,l$.
Define $u_{k_{0},t}(y)=\gamma y$ for some cloud $k_{0}$ and all
$t$, where $\gamma>0$ is a constant. For all other clouds $k\neq k_{0}$,
define $u_{k,t}(y)$ as a general cost function while satisfying Assumption~\ref{condition:costFunc}
and some additional mild assumptions presented below. Assume that
there exists a constant $a_{\textrm{max}}$ such that $a_{i\boldsymbol{\lambda}k}(t)\leq a_{\textrm{max}}$
for all $i,\boldsymbol{\lambda},k,t$. 

Because $u_{k,t}(y)$ is convex non-decreasing and Algorithm~\ref{alg:onlineHighLevel}
operates in a greedy manner, if $\frac{du_{k,t}}{dy}(y)>\gamma$,
no new instance will be placed on cloud $k$, as it incurs higher
cost than placing it on $k_{0}$. As a result, the maximum value of
$y_{k}(t)$ is bounded, let us denote this upper bound by $y_{k}^{\textrm{max}}(t)$.
We note that $y_{k}^{\textrm{max}}(t)$ is \emph{only dependent on
the cost function definition} and is independent of the number of
arrived instances.

Assume $u_{k,t}(y_{k}^{\textrm{max}}(t))<\infty$ and $\frac{du_{k,t}}{dy}(y_{k}^{\textrm{max}}(t)+a_{\textrm{max}})<\infty$
for all $k\neq k_{0}$ and $t$. When ignoring the cost at cloud $k_{0}$,
the ratio $\Gamma(\mathcal{I}(M))$ does not indefinitely grow with
incoming instances, because among all $y_{k}(t)\in[0,y_{k}^{\textrm{max}}(t)]$
for all $t$ and $k\neq k_{0}$, we can find $\phi$ and $\psi$ that
satisfy (\ref{eq:alphaDef}) and (\ref{eq:betaDef}), we can also
find the competitive ratio $\Gamma\triangleq\max_{\mathcal{I}(M)}\Gamma(\mathcal{I}(M))$.
The resulting $\Gamma$ is only dependent on the cost function definition,
hence it does not keep increasing with $M$. Taking into account the
cost at cloud $k_{0}$, the above result still applies, because the
cost at $k_{0}$ is linear in $y_{k_{0}}(t)$, so that in either of
(\ref{eq:alphaDef}), (\ref{eq:betaDef}), or in the expression of
$\Gamma(\mathcal{I}(M))$, the existence of this linear cost only
adds a same quantity (which might be different in different expressions
though) to both the numerator and denominator, which does not increase
$\Gamma$ (because $\Gamma \geq 1$).

The cloud $k_{0}$ can be considered as the backend cloud, which
usually has abundant resources thus its cost-per-unit-resource often
remains unchanged. This example can be generalized to cases with non-zero
migration resource consumption, and we will illustrate such an application
in the simulations in Section \ref{sec:simulation}.

\section{Optimal Look-Ahead Window Size}

\label{sec:findLookAhead}

In this section, we study how to find the optimal window size $T$
to look-ahead. When there are no errors in the cost prediction, setting
$T$ as large as possible can potentially bring the best long-term
performance. However, the problem becomes more complicated when we
consider the prediction error, because the farther ahead we look into
the future, the less accurate the prediction becomes. When $T$ is
large, the predicted cost value may be far away from the actual cost,
which can cause the configuration obtained from predicted costs $D_{\boldsymbol{\pi}}^{t_{0}}(t)$
with size-$T$ windows (denoted by $\boldsymbol{\pi}_{p}$) deviate
significantly from the true optimal configuration $\boldsymbol{\pi}^{*}$
obtained from actual costs $A_{\boldsymbol{\pi}}(t)$. Note that $\boldsymbol{\pi}^{*}$
is obtained from actual costs $A_{\boldsymbol{\pi}}(t)$, which is
different from $\mathbf{x}^{*}$ and $\left(\mathbf{y}^{*},\mathbf{z}^{*}\right)$
which are obtained from predicted costs $D_{\boldsymbol{\pi}}^{t_{0}}(t)$
as defined in Section \ref{sub:PerformanceAnalysisOnlineAlg}. Also
note that $\boldsymbol{\pi}_{p}$ and $\boldsymbol{\pi}^{*}$ specify
the configurations for an arbitrarily large number of timeslots, as
in (\ref{eq:objFunc}). Conversely, when $T$ is small, the solution
may not perform well in the long-term, because the look-ahead window
is small and the long-term effect of migration is not considered.
We have to find the optimal value of $T$ which minimizes
both the impact of prediction error and the impact of truncating the
look-ahead time-span.

We assume that there exists a constant $\sigma$ satisfying 
\begin{equation}
\max_{\boldsymbol{\pi}(t-1,t)}W_{a}(t,\boldsymbol{\pi}(t-1),\boldsymbol{\pi}(t)) \leq \sigma
\label{eq:sigmaDefForMigCostBound}
\end{equation}
 for any $t$, to represent the maximum value of the actual migration
cost in any slot, where $W_{a}(t,\boldsymbol{\pi}(t-1),\boldsymbol{\pi}(t))$
denotes the actual migration cost. The value of $\sigma$ is system-specific
and is related to the cost definition. 

To help with our analysis below, we define the sum-error starting
from slot $t_{0}$ up to slot $t_{0}+T-1$ as 
\begin{equation}
F(T)\triangleq\sum_{t=t_{0}}^{t_{0}+T-1}\epsilon(t-t_{0})
\end{equation}
 Because $\epsilon(t-t_{0})\geq0$ and $\epsilon(t-t_{0})$ is non-decreasing
with $t$, it is obvious that $F(T+2)-F(T+1)\geq F(T+1)-F(T)$. Hence,
$F(T)$ is a convex non-decreasing function for $T\geq0$, where we
define $F(0)=0$.

\subsection{Upper Bound on Cost Difference}

In the following, we focus on the objective function given in (\ref{eq:objFunc}),
and study how worse the configuration $\boldsymbol{\pi}_{p}$ can
perform, compared to the optimal configuration $\boldsymbol{\pi}^{*}$. 

\begin{proposition} \label{prop:costDifferenceBound} For look-ahead
window size $T$, suppose that we can solve (\ref{eq:objFuncFrame})
with competitive ratio $\Gamma\geq1$, the upper bound on the cost
difference (while taking the competitive ratio $\Gamma$ into account)
from configurations $\boldsymbol{\pi}_{p}$ and $\boldsymbol{\pi}^{*}$
is given by 
\begin{equation}
\lim_{T_{\textrm{max}}\rightarrow\infty}\!\!\left(\!\frac{\sum_{t=1}^{T_{\textrm{max}}}A_{\boldsymbol{\pi}_{p}}(t)}{T_{\textrm{max}}}\!-\!\Gamma\frac{\sum_{t=1}^{T_{\textrm{max}}}A_{\boldsymbol{\pi}^{*}}(t)}{T_{\textrm{max}}}\!\right)\!\leq\!\frac{(\Gamma\!+\!1)F(T)\!+\!\sigma}{T}\label{eq:CostAccuracyBound}
\end{equation}
 \end{proposition}

\begin{proof} See Appendix \ref{sec:costDifferenceBoundProof}. \end{proof} 

We assume in the following that the competitive ratio $\Gamma$ is
independent of the choice of $T$, and regard it as a given parameter
in the problem of finding optimal $T$. This assumption is justified
for several cost functions where there exist a uniform bound on the
competitive ratio for arbitrarily many services (see Sections \ref{sub:LinearCostFunctions}--\ref{sub:onlinePerformanceOtherCostFunc}).
We define the \emph{optimal look-ahead window size} as the solution to the
following optimization problem:
\begin{align}
\min_{T}\quad & \frac{(\Gamma+1)F(T)+\sigma}{T}\label{eq:errorCostObj}\\
\textrm{s.t.}\quad & T\geq1\nonumber 
\end{align}

Considering the original objective in (\ref{eq:objFunc}), the problem
(\ref{eq:errorCostObj}) can be regarded as finding the optimal look-ahead
window size such that an upper bound of the objective function in
(\ref{eq:objFunc}) is minimized (according to Proposition \ref{prop:costDifferenceBound}).
The solution to (\ref{eq:errorCostObj}) is the optimal window size
to look-ahead so that (in the worst case) the cost is closest to the
cost of the optimal configuration $\boldsymbol{\pi}^{*}$.

\subsection{Characteristics of the Problem in (\ref{eq:errorCostObj}) }

\label{sub:charOptWindLen}

We now study the characteristics of (\ref{eq:errorCostObj}). To help
with the analysis, we interchangeably use variable $T$ to represent
either a discrete or a continuous variable. We define a continuous
convex function $G(T)$, where $T\geq1$ is a continuous variable.
The function $G(T)$ is defined in such a way that $G(T)=F(T)$ for
all the discrete values $T\in\left\{ 1,2,...\right\} $, i.e., $G(T)$
is a \emph{continuous time extension} of $F(T)$. Such a definition
is always possible by connecting the discrete points in $F(T)$. Note
that we do not assume the continuity of the derivatives of $G(T)$,
which means that $\frac{dG(T)}{dT}$ may be non-continuous and $\frac{d^{2}G(T)}{dT^{2}}$
may have $+\infty$ values. However, these do not affect our analysis
below. We will work with continuous values of $T$ in some parts and
will discretize it when appropriate.

We define a function $\theta(T)\triangleq\frac{(\Gamma+1)G(T)+\sigma}{T}$
to represent the objective function in (\ref{eq:errorCostObj}) after
replacing $F(T)$ with $G(T)$, where $T$ is regarded as a continuous
variable. We take the logarithm of $\theta(T)$, yielding
\begin{equation}
\ln\theta=\ln\left((\Gamma+1)G(T)+\sigma\right)-\ln T
\end{equation}
Taking the derivative of $\ln\theta$, we have
\begin{equation}
\frac{d\ln\theta}{dT}=\frac{(\Gamma+1)\frac{dG(T)}{dT}}{(\Gamma+1)G(T)+\sigma}-\frac{1}{T}\label{eq:diffLnErrorCost}
\end{equation}
We set (\ref{eq:diffLnErrorCost}) equal to zero, and rearrange the
equation, yielding
\begin{equation}
\Phi(T)\triangleq(\Gamma+1)T\frac{dG(T)}{dT}-(\Gamma+1)G(T)-\sigma=0\label{eq:diffTEqu0}
\end{equation}

We have the following proposition and its corollary, their proofs are given in Appendix \ref{sec:costOptTProof}.

\begin{proposition} \label{prop:optT} Let $T_{0}$ denote a solution
to (\ref{eq:diffTEqu0}), if the solution exists, then the optimal
look-ahead window size $T^{*}$ for problem (\ref{eq:errorCostObj})
is either $\left\lfloor T_{0}\right\rfloor $ or $\left\lceil T_{0}\right\rceil $,
where $\left\lfloor x\right\rfloor $ and $\left\lceil x\right\rceil $
respectively denote the floor (rounding down to integer) and ceiling
(rounding up to integer) of $x$ .\end{proposition}

\begin{corollary} \label{prop:optTcorr} For window sizes $T$ and
$T+1$, if $\theta(T)<\theta(T+1)$, then the optimal size $T^{*}\leq T$;
if $\theta(T)>\theta(T+1)$, then $T^{*}\geq T+1$; if $\theta(T)=\theta(T+1)$,
then $T^{*}=T$.\end{corollary}

\subsection{Finding the Optimal Solution \label{sub:predErrorOptWind}}

According to Proposition \ref{prop:optT}, we can solve (\ref{eq:diffTEqu0})
to find the optimal look-ahead window size. When $G(T)$ (and $F(T)$)
can be expressed in some specific analytical forms, the solution to
(\ref{eq:diffTEqu0}) can be found analytically. For example, consider
$G(T)=F(T)=\beta T^{\alpha}$, where $\beta>0$ and $\alpha>1$. In
this case, $T_{0}=\left(\frac{\sigma}{(\Gamma+1)\beta(\alpha-1)}\right)^{\frac{1}{\alpha}}$,
and $T^{*}=\arg\min_{T\in\left\{ \left\lfloor T_{0}\right\rfloor ,\left\lceil T_{0}\right\rceil \right\} }\theta\left(T\right)$.
One can also use such specific forms as an upper bound for a general
function.

When $G(T)$ (and $F(T)$) have more general forms, we can perform
a search on the optimal window size according to the properties discussed
in Section \ref{sub:charOptWindLen}. Because we do not know the convexity
of $\theta(T)$ or $\Phi(T)$, standard numerical methods for solving
(\ref{eq:errorCostObj}) or (\ref{eq:diffTEqu0}) may not be efficient.
However, from Corollary \ref{prop:optTcorr}, we know that the local
minimum of $\theta(T)$ is the global minimum, so we can
develop algorithms that use this property.

The optimal window size $T^{*}$ takes discrete values, so we
can perform a discrete search on $T\in\left\{ 1,2,...,T_{m}\right\} $,
where $T_{m}>1$ is a pre-specified upper limit on the search range.
We then compare $\theta(T)$ with $\theta(T+1)$ and determine the
optimal solution according to Corollary \ref{prop:optTcorr}. One
possible approach is to use binary search, as shown in Algorithm \ref{alg:optTBinSearch},
which has time-complexity of $O\left(\log T_{m}\right)$.

\emph{Remark:} The exact value of $\Gamma$ may be difficult
to find in practice, and (\ref{eq:CostAccuracyBound}) is an upper
bound which may have a gap from the actual value of the left hand-side
of (\ref{eq:CostAccuracyBound}). Therefore, in practice, we can regard
$\Gamma$ as a tuning parameter, which can be tuned so that the resulting
window size $T^{*}$ yields good performance. For a similar reason,
the parameter $\sigma$ can also be regarded as a tuning parameter
in practice.

\begin{algorithm} 
\caption{Binary search for finding optimal window size} 
\label{alg:optTBinSearch} 
{\footnotesize
\begin{algorithmic}[1] 

\STATE Initialize variables $T_- \leftarrow 1$ and $T_+ \leftarrow T_{m}$

\REPEAT 

\STATE $T \leftarrow \left\lfloor \left(T_- + T_+\right)/2 \right\rfloor $
\IF {$\theta(T) < \theta(T+1)$}
  \STATE $T_+ \leftarrow T$
\ELSIF {$\theta(T) > \theta(T+1)$}
  \STATE $T_- \leftarrow T+1$
\ELSIF {$\theta(T) = \theta(T+1)$}
  \RETURN $T$ //Optimum found
\ENDIF

\UNTIL{$T_- = T_+$}

\RETURN $T_-$  
\end{algorithmic} 
}
\end{algorithm}

\section{Simulation Results }

\label{sec:simulation} 

In the simulations, we assume that there exist a backend cloud (with
index $k_{0}$) and multiple MMCs. A service instance can be placed
either on one of the MMCs or on the backend cloud. We first define
\begin{equation}
R(y)\triangleq\begin{cases}
\frac{1}{1-\frac{y}{Y}}, & \textrm{if }y<Y\\
+\infty, & \textrm{if }y\geq Y
\end{cases}\label{eq:simCostR}
\end{equation}
where $Y$ denotes the capacity of a single MMC. Then, we define the
local and migration costs as in (\ref{eq:localCostBasedOnPiDef}), (\ref{eq:migCostBasedOnPiDef}), with 
\begin{align}
& u_{k,t}(y_{k}(t))\triangleq\begin{cases}
\tilde{g}y_{k}(t), & \textrm{if }k=k_{0}\\
y_{k}(t)R(y_{k}(t))+gr_{k}(t),\!\!\!\!\! & \textrm{if }k\neq k_{0}
\end{cases}\label{eq:simCost1} \\
 & w_{kl,t}(y_{k}(t-1),y_{l}(t),z_{kl}(t))\nonumber \\
 & \quad \triangleq\begin{cases}
\tilde{h}z_{kl}(t),\quad\quad\quad\quad\quad\textrm{ if }k=k_{0}\textrm{ or/and }l=k_{0}\\
z_{kl}(t)\left(R(y_{k}(t))+R(y_{l}(t))\right)+hs_{kl}(t),\textrm{ else}
\end{cases}\label{eq:simCost2}
\end{align}
where $y_{k}(t)$ and $z_{kl}(t)$ are sum resource consumptions defined as
in Section \ref{sub:PerfAnalysisDefinitions}, $r_{k}(t)$ is the
sum of the distances between each instance running on cloud $k$ and all users connected to this instance, $s_{kl}(t)$ is the 
distance between clouds $k$ and $l$ multiplied by the number migrated
instances from cloud $k$ to cloud $l$, and $\tilde{g},g,\tilde{h},h$
are simulation parameters (specified later). The \emph{distance} here is expressed
as the number of hops on the communication network.

Similar to Section \ref{sub:onlinePerformanceOtherCostFunc}
(but with migration cost here), we consider the scenario where the connection status to the backend cloud remains relatively unchanged.
Thus, in (\ref{eq:simCost1}) and (\ref{eq:simCost2}), $u_{k,t}(\cdot)$ and $w_{kl,t}(\cdot,\cdot,\cdot)$ are linear in  $y_{k}(t)$ and $z_{kl}(t)$ when involving the backend cloud $k_{0}$.
When not involving the
backend  cloud, the cost functions have two terms. The first term
contains $R(\cdot)$ and is related to the queuing
delay of data processing/transmission, because
$R(\cdot)$ has a similar form as the average queueing delay expression
from queueing theory. The additional coefficient $y_{k}(t)$ or
$z_{kl}(t)$ scales the delay by the total amount of workload so that
experiences of all instances (hosted at a cloud or being
migrated) are considered. This expression is also a widely used objective
(such as in \cite{refViNEYard}) for pushing the system towards a
load-balanced state. The second term has the distance of data transmission or migration, which is related to propagation delay. Thus, both queueing and propagation delays are captured in the cost definition above.

Note that the above defined cost functions are heterogeneous, because
the cost definitions are different depending on whether the backend
cloud is involved or not. Therefore, we cannot directly apply the
existing MDP-based approaches \cite{MDPFollowMeICC2014,wang2014milcom,wang2015IFIPNetworking}
to solve this problem. We consider users continuously connected to
service instances, so we also cannot apply the technique in \cite{urgaonkar2015performance}.

\subsection{Synthetic Arrivals and Departures}

To evaluate how much worse the online placement (presented in Section
\ref{sec:onlinePlacement}) performs compared to the optimal offline
placement (presented in Section \ref{sec:optSolutionGivenLookAhead}),
we first consider a setting with synthetic instance arrivals and departures.
For simplicity, we ignore the migration cost and set $g=0$ to make
the local cost independent of the distance $r_{k}(t)$. We set $Y=5$,
$\tilde{g}=3$, and the total number of clouds $K=5$ among which
one is the backend cloud. We simulate $4000$ arrivals, where the
local resource consumption of each arrival is uniformly distributed
within interval $[0.5,1.5]$. Before a new instance arrives, we generate
a random variable $H$ that is uniformly distributed within $[0,1]$.
If $H<0.1$, one randomly selected instance that is currently running
in the system (if any) departs. We only focus on the cost in a single
timeslot and assume that arrival and departure events happen within
this slot. The online placement greedily places each instance, while
the offline placement considers all instances as an entirety. We compare
the cost of the proposed online placement algorithm with a lower bound
of the cost of the optimal placement. The optimal lower bound is obtained
by solving an optimization problem that allows every instance to be
arbitrarily split across multiple clouds, in which case the problem
becomes a convex optimization problem due to the relaxation of integer
constraints. 

The simulation is run with $100$ different random seeds. Fig.~\ref{fig:simSynthetic}
shows the overall results. We see that the cost is convex increasing
when the number of arrived instances is small, and it increases linearly
when the number of instances becomes large, because in the latter
case, the MMCs are close to being overloaded and most instances are
placed at the backend  cloud. The fact that the average performance
ratio (defined as $\textrm{mean}\left(\widetilde{D}(\mathbf{x})\right) \Big/ \textrm{mean}\left(\widetilde{D}(\mathbf{x}^{*})\right)$)
converges with increasing number of instances in Fig.~\ref{fig:simSynthetic}(b)
supports our analysis in Section \ref{sub:onlinePerformanceOtherCostFunc}.

\begin{figure}
\centering \subfigure[]{ \includegraphics[width=0.48\columnwidth]{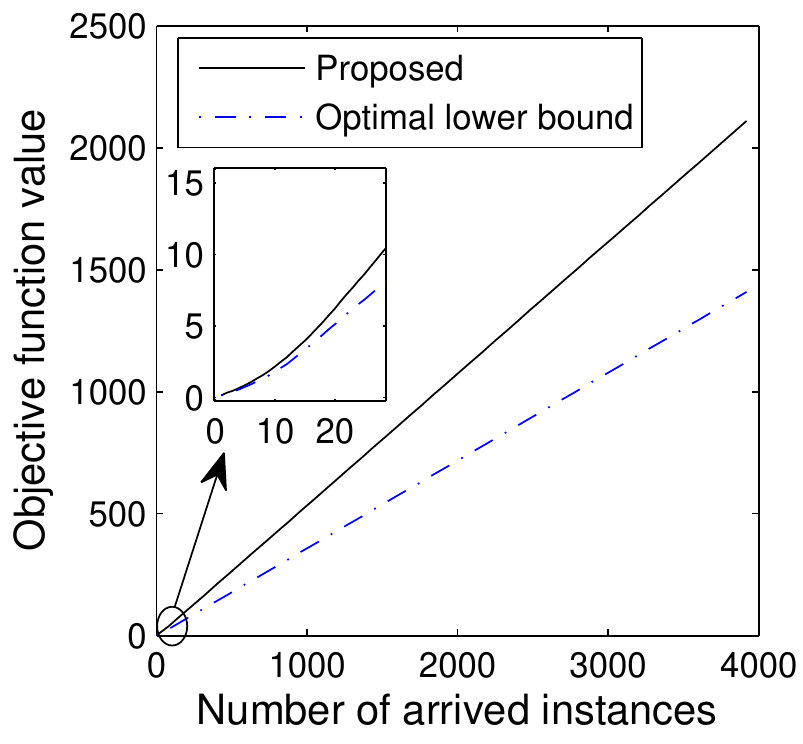}} \subfigure[]{\includegraphics[width=0.48\columnwidth]{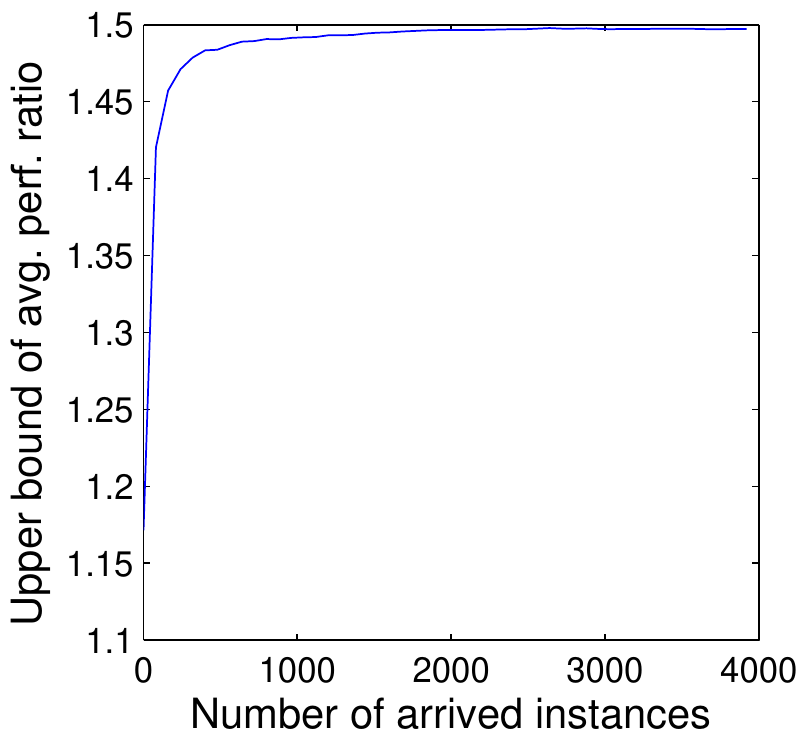}}

\protect\caption{Results with synthetic traces: (a) objective function value, (b) average
performance ratio.}
\label{fig:simSynthetic}
\end{figure}

\subsection{Real-World Traces}

To further evaluate the performance while considering the impact of
prediction errors and look-ahead window size, we perform simulations
using real-world San Francisco taxi traces obtained on the day of
May 31, 2008 \cite{comsnets09piorkowskiMobilityTraces,epfl-mobility-2009-02-24}.
Similar to \cite{wang2015IFIPNetworking}, we assume that the MMCs are deployed according to a hexagonal cellular
structure in the central area of San Francisco (the center of this area has latitude $37.762$ and longitude $-122.43$). The distance between the center points of adjacent
cells is $1000\textrm{ m}$. We consider $K-1=91$ cells (thus MMCs),
one backend cloud, and $50$ users (taxis) in total and not all the
users are active at a given time. A user is considered active if its most recent location update was received within $600$~s from the current time and its location is within the area covered by MMCs. Each user may require at most one service at a time from the cloud
when it is active, where the duration that each active user requires
(or, does not require) service is exponentially distributed with a
mean value of $50$ slots (or, $10$ slots). When a user requires
service, we assume that there is a service instance for this particular
request (independent from other users) running on one of the clouds.
The local and migration (if migration occurs) resource consumptions
of each such instance are set to $1$. We assume that the online algorithm
has no knowledge on the departure time of instances and set $T_{\textrm{life}}=\infty$
for all instances. Note that the taxi locations in the dataset are
unevenly distributed,
so it is still possible that one MMC hosts multiple services although
the maximum possible number of instances ($50$) is smaller than the
number of MMCs ($91$). The distance metric (for evaluating $r_{k}(t)$
and $s_{kl}(t)$) is defined as the minimum number of hops between
two locations on the cellular structure. The physical time corresponding
to each slot is set to $60$~s. We set the parameters $\Gamma=1.5,\sigma=2,Y=5,\tilde{g}=\tilde{h}=3,g=h=0.2$.
The cost prediction error is assumed to have an upper bound in the
form of $F(T)=\beta T^{\alpha}$ (see Section \ref{sub:predErrorOptWind}),
where we fix $\alpha=1.1$. The prediction error is generated randomly while ensuring that the upper
bound is satisfied.

The simulation results are shown in Fig. \ref{fig:simRealWorld}.
In Fig. \ref{fig:simRealWorld}(a), we can see that the result of
the proposed online placement approach (E) performs close to the case
of online placement with precise future knowledge (D), where approach
D assumes that all the future costs as well as instance arrival and
departure times are precisely known, but we still use the online algorithm
to determine the placement (i.e., we greedily place each instance),
because the offline algorithm is too time consuming due to its high
complexity. The proposed method E also outperforms alternative methods
including only placing on MMCs and never migrate the service instance
after initialization (A), always following the user when the user
moves to a different cell (B), as well as always placing the service
instance on the backend cloud (C). In approaches A and B, the instance
placement is determined greedily so that the distance between the
instance and its corresponding user is the shortest, subject to the
MMC capacity constraint $Y$ so that the costs are finite (see (\ref{eq:simCostR})).
The fluctuation of the cost during the day is because of different
number of users that require the service (thus different system load).
In Fig. \ref{fig:simRealWorld}(b), we show the average cost over
the day with different look-ahead window sizes and $\beta$ values
(a large $\beta$ indicates a large prediction error), where the average
results from $8$ different random seeds are shown. We see that the
optimal window size ($T^{*}$) found from the method proposed in Section
\ref{sec:findLookAhead} is close to the window size that brings the
lowest cost, which implies that the proposed method for finding $T^{*}$
is reasonably accurate. 

Additional results on the amount of computation time and floating-point operations (FLOP) for the results in Fig. \ref{fig:simRealWorld}(a) are given in Appendix \ref{sec:additionalSimResult}.

\begin{figure}
\centering \subfigure[]{ \includegraphics[width=0.55\columnwidth]{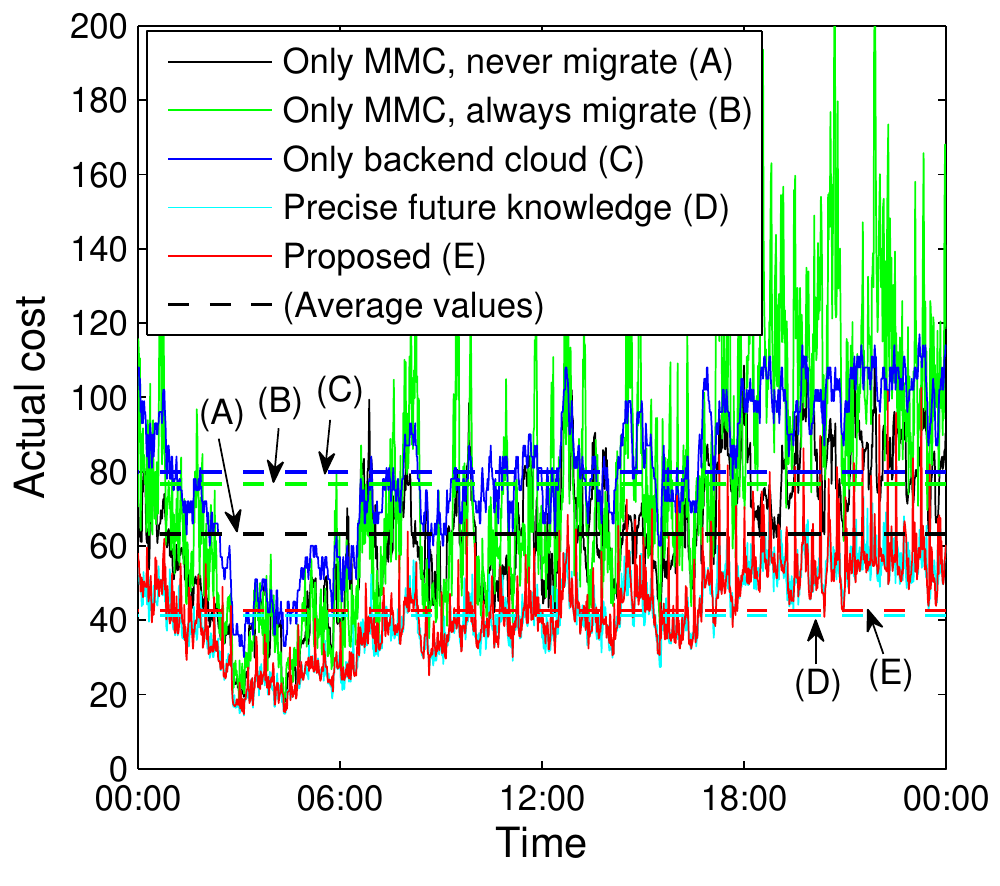}} \subfigure[]{\includegraphics[width=0.42\columnwidth]{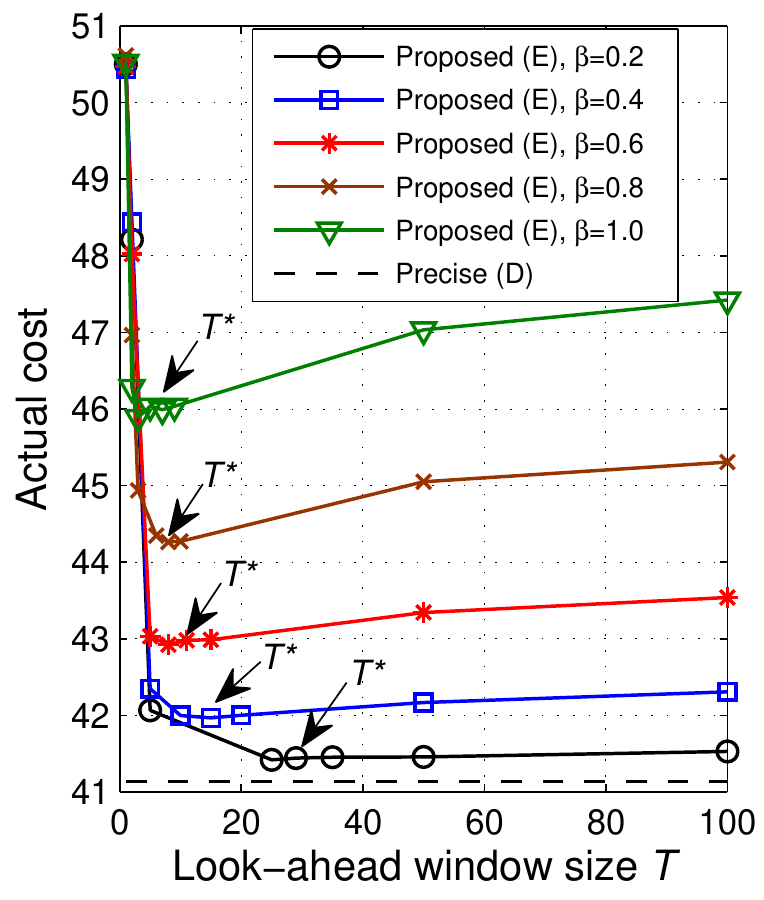}}

\protect\caption{Results with real-world traces (where the costs are summed over all
clouds, i.e., the $A(t)$ values): (a) Actual costs at different time
of a day, where $\beta=0.4$ for the proposed method E. The arrows
point to the average values over the whole day of the corresponding
policy. (b) Actual costs averaged over the whole day.}
\label{fig:simRealWorld}
\end{figure}

\section{Conclusions \label{sec:Conclusions}}

In this paper, we have studied the dynamic service placement problem
for MMCs with multiple service instances, where the future costs
are predictable within a known accuracy. We have proposed algorithms
for both offline and online placements, as well as a method for finding
the optimal look-ahead window size. The performance of the proposed
algorithms has been evaluated both analytically and using simulations
with synthetic instance arrival/departure traces and real-world user mobility traces of San Francisco
taxis. The simulation results support our analysis.

Our results are based on a general cost function that can represent different aspects in practice. As long as one can assign a cost for every possible configuration, the proposed algorithms are applicable, and the time-complexity results hold. The optimality gap for the online algorithm has been analyzed for a narrower class of functions, which is still very general as discussed earlier in the paper.

The theoretical framework used for analyzing the performance of online
placement can be extended to incorporate more general cases, such
as those where there exist multiple types of resources in each cloud.
We envision that the performance results are similar. We also note
that our framework can be applied to analyzing a large class of online
resource allocation problems that have convex objective functions.

\section*{Acknowledgment}

Contribution of S. Wang is related to his previous affiliation with Imperial College London. Contributions of R. Urgaonkar and T. He are related to their previous affiliation with IBM T. J. Watson Research Center. Contribution of M. Zafer is not related to his current employment at Nyansa Inc.

This research was sponsored in part by the U.S. Army Research Laboratory
and the U.K. Ministry of Defence and was accomplished under Agreement
Number W911NF-06-3-0001. The views and conclusions contained in this
document are those of the author(s) and should not be interpreted
as representing the official policies, either expressed or implied,
of the U.S. Army Research Laboratory, the U.S. Government, the U.K.
Ministry of Defence or the U.K. Government. The U.S. and U.K. Governments
are authorized to reproduce and distribute reprints for Government
purposes notwithstanding any copyright notation hereon.

\bibliographystyle{IEEEtran}
\bibliography{IEEEabrv,migrationWithPredictedCost}

\begin{thebibliography}{10}
\providecommand{\url}[1]{#1}
\csname url@samestyle\endcsname
\providecommand{\newblock}{\relax}
\providecommand{\bibinfo}[2]{#2}
\providecommand{\BIBentrySTDinterwordspacing}{\spaceskip=0pt\relax}
\providecommand{\BIBentryALTinterwordstretchfactor}{4}
\providecommand{\BIBentryALTinterwordspacing}{\spaceskip=\fontdimen2\font plus
\BIBentryALTinterwordstretchfactor\fontdimen3\font minus
  \fontdimen4\font\relax}
\providecommand{\BIBforeignlanguage}[2]{{%
\expandafter\ifx\csname l@#1\endcsname\relax
\typeout{** WARNING: IEEEtran.bst: No hyphenation pattern has been}%
\typeout{** loaded for the language `#1'. Using the pattern for}%
\typeout{** the default language instead.}%
\else
\language=\csname l@#1\endcsname
\fi
#2}}
\providecommand{\BIBdecl}{\relax}
\BIBdecl

\bibitem{wang2015dynamic}
S.~Wang, R.~Urgaonkar, K.~Chan, T.~He, M.~Zafer, and K.~K. Leung, ``Dynamic
  service placement for mobile micro-clouds with predicted future cost,'' in
  \emph{Proc. of IEEE ICC 2015}, Jun. 2015.

\bibitem{cloudletCMUMobiCASE}
M.~Satyanarayanan, Z.~Chen, K.~Ha, W.~Hu, W.~Richter, and P.~Pillai,
  ``Cloudlets: at the leading edge of mobile-cloud convergence,'' in
  \emph{Proc. of MobiCASE 2014}, Nov. 2014.

\bibitem{IBMWhitepaper}
\BIBentryALTinterwordspacing
``Smarter wireless networks,'' \emph{IBM Whitepaper No. WSW14201USEN}, Feb.
  2013. [Online]. Available:
  \url{www.ibm.com/services/multimedia/Smarter\_wireless\_networks.pdf}
\BIBentrySTDinterwordspacing

\bibitem{CommMagEdgeComput}
M.~Satyanarayanan, R.~Schuster, M.~Ebling, G.~Fettweis, H.~Flinck, K.~Joshi,
  and K.~Sabnani, ``An open ecosystem for mobile-cloud convergence,''
  \emph{IEEE Communications Magazine}, vol.~53, no.~3, pp. 63--70, Mar. 2015.

\bibitem{FollowMeMagazine}
T.~Taleb and A.~Ksentini, ``Follow me cloud: interworking federated clouds and
  distributed mobile networks,'' \emph{IEEE Network}, vol.~27, no.~5, pp.
  12--19, Sept. 2013.

\bibitem{Edge-as-a-service}
S.~Davy, J.~Famaey, J.~Serrat-Fernandez, J.~Gorricho, A.~Miron, M.~Dramitinos,
  P.~Neves, S.~Latre, and E.~Goshen, ``Challenges to support
  edge-as-a-service,'' \emph{IEEE Communications Magazine}, vol.~52, no.~1, pp.
  132--139, Jan. 2014.

\bibitem{becvar2014pimrc}
Z.~Becvar, J.~Plachy, and P.~Mach, ``Path selection using handover in mobile
  networks with cloud-enabled small cells,'' in \emph{Proc. of IEEE PIMRC
  2014}, Sept. 2014.

\bibitem{VNEmbeddingSurvey}
A.~Fischer, J.~Botero, M.~Beck, H.~De~Meer, and X.~Hesselbach, ``Virtual
  network embedding: A survey,'' \emph{{IEEE} Commun. Surveys Tuts.}, vol.~15,
  no.~4, pp. 1888--1906, 2013.

\bibitem{refViNEYard}
M.~Chowdhury, M.~Rahman, and R.~Boutaba, ``Vineyard: Virtual network embedding
  algorithms with coordinated node and link mapping,'' \emph{IEEE/ACM
  Transactions on Networking}, vol.~20, no.~1, pp. 206--219, 2012.

\bibitem{MDPFollowMeICC2014}
A.~Ksentini, T.~Taleb, and M.~Chen, ``A {Markov} decision process-based service
  migration procedure for follow me cloud,'' in \emph{Proc. of IEEE ICC 2014},
  Jun. 2014.

\bibitem{wang2014milcom}
S.~Wang, R.~Urgaonkar, T.~He, M.~Zafer, K.~Chan, and K.~K. Leung,
  ``Mobility-induced service migration in mobile micro-clouds,'' in \emph{Proc.
  of IEEE MILCOM 2014}, Oct. 2014.

\bibitem{wang2015IFIPNetworking}
S.~Wang, R.~Urgaonkar, M.~Zafer, T.~He, K.~Chan, and K.~K. Leung, ``Dynamic
  service migration in mobile edge-clouds,'' in \emph{Proc. of IFIP Networking
  2015}, May 2015.

\bibitem{SrivatsaNonMarkovianMobility}
M.~Srivatsa, R.~Ganti, J.~Wang, and V.~Kolar, ``Map matching: Facts and
  myths,'' in \emph{Proc. of ACM SIGSPATIAL 2013}, 2013, pp. 484--487.

\bibitem{bookOnlineComputation}
A.~Borodin and R.~El-Yaniv, \emph{Online Computation and Competitive
  Analysis}.\hskip 1em plus 0.5em minus 0.4em\relax Cambridge University Press,
  1998.

\bibitem{krumke2006online}
S.~O. Krumke, \emph{Online optimization: Competitive analysis and
  beyond}.\hskip 1em plus 0.5em minus 0.4em\relax Habilitationsschrift
  Technische Universitaet Berlin, 2001.

\bibitem{AzarOnlineConvexObj14}
\BIBentryALTinterwordspacing
Y.~Azar, I.~R. Cohen, and D.~Panigrahi, ``Online covering with convex
  objectives and applications,'' \emph{CoRR}, vol. abs/1412.3507, Dec. 2014.
  [Online]. Available: \url{http://arxiv.org/abs/1412.3507}
\BIBentrySTDinterwordspacing

\bibitem{BuchbinderOnlinePackingCovering}
\BIBentryALTinterwordspacing
N.~Buchbinder, S.~Chen, A.~Gupta, V.~Nagarajan, and J.~Naor, ``Online packing
  and covering framework with convex objectives,'' \emph{CoRR}, vol.
  abs/1412.8347, Dec. 2014. [Online]. Available:
  \url{http://arxiv.org/abs/1412.8347}
\BIBentrySTDinterwordspacing

\bibitem{urgaonkar2015performance}
R.~Urgaonkar, S.~Wang, T.~He, M.~Zafer, K.~Chan, and K.~K. Leung, ``Dynamic
  service migration and workload scheduling in edge-clouds,'' \emph{Performance
  Evaluation}, vol.~91, pp. 205--228, Sept. 2015, to be presented at IFIP
  Performance 2015.

\bibitem{hsiao2013load}
H.-C. Hsiao, H.-Y. Chung, H.~Shen, and Y.-C. Chao, ``Load rebalancing for
  distributed file systems in clouds,'' \emph{IEEE Trans. on Parallel and
  Distributed Systems}, vol.~24, no.~5, pp. 951--962, 2013.

\bibitem{AzarLoadBalancingSurvey}
Y.~Azar, ``On-line load balancing,'' \emph{Theoretical Computer Science}, pp.
  218--225, 1992.

\bibitem{LiOptimalDynamicMobility}
J.~Li, H.~Kameda, and K.~Li, ``Optimal dynamic mobility management for pcs
  networks,'' \emph{IEEE/ACM Trans. Netw.}, vol.~8, no.~3, pp. 319--327, Jun.
  2000.

\bibitem{LiLoadBalancing}
J.~Li and H.~Kameda, ``Load balancing problems for multiclass jobs in
  distributed/parallel computer systems,'' \emph{IEEE Transactions on
  Computers}, vol.~47, no.~3, pp. 322--332, Mar. 1998.

\bibitem{chen1999ordinal}
C.-H. Chen, S.~D. Wu, and L.~Dai, ``Ordinal comparison of heuristic algorithms
  using stochastic optimization,'' \emph{IEEE Trans. on Robotics and
  Automation}, vol.~15, no.~1, pp. 44--56, 1999.

\bibitem{Aspnes:1997:ORV:258128.258201}
J.~Aspnes, Y.~Azar, A.~Fiat, S.~Plotkin, and O.~Waarts, ``On-line routing of
  virtual circuits with applications to load balancing and machine
  scheduling,'' \emph{J. ACM}, vol.~44, no.~3, pp. 486--504, May 1997.

\bibitem{cloudMonitorSurvey}
G.~Aceto, A.~Botta, W.~de~Donato, and A.~Pescape, ``Cloud monitoring: A
  survey,'' \emph{Computer Networks}, vol.~57, no.~9, pp. 2093 -- 2115, 2013.

\bibitem{Cho:2011:FMU:2020408.2020579}
E.~Cho, S.~A. Myers, and J.~Leskovec, ``Friendship and mobility: User movement
  in location-based social networks,'' in \emph{Proc. of the 17th ACM SIGKDD
  International Conference on Knowledge Discovery and Data Mining}, ser. KDD
  '11, 2011, pp. 1082--1090.

\bibitem{lacurts2014cicada}
K.~LaCurts, J.~Mogul, H.~Balakrishnan, and Y.~Turner, ``Cicada: Introducing
  predictive guarantees for cloud networks,'' Jun. 2014.

\bibitem{powell2007approximate}
W.~B. Powell, \emph{Approximate Dynamic Programming: Solving the curses of
  dimensionality}.\hskip 1em plus 0.5em minus 0.4em\relax John Wiley \& Sons,
  2007.

\bibitem{korte2002combinatorial}
B.~Korte and J.~Vygen, \emph{Combinatorial optimization}.\hskip 1em plus 0.5em
  minus 0.4em\relax Springer, 2002.

\bibitem{boyd2004convex}
S.~Boyd and L.~Vandenberghe, \emph{Convex optimization}.\hskip 1em plus 0.5em
  minus 0.4em\relax Cambridge university press, 2004.

\bibitem{comsnets09piorkowskiMobilityTraces}
M.~Piorkowski, N.~Sarafijanovoc-Djukic, and M.~Grossglauser, ``A parsimonious
  model of mobile partitioned networks with clustering,'' in \emph{Proc. of
  COMSNETS}, Jan. 2009.

\bibitem{epfl-mobility-2009-02-24}
M.~Piorkowski, N.~Sarafijanovic-Djukic, and M.~Grossglauser, ``{CRAWDAD} data
  set epfl/mobility (v. 2009-02-24),'' Downloaded from
  http://crawdad.org/epfl/mobility/, Feb. 2009.

\end{thebibliography}

\begin{IEEEbiography}[{\includegraphics[width=1in,height=1.25in,clip,keepaspectratio]{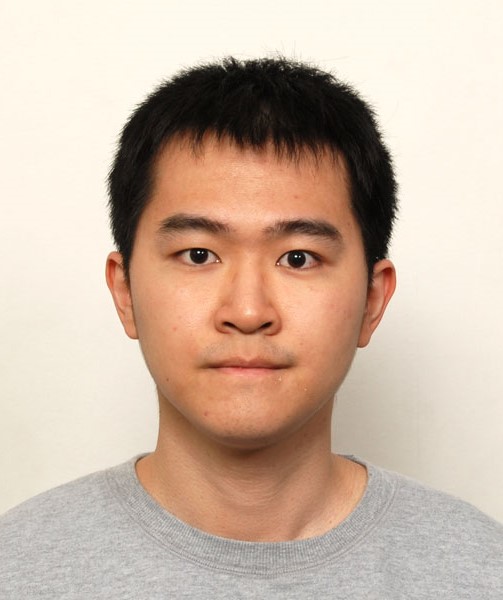}}]{Shiqiang Wang}
received the PhD degree from Imperial College London, United Kingdom, in 2015. Before that, he received the MEng and BEng degrees from Northeastern University, China, respectively in 2011 and 2009. He is currently a Research Staff Member at IBM T.J. Watson Research Center, United States, where he also worked as a Graduate-Level Co-op in the summers of 2014 and 2013. In the autumn of 2012, he worked at NEC Laboratories Europe, Germany.
His research interests include dynamic control mechanisms, optimization algorithms, protocol design and prototyping, with applications to mobile cloud
computing, hybrid and heterogeneous networks, ad-hoc networks, and cooperative communications. He has over 30 scholarly publications, and has served as a
technical program committee (TPC) member or reviewer for a number of international journals and conferences. He received the 2015 Chatschik Bisdikian Best Student Paper Award of the Network and Information Sciences International Technology Alliance (ITA).
\end{IEEEbiography}

\begin{IEEEbiography}[{\includegraphics[width=1in,height=1.25in,clip,keepaspectratio]{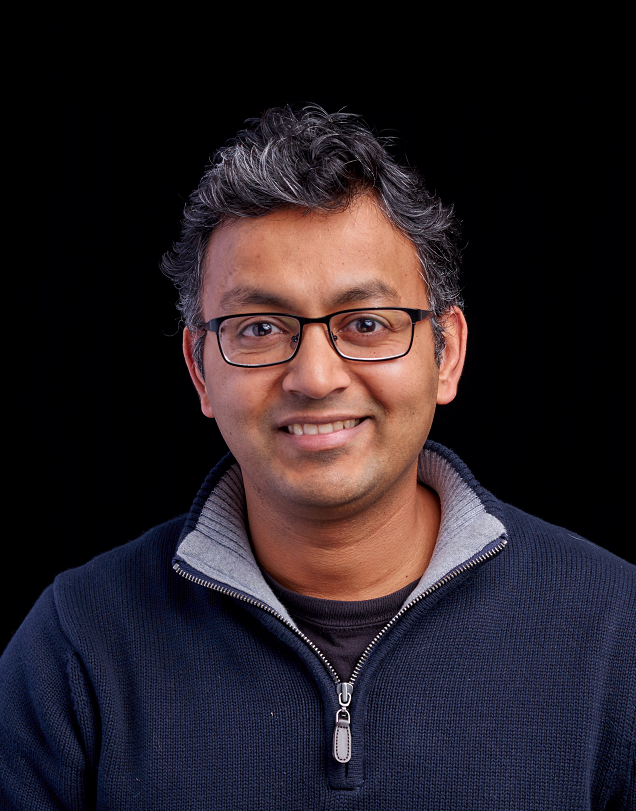}}]{Rahul Urgaonkar}
is an Operations Research Scientist with the Modeling and Optimization group at Amazon. Previously, he was with IBM Research where he was a task leader on the US Army Research Laboratory (ARL) funded Network Science Collaborative Technology Alliance (NS CTA) program. He was also a Primary Researcher in the US/UK International Technology Alliance (ITA) research programs. His research is in the area of stochastic optimization, algorithm design and control with applications to communication networks and cloud-computing systems. Dr. Urgaonkar obtained his Masters and PhD degrees from the University of Southern California and his Bachelors degree (all in Electrical Engineering) from the Indian Institute of Technology Bombay.
\end{IEEEbiography}

\begin{IEEEbiography}[{\includegraphics[width=1in,height=1.25in,clip,keepaspectratio]{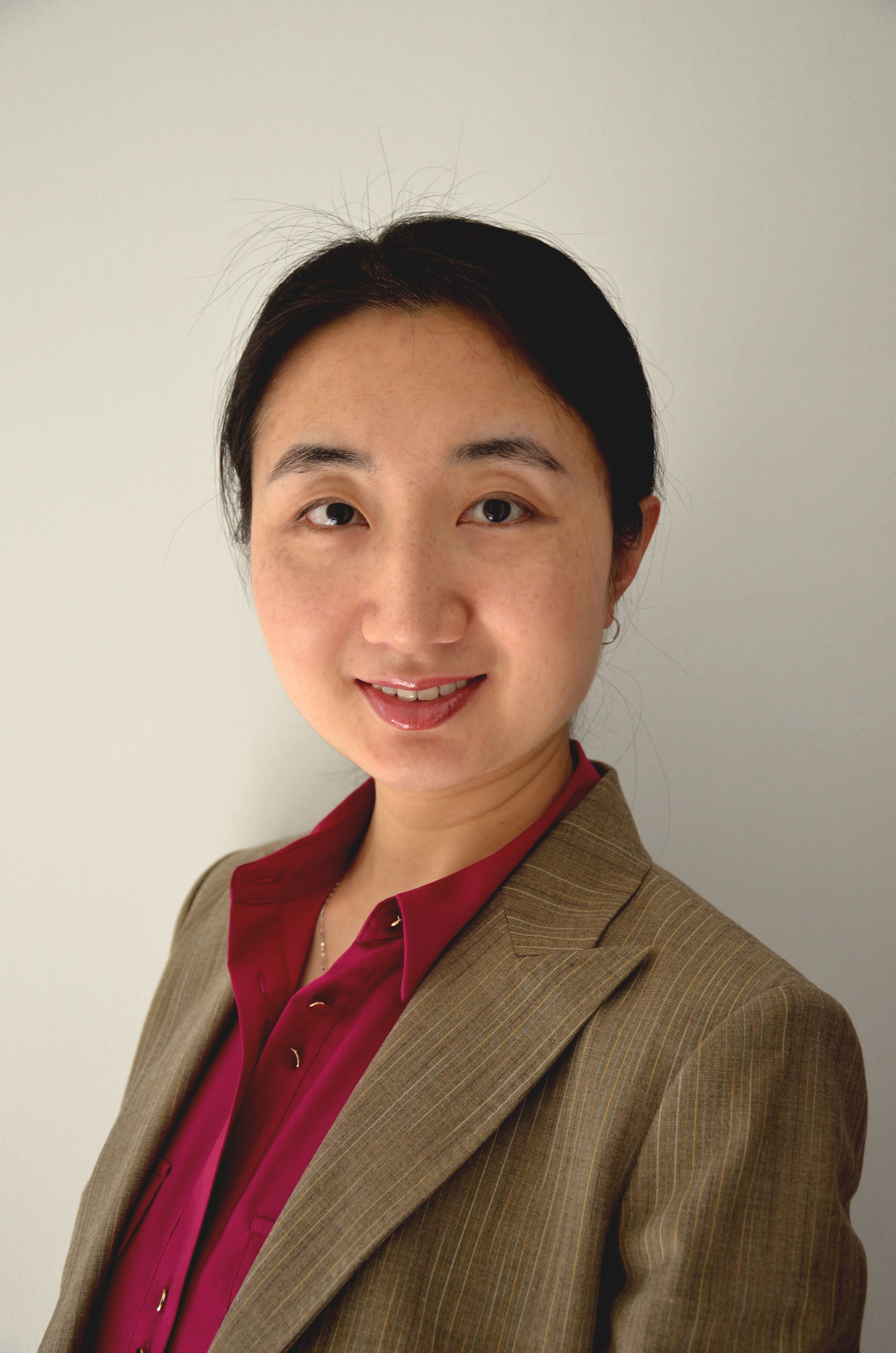}}]{Ting He}
received the B.S. degree in computer science from Peking University, China, in 2003 and the Ph.D. degree in electrical and computer engineering from Cornell University, Ithaca, NY, in 2007.
She is an Associate Professor in the School of Electrical Engineering and Computer Science at Pennsylvania State University, University Park, PA. From 2007 to 2016, she was a Research Staff Member in the Network Analytics Research Group at IBM T.J. Watson Research Center, Yorktown Heights, NY. Her work is in the broad areas of network modeling and optimization, statistical inference, and information theory. 
Dr. He is a senior member of IEEE. She has served as the Membership co-chair of ACM N2Women and the GHC PhD Forum committee. She has served on the TPC of a range of communications and networking conferences, including IEEE INFOCOM (Distinguished TPC Member), IEEE SECON, IEEE WiOpt, IEEE/ACM IWQoS, IEEE MILCOM, IEEE ICNC, IFIP Networking, etc. She received 
the Research Division Award and the Outstanding Contributor Awards from IBM in 2016, 2013, and 2009. She received the Most Collaboratively  
Complete Publications Award by ITA in 2015, the Best Paper Award at the 2013 International Conference on Distributed Computing Systems
(ICDCS), a Best Paper Nomination at the 2013 Internet Measurement Conference (IMC), and 
the Best Student Paper Award at the 2005 International Conference on Acoustic, Speech and Signal Processing (ICASSP). Her students received the Outstanding Student Paper Award at the 2015 ACM SIGMETRICS and the Best Student Paper Award at the 2013 ITA Annual Fall Meeting.\end{IEEEbiography}

\begin{IEEEbiography}[{\includegraphics[width=1in,height=1.25in,clip,keepaspectratio]{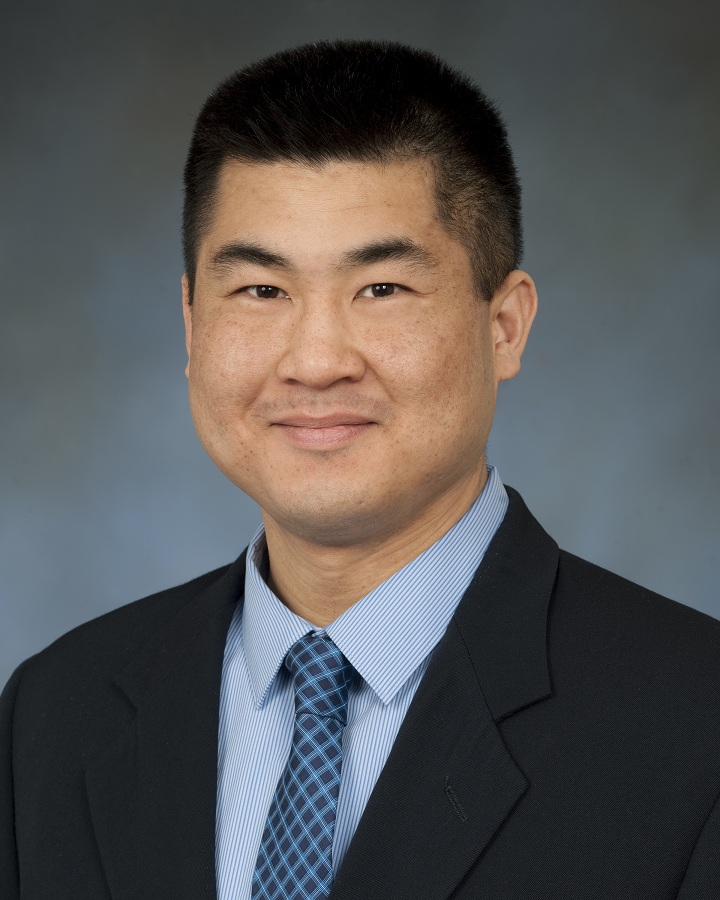}}]{Kevin Chan}
is research scientist with the Computational and Information Sciences Directorate at the U.S. Army Research Laboratory (Adelphi, MD). Previously, he was an ORAU postdoctoral research fellow at ARL. His research interests are in network science and dynamic distributed computing, with past work in dynamic networks, trust and distributed decision making and quality of information. He has been an active researcher in ARL's collaborative programs, the Network Science Collaborative Technology Alliance and Network and Information Sciences International Technology Alliance. Prior to ARL, he received a PhD in Electrical and Computer Engineering (ECE) and MSECE from Georgia Institute of Technology (Atlanta, GA. He also received a BS in ECE/EPP from Carnegie Mellon University (Pittsburgh, PA).
\end{IEEEbiography}

\begin{IEEEbiography}[{\includegraphics[width=1in,height=1.25in,clip,keepaspectratio]{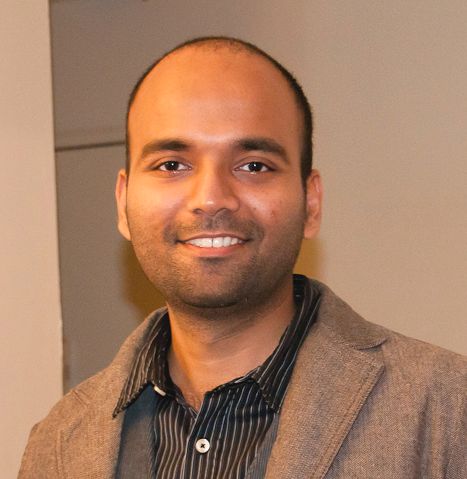}}]{Murtaza Zafer}
received the B.Tech. degree in Electrical Engineering from the Indian Institute of Technology, Madras, in 2001, and the Ph.D. and S.M. degrees in Electrical Engineering and Computer Science from the Massachusetts Institute of Technology
in 2003 and 2007 respectively. He currently works at Nyansa Inc., where he heads the analytics portfolio of the company, building a scalable big data system for analyzing network data. Prior to this, he was a Senior Research Engineer at Samsung Research
America, where his research focused on machine learning, deep-neural networks, big data and cloud computing systems. From 2007-2013 he was a Research Scientist at the IBM T.J. Watson Research Center, New York, where his research focused on computer and communication networks, data-analytics and cloud computing. He was a technical lead on several research projects in the US-UK funded multi-institutional International Technology Alliance program with emphasis on fundamental research in mobile wireless networks. He has previously worked at the Corporate R\&D center of Qualcomm Inc. and at Bell Laboratories, Alcatel-Lucent Inc., during the summers of 2003 and 2004 respectively. 

Dr. Zafer serves as an Associate Editor for the IEEE Network magazine. He is a co-recipient of the Best Paper Award at the IEEE/IFIP International Symposium on Integrated Network Management, 2013, and the Best Student Paper award at the International Symposium on Modeling and Optimization in Mobile, Ad Hoc, and Wireless Networks (WiOpt) in 2005, a recipient of the Siemens and Philips Award in 2001 and a recipient of several invention achievement awards at IBM Research.
\end{IEEEbiography}

\begin{IEEEbiography}[{\includegraphics[width=1in,height=1.25in,clip,keepaspectratio]{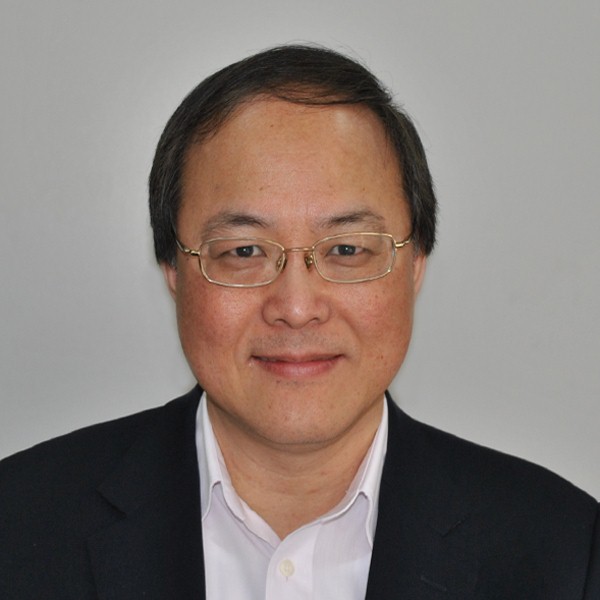}}]{Kin K. Leung}
received his B.S. degree from the Chinese University of Hong Kong in 1980, and his M.S. and Ph.D. degrees from University of California, Los
Angeles, in 1982 and 1985, respectively. He joined AT\&T Bell Labs in New Jersey in 1986 and worked at its successors, AT\&T Labs and Lucent Technologies
Bell Labs, until 2004. Since then, he has been the Tanaka Chair Professor in the Electrical and Electronic Engineering (EEE), and Computing Departments at
Imperial College in London. He is the Head of Communications and Signal Processing Group in the EEE Department. His current research focuses on protocols,
optimization and modeling of various wireless networks. He also works on multi-antenna and cross-layer designs for these networks. He received the
Distinguished Member of Technical Staff Award from AT\&T Bell Labs (1994), and was a co-recipient of the Lanchester Prize Honorable Mention Award (1997).
He was elected an IEEE Fellow (2001), received the Royal Society Wolfson Research Merits Award (2004-09) and became a member of Academia Europaea (2012).
He also received several best paper awards, including the IEEE PIMRC 2012 and ICDCS 2013. He has actively served on conference committees.  He serves as a
member (2009-11) and the chairman (2012-15) of the IEEE Fellow Evaluation Committee for Communications Society. He was a guest editor for the IEEE JSAC,
IEEE Wireless Communications and the MONET journal, and as an editor for the JSAC: Wireless Series, IEEE Transactions on Wireless Communications and IEEE
Transactions on Communications. Currently, he is an editor for the ACM Computing Survey and International Journal on Sensor Networks.
\end{IEEEbiography}

\clearpage

\appendices{}

\section{Summary of Notations}

\label{app:SummaryNotations} 

The main notations used in this paper are summarized in Table \ref{tab:mainNotations}. 

\begin{table*} \protect\caption{Summary of main notations}  \label{tab:mainNotations} 

\renewcommand{\arraystretch}{1.4} 

\center{\small

\begin{tabularx}{\linewidth}
{>{\setlength\hsize{0.3\hsize}\centering}X X} 
\hline Notation & Description\\ 
\hline 

$\triangleq$ & Defined to be equal to \\

$\cdot$ & Dot-product \\

$K$ & Total number of clouds \\

$k,l\in\{1,2,...,K\}$ & Cloud index \\

$t$ & Timeslot index \\

$T$ & Size of look-ahead window \\

$M$ & Total/maximum number of service instances under consideration \\

$i$ & Service instance index \\

$\boldsymbol{\pi}(t_{0},...,t_{n})$ & Configuration matrix for slots $\{t_{0},...,t_{n}\}$, written in $\boldsymbol{\pi}$ for short \\

$U(t,\boldsymbol{\pi}(t))$ & Local cost in slot $t$\\

$W(t,\boldsymbol{\pi}(t-1),\boldsymbol{\pi}(t))$ & Migration cost between slots $t-1$ and $t$\\

$C_{\boldsymbol{\pi}(t-1,t)}(t)$ & Sum of local and migration costs in slot $t$ when following configuration $\boldsymbol{\pi}(t-1,t)$\\

$A_{\boldsymbol{\pi}(t-1,t)}(t)$ & Actual cost (actual value of $C_{\boldsymbol{\pi}(t-1,t)}(t)$) \\

$D_{\boldsymbol{\pi}(t-1,t)}^{t_0}(t)$ & Predicted cost (predicted value of $C_{\boldsymbol{\pi}(t-1,t)}(t)$ when prediction is made at slot $t_0$) \\

$\epsilon(\tau)$ & Equal to $\max_{\boldsymbol{\pi}(t-1,t),t_{0}}\left|A_{\boldsymbol{\pi}(t-1,t)}(t_{0}+\tau)-D_{\boldsymbol{\pi}(t-1,t)}^{t_{0}}(t_{0}+\tau)\right|$, the maximum error when looking ahead for $\tau$ slots \\

$\Lambda$ & Set of all possible configuration sequences \\

$\boldsymbol{\lambda}\in\Lambda$ & Configuration sequence (when considering a particular instance $i$, it is equal to the $i$th column of $\boldsymbol{\pi}$) \\

$\Lambda_{i} \subseteq \Lambda$ & Subset of configuration sequences that conform to the arrival and departure times of instance $i$ \\

$x_{i\boldsymbol{\lambda}}$ & Binary variable specifying whether instance $i$ operates in configuration sequence $\boldsymbol{\lambda}$ \\

$a_{i\boldsymbol{\lambda}k}(t)$ & Local resource consumption at cloud $k$ in slot $t$ when instance $i$ is operating under configuration sequence $\boldsymbol{\lambda}$ \\

$b_{i\boldsymbol{\lambda}kl}(t)$ & Migration resource consumption when instance $i$ operating under configuration sequence $\boldsymbol{\lambda}$ is assigned to cloud $k$ in slot $t-1$ and to cloud $l$ in slot $t$ \\

$y_{k}(t)$ & Equal to $\sum_{i=1}^{M}\sum_{\boldsymbol{\lambda}\in\Lambda}a_{i\boldsymbol{\lambda}k}(t)x_{i\boldsymbol{\lambda}}$, sum local resource consumption at cloud $k$ \\

$z_{kl}(t)$ & Equal to $\sum_{i=1}^{M}\sum_{\boldsymbol{\lambda}\in\Lambda}b_{i\boldsymbol{\lambda}kl}(t)x_{i\boldsymbol{\lambda}}$, sum migration resource consumption from cloud $k$ to cloud $l$ \\

$u_{k,t}\left(y_{k}(t)\right)$ & Local cost at cloud $k$ in timeslot $t$ \\

$w_{kl,t}\left(y_{k}(t-1),y_{l}(t),z_{kl}(t)\right)$ & Migration cost from cloud $k$ to cloud $l$ between slots $t-1$ and $t$ \\

$\left(\mathbf{g}\right)_{h_{1}h_{2}}$ (or $\left(\mathbf{g}\right)_{h_{1}h_{2}h_{3}}$) & The $(h_{1},h_{2})$th (or $(h_{1},h_{2},h_{3})$th) element in an arbitrary vector or matrix $\mathbf{g}$ \\

$\mathbf{y}$ & Vector with elements $\left(\mathbf{y}\right)_{kt} \triangleq y_{k}(t)$ \\

$\mathbf{z}$ & Vector with elements $\left(\mathbf{z}\right)_{klt} \triangleq z_{kl}(t)$ \\

$\mathbf{x}$ & Vector with elements $\left(\mathbf{x}\right)_{i\boldsymbol{\lambda}} \triangleq x_{i\boldsymbol{\lambda}}$ \\

$\mathbf{a}_{i\boldsymbol{\lambda}}$ & Vector with elements $\left(\mathbf{a}_{i\boldsymbol{\lambda}}\right)_{kt} \triangleq a_{i\boldsymbol{\lambda}k}(t)$ \\

$\mathbf{b}_{i\boldsymbol{\lambda}}$ & Vector with elements $\left(\mathbf{b}_{i\boldsymbol{\lambda}}\right)_{klt} \triangleq b_{i\boldsymbol{\lambda}kl}(t) $ \\

$\widetilde{D}\left(\mathbf{x}\right)$, $\widetilde{D}\left(\mathbf{y},\mathbf{z}\right)$ & Sum (predicted) cost of all $T$ slots, defined in (\ref{eq:costDef_T_xi}) \\

$\phi$, $\psi$ & Parameters related to the performance gap, defined in (\ref{eq:alphaDef}) and (\ref{eq:betaDef})  \\

$\Gamma$ & Competitive ratio of Algorithm \ref{alg:onlineHighLevel}  \\

$\sigma$ & Parameter related to the migration cost, defined in (\ref{eq:sigmaDefForMigCostBound}) \\

$F(T)$ & Equal to $\sum_{t=t_{0}}^{t_{0}+T-1}\epsilon(t-t_{0})$, the sum-error starting from slot $t_{0}$ up to slot $t_{0}+T-1$ \\

$G(T)$ & The continuous time extension of $F(T)$, see Section \ref{sub:charOptWindLen} \\

$\theta(T)$ & Equal to $\frac{(\Gamma+1)G(T)+\sigma}{T}$, the upper bound in (\ref{eq:CostAccuracyBound}) after replacing $F(T)$ with $G(T)$ \\

\hline
\end{tabularx} 

Note: The timeslot argument $t$ may be omitted in some parts of the discussion for simplicity. Vector elements are referred to with multiple indexes, but we regard vectors as single-indexed vectors for the purposes of vector concatenation (i.e., joining two vectors into one vector) and gradient computation.
}

\end{table*}

\section{Proof of Proposition \ref{prop:NPHardness}}

\label{sec:NPHardProof}

We show that problem (\ref{eq:optWithVectors}) can be reduced from
the partition problem, which is known to be NP-complete \cite[Corollary 15.28]{korte2002combinatorial}.
The partition problem is defined as follows.

\begin{definition} \textbf{(Partition Problem)} Given \emph{positive
integers} $v_{1},v_{2},...,v_{M}$, is there a subset $\mathcal{S}\subseteq\{1,2,...,M\}$
such that $\sum_{j\in\mathcal{S}}v_{j}=\sum_{j\in\mathcal{S}^{c}}v_{j}$,
where $\mathcal{S}^{c}$ is the complement set of $\mathcal{S}$?
\end{definition}

Similarly to the proof of \cite[Theorem 18.1]{korte2002combinatorial},
we define a decision version of the bin packing problem, where we
assume that there are $M$ items each with size 
\[
a_{i}\triangleq\frac{2v_{i}}{\sum_{j=1}^{M}v_{j}}
\]
for all $i\in\{1,2,...,M\}$, and the problem is to determine whether
these $M$ items can be packed into \emph{two bins }each with \emph{unit
size} (i.e., its size is equal to one). It is obvious that this bin
packing decision problem is equivalent to the partition problem.

To solve the above defined bin packing decision problem, we can set
$t_{0}=1$, $T=1$, and $K=2$ in (\ref{eq:optWithVectors}). Because
we attempt to place all items, we set $\Lambda_{i}=\{1,2\}$ for all
$i$. By definition, $w_{kl,t}(\cdot,\cdot,\cdot)=0$ for $t=1$.
We omit the subscript $t$ in the following as we only consider a
single slot. We define $a_{i\boldsymbol{\lambda}k}=a_{i}$ for all
$\boldsymbol{\lambda},k$, and define 
\begin{equation}
u_{k}(y)=\begin{cases}
\epsilon y, & \textrm{if }y\leq1\\
\frac{2\epsilon}{c}(y-1)+\epsilon, & \textrm{if }y>1
\end{cases}\label{eq:NPHardProof}
\end{equation}
where $c\triangleq\frac{1}{\sum_{j=1}^{M}v_{j}}$, and $\epsilon>0$
is an arbitrary constant. 

Because $v_{i}$ is a positive integer for any $i$, we have that
$\frac{a_{i}}{c}=2v_{i}$ is always a positive integer, and $\frac{1}{c}=\sum_{j=1}^{M}v_{j}$
is also always a positive integer. It follows that $y$ can only be
integer multiples of $c$ (where we recall that $y$ is the sum of
$a_{i}$ for those items $i$ that are placed in the bin), and there
exists a positive integer $c'\triangleq\sum_{j=1}^{M}v_{j}$ such
that $c'c=1$. Thus, when $y>1$, we always have $y-1\geq c$. Therefore,
the choice of $u_{k}(y)$ in (\ref{eq:NPHardProof}) guarantees that
$u_{k}(y_{k})\geq3\epsilon>2\epsilon$ whenever bin $k$ ($k\in\{1,2\}$)
exceeds its size, and $\sum_{k=1}^{2}u_{k}\left(y_{k}\right)\leq2\epsilon$
when no bin has exceeded its size. At the same time, $u_{k}(y)$ satisfies
Assumption \ref{condition:costFunc} as long as $c\leq2$. 

By the definition of $c$, we always have $c\leq2$ because $\sum_{j=1}^{M}v_{j}\geq1$.
To solve the bin packing decision problem defined above (thus the
partition problem), we can solve (\ref{eq:optWithVectors}) with the
above definitions. If the solution is not larger than $2\epsilon$,
the packing is feasible and the answer to the partition problem is
``yes''; otherwise, the packing is infeasible and the answer to the
partition problem is ``no''. It follows that problem (\ref{eq:optWithVectors})
is ``at least as hard as'' the partition problem, which proves that
(\ref{eq:optWithVectors}) is NP-hard.

\section{Proof of Proposition \ref{prop:performanceGapResult}}

\label{sec:performanceGapProof}

We first introduce a few lemmas, with results used later in the proof. 

\begin{lemma}\label{lemma:viBound_gradientRelationship1}For any
instance $j$ and configuration sequence $\boldsymbol{\lambda}$,
we have
\begin{equation}
\frac{\partial\widetilde{D}}{\partial x_{j\boldsymbol{\lambda}}}\left(\mathbf{x}\right)=\nabla_{\mathbf{y,z}}\widetilde{D}\left(\mathbf{y},\mathbf{z}\right)\cdot\left(\mathbf{a}_{j\boldsymbol{\lambda}},\mathbf{b}_{j\boldsymbol{\lambda}}\right)
\end{equation}
 \end{lemma}\begin{proof}
\begin{align*}
\frac{\partial\widetilde{D}}{\partial x_{j\boldsymbol{\lambda}}}\left(\mathbf{x}\right) & =\!\!\!\sum_{t=t_{0}}^{t_{0}+T-1}\Bigg[\sum_{k=1}^{K}\frac{\partial\widetilde{D}}{\partial y_{k}(t)}\left(\mathbf{x}\right)\cdot\frac{\partial y_{k}(t)}{\partial x_{j\boldsymbol{\lambda}}}\left(\mathbf{x}\right)+\\
 & \quad\sum_{k=1}^{K}\sum_{l=1}^{K}\frac{\partial\widetilde{D}}{\partial z_{kl}(t)}\left(\mathbf{x}\right)\cdot\frac{\partial z_{kl}(t)}{\partial x_{j\boldsymbol{\lambda}}}\left(\mathbf{x}\right)\Bigg]\\
 & =\!\!\!\sum_{t=t_{0}}^{t_{0}+T-1}\Bigg[\sum_{k=1}^{K}\frac{\partial\widetilde{D}}{\partial y_{k}(t)}\left(\mathbf{x}\right)\cdot a_{j\boldsymbol{\lambda}k}(t)+\\
 & \quad\sum_{k=1}^{K}\sum_{l=1}^{K}\frac{\partial\widetilde{D}}{\partial z_{kl}(t)}\left(\mathbf{x}\right)\cdot b_{j\boldsymbol{\lambda}kl}(t)\Bigg]\\
 & =\nabla_{\mathbf{y,z}}\widetilde{D}\left(\mathbf{y},\mathbf{z}\right)\cdot\left(\mathbf{a}_{j\boldsymbol{\lambda}},\mathbf{b}_{j\boldsymbol{\lambda}}\right)
\end{align*}
where we recall that $y_{k}(t)$ and $z_{kl}(t)$ are functions of
$x_{j\boldsymbol{\lambda}}$ for all $j$ and $\boldsymbol{\lambda}$,
thus they are also functions of vector $\mathbf{x}$. \end{proof}

\begin{lemma}\label{lemma:viBound_gradientRelationship2}For any
instance $j$ and configuration sequence $\boldsymbol{\lambda}$,
we have
\begin{equation}
\nabla_{\mathbf{x}}\widetilde{D}\left(\mathbf{x}\right)\cdot\mathbf{x}=\nabla_{\mathbf{y,z}}\widetilde{D}\left(\mathbf{y},\mathbf{z}\right)\cdot\left(\mathbf{y},\mathbf{z}\right)
\end{equation}
 \end{lemma}\begin{proof}
\begin{align*}
 & \nabla_{\mathbf{x}}\widetilde{D}\left(\mathbf{x}\right)\cdot\mathbf{x}\\
 & =\sum_{j=1}^{M}\sum_{\boldsymbol{\lambda}\in\Lambda}\frac{\partial\widetilde{D}}{\partial x_{j\boldsymbol{\lambda}}}\left(\mathbf{x}\right)\cdot x_{j\boldsymbol{\lambda}}\\
 & =\sum_{j=1}^{M}\sum_{\boldsymbol{\lambda}\in\Lambda}\nabla_{\mathbf{y,z}}\widetilde{D}\left(\mathbf{y},\mathbf{z}\right)\cdot\left(\mathbf{a}_{j\boldsymbol{\lambda}},\mathbf{b}_{j\boldsymbol{\lambda}}\right)\cdot x_{j\boldsymbol{\lambda}}\\
 & =\nabla_{\mathbf{y,z}}\widetilde{D}\left(\mathbf{y},\mathbf{z}\right)\cdot\left(\sum_{j=1}^{M}\sum_{\boldsymbol{\lambda}\in\Lambda}\left(\mathbf{a}_{j\boldsymbol{\lambda}},\mathbf{b}_{j\boldsymbol{\lambda}}\right)\cdot x_{j\boldsymbol{\lambda}}\right)\\
 & =\nabla_{\mathbf{y,z}}\widetilde{D}\left(\mathbf{y},\mathbf{z}\right)\cdot\left(\mathbf{y},\mathbf{z}\right)
\end{align*}
where the second step follows from Lemma \ref{lemma:viBound_gradientRelationship1},
the last step follows from the definition of vectors $\mathbf{y},\mathbf{z},\mathbf{a}_{j\boldsymbol{\lambda}},\mathbf{b}_{j\boldsymbol{\lambda}}$.
\end{proof}

We introduce some additional notations that are used in the proof
below. Recall that the values of vectors $\mathbf{x}$, $\mathbf{y}$,
and $\mathbf{z}$ may vary over time due to service arrivals and departures.
Let $\mathbf{x}_{j}^{(j)}$, $\mathbf{y}_{j}^{(j)}$, and $\mathbf{z}_{j}^{(j)}$
respectively denote the values of $\mathbf{x}$, $\mathbf{y}$, and
$\mathbf{z}$ \emph{immediately after} instance $j$ is placed; and
let $\mathbf{x}_{j-1}^{(j)}$, $\mathbf{y}_{j-1}^{(j)}$, and $\mathbf{z}_{j-1}^{(j)}$
respectively denote the values of $\mathbf{x}$, $\mathbf{y}$, and
$\mathbf{z}$ \emph{immediately before} instance $j$ is placed. We
note that the values of $\mathbf{x}$, $\mathbf{y}$, and $\mathbf{z}$
may change after placing each instance. Therefore, the notions of
``before'', ``after'', and ``time'' (used below) here correspond to
the sequence of service instance placement, instead of the actual
physical time.

We then introduce vectors that only consider the placement up to the
$j$th service instance, which are necessary because the proof below
uses an iterative approach. Let $\mathbf{x}_{j}$, $\mathbf{y}_{j}$,
and $\mathbf{z}_{j}$ respectively denote the values of $\mathbf{x}$,
$\mathbf{y}$, and $\mathbf{z}$ at \emph{any time after} placing
instance $j$ (where instance $j$ can be either still running in
the system or already departed) while \emph{ignoring the placement
of any subsequent instances} $j'>j$ (if any). This means, in vector
$\mathbf{x}_{j}$, we set $\left(\mathbf{x}_{j}\right)_{i\boldsymbol{\lambda}}\triangleq x_{i\boldsymbol{\lambda}}$
for any $i\leq j$ and $\boldsymbol{\lambda}$, and set $\left(\mathbf{x}_{j}\right)_{i\boldsymbol{\lambda}}\triangleq0$
for any $i>j$ and $\boldsymbol{\lambda}$, although the value of
$x_{i\boldsymbol{\lambda}}$ at the current time of interest may be
non-zero for some $i>j$ and $\boldsymbol{\lambda}$. Similarly, in
vectors $\mathbf{y}_{j}$ and $\mathbf{z}_{j}$, we only consider
the resource consumptions up to instance $j$, i.e., $\left(\mathbf{y}_{j}\right)_{kt}\triangleq\sum_{i=1}^{j}\sum_{\boldsymbol{\lambda}\in\Lambda}a_{i\boldsymbol{\lambda}k}(t)x_{i\boldsymbol{\lambda}}$
and $\left(\mathbf{z}_{j}\right)_{klt}\triangleq\sum_{i=1}^{j}\sum_{\boldsymbol{\lambda}\in\Lambda}b_{i\boldsymbol{\lambda}kl}(t)x_{i\boldsymbol{\lambda}}$
for any $k$, $l$, and $t$.

We assume that the last service instance that has arrived before the
current time of interest has index $M$, thus $\mathbf{x}=\mathbf{x}_{M}$,
$\mathbf{y}=\mathbf{y}_{M}$, and $\mathbf{z}=\mathbf{z}_{M}$.

Because an instance will never come back after it has departed (even
if an instance of the same type comes back, it will be given a new
index), we have $\mathbf{y}{}_{j-1}\leq\mathbf{y}_{j-1}^{(j)}$ and
$\mathbf{z}{}_{j-1}\leq\mathbf{z}_{j-1}^{(j)}$, where the inequalities
are defined element-wise for the vector. 

Define $v_{j}\triangleq\widetilde{D}\left(\mathbf{y}_{j}^{(j)},\mathbf{z}_{j}^{(j)}\right)-\widetilde{D}\left(\mathbf{y}_{j-1}^{(j)},\mathbf{z}_{j-1}^{(j)}\right)$
to denote the increase in the sum cost $\widetilde{D}(\mathbf{y},\mathbf{z})$
(or, equivalently, $\widetilde{D}(\mathbf{x})$) at the time when
placing service $j$. Note that after this placement, the value of
$\widetilde{D}\left(\mathbf{y}_{j},\mathbf{z}_{j}\right)-\widetilde{D}\left(\mathbf{y}_{j-1},\mathbf{z}_{j-1}\right)$
may vary over time, because some services $i\leq j$ may leave the
system, but the value of $v_{j}$ is only taken when service $j$
is placed upon its arrival.

\begin{lemma}\label{lemma:viBound2}When Assumption \ref{condition:costFunc}
is satisfied, for any $M$, we have
\begin{equation}
\widetilde{D}\left(\mathbf{x}_{M}\right)\leq\sum_{j=1}^{M}v_{j}
\end{equation}
\end{lemma}\begin{proof} Assume that service $j$ takes configuration
$\boldsymbol{\lambda}_{0}$ after its placement (and before it possibly
unpredictably departs), then $\mathbf{y}_{j}^{(j)}-\mathbf{y}_{j-1}^{(j)}=\mathbf{a}{}_{j\boldsymbol{\lambda}_{0}}$
and $\mathbf{z}_{j}^{(j)}-\mathbf{z}_{j-1}^{(j)}=\mathbf{b}{}_{j\boldsymbol{\lambda}_{0}}$.
For any time after placing instance $j$ we define $\Delta\mathbf{y}_{j}\triangleq\mathbf{y}{}_{j}-\mathbf{y}{}_{j-1}$
and $\Delta\mathbf{z}_{j}\triangleq\mathbf{z}{}_{j}-\mathbf{z}{}_{j-1}$.
We always have $\Delta\mathbf{y}_{j}=\mathbf{a}{}_{j\boldsymbol{\lambda}_{0}}$,
$\Delta\mathbf{z}_{j}=\mathbf{b}{}_{j\boldsymbol{\lambda}_{0}}$,
if instance $j$ has not yet departed from the system, and $\Delta\mathbf{y}_{j}=\Delta\mathbf{z}_{j}=0$
if $j$ has already departed from the system.

Noting that $\widetilde{D}\left(\mathbf{y}_{j},\mathbf{z}_{j}\right)$
is convex non-decreasing (from Lemma \ref{lemma:convexityOfH}), we
have
\begin{align}
 & \widetilde{D}\left(\mathbf{y}{}_{j},\mathbf{z}{}_{j}\right)-\widetilde{D}\left(\mathbf{y}{}_{j-1},\mathbf{z}{}_{j-1}\right)\nonumber \\
 & =\widetilde{D}\left(\mathbf{y}{}_{j-1}+\Delta\mathbf{y}_{j},\mathbf{z}{}_{j-1}+\Delta\mathbf{z}_{j}\right)-\widetilde{D}\left(\mathbf{y}{}_{j-1},\mathbf{z}{}_{j-1}\right)\nonumber \\
 & \leq\widetilde{D}\left(\mathbf{y}{}_{j-1}^{(j)}+\Delta\mathbf{y}_{j},\mathbf{z}{}_{j-1}^{(j)}+\Delta\mathbf{z}_{j}\right)-\widetilde{D}\left(\mathbf{y}{}_{j-1}^{(j)},\mathbf{z}{}_{j-1}^{(j)}\right)\label{eq:proofVi_Ineq1}\\
 & \leq\widetilde{D}\left(\mathbf{y}{}_{j-1}^{(j)}+\mathbf{a}{}_{j\boldsymbol{\lambda}_{0}},\mathbf{z}{}_{j-1}^{(j)}+\mathbf{b}{}_{j\boldsymbol{\lambda}_{0}}\right)-\widetilde{D}\left(\mathbf{y}{}_{j-1}^{0},\mathbf{z}{}_{j-1}^{0}\right)\label{eq:proofVi_Ineq2}\\
 & =\widetilde{D}\left(\mathbf{y}_{j}^{(j)},\mathbf{z}_{j}^{(j)}\right)-\widetilde{D}\left(\mathbf{y}_{j-1}^{(j)},\mathbf{z}_{j-1}^{(j)}\right)\nonumber \\
 & =v_{j}\nonumber 
\end{align}
where inequality (\ref{eq:proofVi_Ineq1}) is because $\mathbf{y}{}_{j-1}\leq\mathbf{y}_{j-1}^{(j)}$,
$\mathbf{z}{}_{j-1}\leq\mathbf{z}_{j-1}^{(j)}$ (see discussion above)
and due to the convex non-decreasing property of $\widetilde{D}\left(\mathbf{y}_{j},\mathbf{z}_{j}\right)$;
inequality (\ref{eq:proofVi_Ineq2}) is because $\Delta\mathbf{y}_{j}\leq\mathbf{a}{}_{j\boldsymbol{\lambda}_{0}}$,
$\Delta\mathbf{z}_{j}\leq\mathbf{b}{}_{j\boldsymbol{\lambda}_{0}}$
and also due to the non-decreasing property of $\widetilde{D}\left(\mathbf{y}_{j},\mathbf{z}_{j}\right)$. 

We now note that $\widetilde{D}\left(\mathbf{x}_{0}\right)=0$, where
$\mathbf{x}_{0}=\mathbf{0}$ and $\mathbf{0}$ is defined as a vector
with all zeros, thus $\mathbf{y}_{0}=\mathbf{z}_{0}=\mathbf{0}$.
We have 
\begin{align*}
\sum_{j=1}^{M}v_{j} & \geq\sum_{j=1}^{M}\left[\widetilde{D}\left(\mathbf{y}{}_{j},\mathbf{z}{}_{j}\right)-\widetilde{D}\left(\mathbf{y}{}_{j-1},\mathbf{z}{}_{j-1}\right)\right]\\
 & =\widetilde{D}\left(\mathbf{x}_{M}\right)-\widetilde{D}\left(\mathbf{x}_{0}\right)=\widetilde{D}\left(\mathbf{x}_{M}\right)
\end{align*}
\end{proof}

\begin{lemma}\label{lemma:viBound1}When Assumption \ref{condition:costFunc}
is satisfied, for any $j$ and $\boldsymbol{\lambda}$, we have
\begin{equation}
v_{j}\leq\phi\frac{\partial\widetilde{D}}{\partial x_{j\boldsymbol{\lambda}}}\left(\mathbf{x}_{M}\right)
\end{equation}
where $\phi$ is a constant satisfying (\ref{eq:alphaDef}). \end{lemma}\begin{proof}Assume
that service $j$ takes configuration $\boldsymbol{\lambda}_{0}$
after its placement (and before it possibly unpredictably departs).
Because we perform a greedy assignment in Algorithm \ref{alg:onlineHighLevel},
we have
\begin{align*}
v_{j} & =\widetilde{D}\left(\mathbf{y}_{j}^{(j)},\mathbf{z}_{j}^{(j)}\right)-\widetilde{D}\left(\mathbf{y}_{j-1}^{(j)},\mathbf{z}_{j-1}^{(j)}\right)\\
 & =\widetilde{D}\left(\mathbf{y}_{j-1}^{(j)}+\mathbf{a}_{j\boldsymbol{\lambda}_{0}},\mathbf{z}_{j-1}^{(j)}+\mathbf{b}_{j\boldsymbol{\lambda}_{0}}\right)-\widetilde{D}\left(\mathbf{y}_{j-1}^{(j)},\mathbf{z}_{j-1}^{(j)}\right)\\
 & \leq\widetilde{D}\left(\mathbf{y}_{j-1}^{(j)}+\mathbf{a}_{j\boldsymbol{\lambda}},\mathbf{z}_{j-1}^{(j)}+\mathbf{b}_{j\boldsymbol{\lambda}}\right)-\widetilde{D}\left(\mathbf{y}_{j-1}^{(j)},\mathbf{z}_{j-1}^{(j)}\right)
\end{align*}
for any $\boldsymbol{\lambda}\in\Lambda_{i}$.

Then, we have
\begin{align}
 & \widetilde{D}\left(\mathbf{y}_{j-1}^{(j)}+\mathbf{a}_{j\boldsymbol{\lambda}},\mathbf{z}_{j-1}^{(j)}+\mathbf{b}_{j\boldsymbol{\lambda}}\right)-\widetilde{D}\left(\mathbf{y}_{j-1}^{(j)},\mathbf{z}_{j-1}^{(j)}\right)\nonumber \\
 & \leq\nabla_{\mathbf{y,z}}\widetilde{D}\left(\mathbf{y}_{j-1}^{(j)}+\mathbf{a}_{j\boldsymbol{\lambda}},\mathbf{z}_{j-1}^{(j)}+\mathbf{b}_{j\boldsymbol{\lambda}}\right)\cdot\left(\mathbf{a}_{j\boldsymbol{\lambda}},\mathbf{b}_{j\boldsymbol{\lambda}}\right)\label{eq:proofVUpBound1}\\
 & \leq\nabla_{\mathbf{y,z}}\widetilde{D}\left(\mathbf{y}_{\textrm{max}}+\mathbf{a}_{j\boldsymbol{\lambda}},\mathbf{z}_{\textrm{max}}+\mathbf{b}_{j\boldsymbol{\lambda}}\right)\cdot\left(\mathbf{a}_{j\lambda},\mathbf{b}_{j\boldsymbol{\lambda}}\right)\label{eq:proofVUpBound2}\\
 & \leq\phi\nabla_{\mathbf{y,z}}\widetilde{D}\left(\mathbf{y}_{M},\mathbf{z}_{M}\right)\cdot\left(\mathbf{a}_{j\boldsymbol{\lambda}},\mathbf{b}_{j\boldsymbol{\lambda}}\right)\label{eq:proofVUpBound3}\\
 & =\phi\frac{\partial\widetilde{D}}{\partial x_{j\boldsymbol{\lambda}}}\left(\mathbf{x}_{M}\right)\label{eq:proofVUpBound4}
\end{align}
where ``$\cdot$'' denotes the dot-product. The above relationship
is explained as follows. Inequality (\ref{eq:proofVUpBound1}) follows
from the first-order conditions of convex functions \cite[Section 3.1.3]{boyd2004convex}.
The definition of $\mathbf{y}_{\textrm{max}}$ and $\mathbf{z}_{\textrm{max}}$
in Proposition \ref{prop:performanceGapResult} gives (\ref{eq:proofVUpBound2}).
The definition of $\phi$ in (\ref{eq:alphaDef}) gives (\ref{eq:proofVUpBound3}).
Equality (\ref{eq:proofVUpBound4}) follows from Lemma \ref{lemma:viBound_gradientRelationship1}.
This completes the proof. \end{proof}

\vspace{0.2in}

Using the above lemmas, we now proof Proposition \ref{prop:performanceGapResult}.

\vspace{0.2in}

\begin{proof} \textbf{(Proposition \ref{prop:performanceGapResult})}
Due to the convexity of $\widetilde{D}(\mathbf{x})$, from the first-order
conditions of convex functions \cite[Section 3.1.3]{boyd2004convex},
we have
\begin{align}
 & \widetilde{D}(\phi\psi\mathbf{x}_{M}^{*})-\widetilde{D}(\mathbf{x}_{M})\nonumber \\
 & \geq\nabla_{\mathbf{x}}\widetilde{D}\left(\mathbf{x}_{M}\right)\cdot\left(\phi\psi\mathbf{x}_{M}^{*}-\mathbf{x}_{M}\right)\\
 & =\phi\psi\nabla_{\mathbf{x}}\widetilde{D}\left(\mathbf{x}_{M}\right)\cdot\mathbf{x}_{M}^{*}-\nabla_{\mathbf{x}}\widetilde{D}\left(\mathbf{x}_{M}\right)\mathbf{x}_{M}\\
 & =\sum_{i=1}^{M}\sum_{\boldsymbol{\lambda}\in\Lambda}\phi\psi x_{i\boldsymbol{\lambda}}^{*}\frac{\partial\widetilde{D}}{\partial x_{i\boldsymbol{\lambda}}}\left(\mathbf{x}_{M}\right)-\nabla_{\mathbf{x}}\widetilde{D}\left(\mathbf{x}_{M}\right)\cdot\mathbf{x}_{M}\\
 & =\psi\left(\sum_{i=1}^{M}\sum_{\boldsymbol{\lambda}\in\Lambda}x_{i\boldsymbol{\lambda}}^{*}\phi\frac{\partial\widetilde{D}}{\partial x_{i\boldsymbol{\lambda}}}\left(\mathbf{x}_{M}\right)-\frac{\nabla_{\mathbf{x}}\widetilde{D}\left(\mathbf{x}_{M}\right)\cdot\mathbf{x}_{M}}{\psi}\right)\label{eq:boundProofRHS1}
\end{align}
where $x_{i\boldsymbol{\lambda}}^{*}$ is the $(i,\boldsymbol{\lambda})$th
element of vector $\mathbf{x}_{M}^{*}$. From Lemma \ref{lemma:viBound1},
we have

\begin{align}
\textrm{Eq. (\ref{eq:boundProofRHS1})} & \geq\psi\left(\sum_{i=1}^{M}\sum_{\boldsymbol{\lambda}\in\Lambda}x_{i\boldsymbol{\lambda}}^{*}v_{i}-\frac{\nabla_{\mathbf{x}}\widetilde{D}\left(\mathbf{x}_{M}\right)\cdot\mathbf{x}_{M}}{\psi}\right)\\
 & =\psi\left(\sum_{i=1}^{M}v_{i}\sum_{\boldsymbol{\lambda}\in\Lambda}x_{i\boldsymbol{\lambda}}^{*}-\frac{\nabla_{\mathbf{x}}\widetilde{D}\left(\mathbf{x}_{M}\right)\cdot\mathbf{x}_{M}}{\psi}\right)\label{eq:boundProofRHS2}
\end{align}
From the constraint $\sum_{\boldsymbol{\lambda}\in\Lambda}x_{i\boldsymbol{\lambda}}^{*}=1$
and the definition of $\psi$, we get
\begin{align}
\textrm{Eq. (\ref{eq:boundProofRHS2})} & =\psi\left(\sum_{i=1}^{M}v_{i}-\frac{\nabla_{\mathbf{x}}\widetilde{D}\left(\mathbf{x}_{M}\right)\cdot\mathbf{x}_{M}}{\psi}\right)\\
 & \geq\psi\left(\sum_{i=1}^{M}v_{i}-\widetilde{D}(\mathbf{x}_{M})\right)\label{eq:boundProofRHS3}\\
 & \geq0
\end{align}
where the last equality follows from Lemma \ref{lemma:viBound2}.
This gives (\ref{eq:performGap_x}). 

Equation (\ref{eq:performGap_yz}) follows from the fact that $y_{k,j}(t)$
and $z_{kl,j}(t)$ are both linear in $x_{i\lambda}$. 

The last equality in (\ref{eq:betaDef}) follows from Lemma \ref{lemma:viBound_gradientRelationship2}
and the fact that $\widetilde{D}\left(\mathbf{x}\right)=\widetilde{D}\left(\mathbf{y},\mathbf{z}\right)$
as well as $\mathbf{x}=\mathbf{x}_{M}$, $\mathbf{y}=\mathbf{y}_{M}$,
and $\mathbf{z}=\mathbf{z}_{M}$.\end{proof}

\section{Proof of Proposition \ref{prop:polyCostCompetitive}}

\label{sec:polyCostCompetProof}

\begin{lemma} \label{lemma:limit} For polynomial functions $\Xi_{1}(y)$
and $\Xi_{2}(y)$ in the general form:
\begin{align*}
\Xi_{1}(y) & \triangleq\sum_{\rho=0}^{\Omega}\omega_{1}^{(\rho)}y^{\rho}\\
\Xi_{2}(y) & \triangleq\sum_{\rho=0}^{\Omega}\omega_{2}^{(\rho)}y^{\rho}
\end{align*}
where the constants $\omega_{1}^{(\rho)}\geq0$ and $\omega_{2}^{(\rho)}\geq0$
for $0\leq\rho<\Omega$, while $\omega_{1}^{(\Omega)}>0$ and $\omega_{2}^{(\Omega)}>0$,
we have
\[
\lim_{y\rightarrow+\infty}\frac{\Xi_{1}(y)}{\Xi_{2}(y)}=\frac{\omega_{1}^{(\Omega)}}{\omega_{2}^{(\Omega)}}
\]
 \end{lemma} \begin{proof} When $\Omega=0$, we have 
\[
\lim_{y\rightarrow+\infty}\frac{\Xi_{1}(y)}{\Xi_{2}(y)}=\frac{\omega_{1}^{(0)}}{\omega_{2}^{(0)}}
\]
When $\Omega>0$, we note that $\lim_{y\rightarrow+\infty}\Xi_{1}(y)=+\infty$
and $\lim_{y\rightarrow+\infty}\Xi_{2}(y)=+\infty$, because $\omega_{1}^{(\Omega)}>0$
and $\omega_{2}^{(\Omega)}>0$. We apply the L'Hospital's rule and
get
\begin{align}
\lim_{y\rightarrow+\infty}\frac{\Xi_{1}(y)}{\Xi_{2}(y)} & =\lim_{y\rightarrow+\infty}\frac{\frac{d\Xi_{1}(y)}{dy}}{\frac{d\Xi_{2}(y)}{dy}}\nonumber \\
 & =\lim_{y\rightarrow+\infty}\frac{\sum_{\rho=1}^{\Omega}\rho\omega_{1}^{(\rho)}y^{\rho-1}}{\sum_{\rho=1}^{\Omega}\rho\omega_{2}^{(\rho)}y^{\rho-1}}\label{eq:limitProof1}
\end{align}
Suppose we have 
\begin{align}
\lim_{y\rightarrow+\infty}\frac{\Xi_{1}(y)}{\Xi_{2}(y)} & =\lim_{y\rightarrow+\infty}\frac{\sum_{\rho=n}^{\Omega}\left(\prod_{m=0}^{n-1}(\rho-m)\right)\omega_{1}^{(\rho)}y^{\rho-n}}{\sum_{\rho=n}^{\Omega}\left(\prod_{m=0}^{n-1}(\rho-m)\right)\omega_{2}^{(\rho)}y^{\rho-n}}\label{eq:limitProof2}
\end{align}
which equals to (\ref{eq:limitProof1}) for $n=1$. For $1\leq n<\Omega$,
we note that $\Omega-n>0$, hence the numerator and denominator in
the right hand-side (RHS) of (\ref{eq:limitProof2}) still respectively
approach $+\infty$ when $y\rightarrow+\infty$ (because $\omega_{1}^{(\Omega)}>0$
and $\omega_{2}^{(\Omega)}>0$). Let $\Psi(n)$ denote the RHS (\ref{eq:limitProof2}),
we can reapply the L'Hospital's rule on $\Psi(n)$, yielding 
\begin{align*}
\Psi(n) & =\lim_{y\rightarrow+\infty}\frac{\sum_{\rho=n+1}^{\Omega}\left(\prod_{m=0}^{(n+1)-1}(\rho-m)\right)\omega_{1}^{(\rho)}y^{\rho-(n+1)}}{\sum_{\rho=n+1}^{\Omega}\left(\prod_{m=0}^{(n+1)-1}(\rho-m)\right)\omega_{2}^{(\rho)}y^{\rho-(n+1)}}\\
 & =\Psi(n+1)
\end{align*}
which proofs that (\ref{eq:limitProof2}) holds for $1\leq n\leq\Omega$.
Therefore, 
\begin{align*}
\lim_{y\rightarrow+\infty}\frac{\Xi_{1}(y)}{\Xi_{2}(y)} & =\Psi(\Omega)=\frac{\rho!\omega_{1}^{(\Omega)}}{\rho!\omega_{2}^{(\Omega)}}=\frac{\omega_{1}^{(\Omega)}}{\omega_{2}^{(\Omega)}}
\end{align*}
\end{proof} 

\begin{lemma} \label{lemma:proofPolyCostMonotonicityDiff} For variables
$0\leq y\leq y',0\leq y_{k}\leq y'_{k},0\leq y_{l}\leq y'_{l},0\leq z_{kl}\leq z'_{kl}$,
we always have 
\begin{align}
\frac{du_{k,t}}{dy}(y) & \leq\frac{du_{k,t}}{dy}(y')\label{eq:polyCostProof_MonoDiff1}\\
\frac{\partial w_{kl,t}}{\partial\Upsilon}(y_{k},y_{l},z_{kl}) & \leq\frac{\partial w_{kl,t}}{\partial\Upsilon}(y'_{k},y'_{l},z'_{kl})\label{eq:polyCostProof_MonoDiff2}
\end{align}
where $\Upsilon$ stands for either $y_{k}$, $y_{l}$, or $z_{kl}$.
\end{lemma} \begin{proof} We note that 
\[
\frac{du_{k,t}}{dy}(y)=\sum_{\rho}\rho\gamma_{k,t}^{(\rho)}y^{\rho-1}
\]
from which (\ref{eq:polyCostProof_MonoDiff1}) follows directly because
$\gamma_{k,t}^{(\rho)}\geq0$. We then note that
\[
\frac{\partial w_{kl,t}}{\partial y_{k}}(y_{k},y_{l},z_{kl})=\sum_{\rho_{1}}\sum_{\rho_{2}}\sum_{\rho_{3}}\rho_{1}\kappa_{kl,t}^{(\rho_{1},\rho_{2},\rho_{3})}y_{k}^{\rho_{1}-1}y_{l}^{\rho_{2}}z_{kl}^{\rho_{3}}
\]
from which (\ref{eq:polyCostProof_MonoDiff2}) follows for $\Upsilon=y_{k}$
because $\kappa_{kl,t}^{(\rho_{1},\rho_{2},\rho_{3})}\geq0$. Similarly,
(\ref{eq:polyCostProof_MonoDiff2}) also follows for $\Upsilon=y_{l}$
and $\Upsilon=z_{kl}$. \end{proof} 

\begin{lemma} \label{lemma:proofPolyCostBound} Let $\Omega$ denote
the maximum value of $\rho$ such that \emph{either} $\gamma_{k,t}^{(\rho)}>0$
\emph{or} $\kappa_{kl,t}^{(\rho_{1},\rho_{2},\rho_{3})}>0$, where
$\rho_{1}+\rho_{2}+\rho_{3}=\rho$. Assume that the cost functions
are defined as in (\ref{eq:polyLocalCost}) and (\ref{eq:polyMigCost}),
then for any constants $\delta>0$, $B\geq0$, there exist sufficiently
large values of $y,y_{k},y_{l},z_{kl}$, such that
\begin{align}
\frac{\frac{du_{k,t}}{dy}(y+B)}{\frac{du_{k,t}}{dy}(y)} & \leq1+\delta\label{eq:polyCostProof_boundResult1}\\
\frac{\frac{du_{k,t}}{dy}(y)\cdot y}{u_{k,t}(y)} & \leq\Omega+\delta\label{eq:polyCostProof_boundResult2}\\
\frac{\frac{\partial w_{kl,t}}{\partial\Upsilon}(y_{k}+B,y_{l}+B,z_{kl}+B)}{\frac{\partial w_{kl,t}}{\partial\Upsilon}(y_{k},y_{l},z_{kl})} & \leq1+\delta\label{eq:polyCostProof_boundResult3}\\
\frac{\frac{\partial w_{kl,t}}{\partial\Upsilon}(y_{k},y_{l},z_{kl})\cdot\Upsilon}{w_{k,t}(y_{k},y_{l},z_{kl})} & \leq\Omega+\delta\label{eq:polyCostProof_boundResult4}
\end{align}
for any $k,l,t$, where $\Upsilon$ stands for either $y_{k}$, $y_{l}$,
or $z_{kl}$.\end{lemma} \begin{proof} Let $\Omega'$ denote the
maximum value of $\rho$ such that $\gamma_{k,t}^{(\rho)}>0$, we
always have $\Omega'\leq\Omega$. We note that
\begin{equation}
\frac{\frac{du_{k,t}}{dy}(y+B)}{\frac{du_{k,t}}{dy}(y)}=\frac{\sum_{\rho=1}^{\Omega'}\rho\gamma_{k,t}^{(\rho)}\left(y+B\right)^{\rho-1}}{\sum_{\rho=1}^{\Omega'}\rho\gamma_{k,t}^{(\rho)}y^{\rho-1}}\label{eq:polyCostProof_ratio1}
\end{equation}
\begin{equation}
\frac{\frac{du_{k,t}}{dy}(y)\cdot y}{u_{k,t}(y)}=\frac{\sum_{\rho=1}^{\Omega'}\rho\gamma_{k,t}^{(\rho)}y^{\rho}}{\sum_{\rho=1}^{\Omega'}\gamma_{k,t}^{(\rho)}y^{\rho}}
\end{equation}
According to Lemma \ref{lemma:limit}, we have 
\begin{align}
\lim_{y\rightarrow+\infty}\frac{\frac{du_{k,t}}{dy}(y+B)}{\frac{du_{k,t}}{dy}(y)} & =\frac{\Omega'\gamma_{k,t}^{(\Omega')}}{\Omega'\gamma_{k,t}^{(\Omega')}}=1\label{eq:polyCostProof_limit1}\\
\lim_{y\rightarrow+\infty}\frac{\frac{du_{k,t}}{dy}(y)\cdot y}{u_{k,t}(y)} & =\frac{\Omega'\gamma_{k,t}^{(\Omega')}}{\gamma_{k,t}^{(\Omega')}}=\Omega'\label{eq:polyCostProof_limit2}
\end{align}
where we note that after expanding the numerator in the RHS of (\ref{eq:polyCostProof_ratio1}),
the constant $B$ does not appear in the coefficient of $y^{\Omega'-1}$.

Now, define a variable $q>0$, and we let $y_{k}=\zeta_{1}q,y_{l}=\zeta_{2}q,z_{kl}=\zeta_{3}q$,
where $\zeta_{1},\zeta_{2},\zeta_{3}>0$ are arbitrary constants.
Using $\zeta_{1},\zeta_{2},\zeta_{3}$, and $q$, we can represent
any value of $\left(y_{k},y_{l},z_{kl}\right)>\mathbf{0}$. With this
definition, we have
\begin{align}
w_{kl,t}\left(q\right) & \triangleq w_{kl,t}\left(\zeta_{1}q,\zeta_{2}q,\zeta_{3}q\right)\nonumber \\
 & =\sum_{\rho_{1}}\sum_{\rho_{2}}\sum_{\rho_{3}}\kappa_{kl,t}^{(\rho_{1},\rho_{2},\rho_{3})}\zeta_{1}^{\rho_{1}}\zeta_{2}^{\rho_{2}}\zeta_{3}^{\rho_{3}}q^{\rho_{1}+\rho_{2}+\rho_{3}}\nonumber \\
 & =\sum_{\rho=1}^{\Omega''}(\kappa')_{kl,t}^{(\rho)}q^{\rho}\label{eq:polyCostProof_redefine_w}
\end{align}
where the constant 
\[
(\kappa'){}_{kl,t}^{(\rho)}\triangleq\sum_{\{(\rho_{1},\rho_{2},\rho_{3}):\rho_{1}+\rho_{2}+\rho_{3}=\rho\}}\kappa_{kl,t}^{(\rho_{1},\rho_{2},\rho_{3})}\zeta_{1}^{\rho_{1}}\zeta_{2}^{\rho_{2}}\zeta_{3}^{\rho_{3}}
\]
and $\Omega''$ is defined as the maximum value of $\rho$ such that
$(\kappa'){}_{kl,t}^{(\rho)}>0$, we always have $\Omega''\leq\Omega$.
Note that (\ref{eq:polyCostProof_redefine_w}) is in the same form
as (\ref{eq:polyLocalCost}). Following the same procedure as for
obtaining (\ref{eq:polyCostProof_limit1}) and (\ref{eq:polyCostProof_limit2}),
we get
\begin{align}
\lim_{q\rightarrow+\infty}\frac{\frac{dw_{kl,t}}{dq}(q+B')}{\frac{dw_{kl,t}}{dq}(q)} & =\frac{\Omega''\gamma_{k,t}^{(\Omega'')}}{\Omega''\gamma_{k,t}^{(\Omega'')}}=1\label{eq:polyCostProof_limit3}\\
\lim_{q\rightarrow+\infty}\frac{\frac{dw_{kl,t}}{dq}(q)\cdot q}{w_{k,t}(q)} & =\frac{\Omega''\gamma_{k,t}^{(\Omega'')}}{\gamma_{k,t}^{(\Omega'')}}=\Omega''\label{eq:polyCostProof_limit4}
\end{align}
where $B'\triangleq\frac{B}{\min\left\{ \zeta_{1};\zeta_{2};\zeta_{3}\right\} }$.

According to the definition of limits, for any $\delta>0$, there
exist sufficiently large values of $y$ and $q$ (thus $y_{k},y_{l},z_{kl}$),
such that
\begin{align}
\frac{\frac{du_{k,t}}{dy}(y+B)}{\frac{du_{k,t}}{dy}(y)} & \leq1+\delta\label{eq:polyCostProof_bound_u1}\\
\frac{\frac{du_{k,t}}{dy}(y)\cdot y}{u_{k,t}(y)} & \leq\Omega'+\delta\label{eq:polyCostProof_bound_u2}\\
\frac{\frac{dw_{kl,t}}{dq}(q+B')}{\frac{dw_{kl,t}}{dq}(q)} & \leq1+\delta\label{eq:polyCostProof_bound1}\\
\frac{\frac{dw_{kl,t}}{dq}(q)\cdot q}{w_{k,t}(q)} & \leq\Omega''+\delta\label{eq:polyCostProof_bound2}
\end{align}
for any $k,l,t$.

Because 
\[
\frac{\frac{dw_{kl,t}}{dq}(q+B')}{\frac{dw_{kl,t}}{dq}(q)}=\frac{\frac{dw_{kl,t}}{d(\zeta q)}(q+B')}{\frac{dw_{kl,t}}{d(\zeta q)}(q)}
\]
\[
\frac{dw_{kl,t}}{dq}(q)\cdot q=\frac{dw_{kl,t}}{d(\zeta q)}(q)\cdot\zeta q
\]
for any $\zeta>0$, we can also express the bounds (\ref{eq:polyCostProof_bound1})
and (\ref{eq:polyCostProof_bound2}) in terms of $y_{k},y_{l},z_{kl}$,
yielding 
\begin{align}
\frac{\frac{\partial w_{kl,t}}{\partial\Upsilon}(y_{k}+\zeta_{1}B',y_{l}+\zeta_{2}B',z_{kl}+\zeta_{3}B')}{\frac{\partial w_{kl,t}}{\partial\Upsilon}(y_{k},y_{l},z_{kl})} & \leq1+\delta\label{eq:polyCostProof_bound1-1}\\
\frac{\frac{\partial w_{kl,t}}{\partial\Upsilon}(y_{k},y_{l},z_{kl})\cdot\Upsilon}{w_{k,t}(y_{k},y_{l},z_{kl})} & \leq\Omega''+\delta\label{eq:polyCostProof_bound2-1}
\end{align}
where $\Upsilon$ stands for either $y_{k}$, $y_{l}$, or $z_{kl}$.
According to the definition of $B'$, we have $B\leq\zeta_{1}B'$,
$B\leq\zeta_{2}B'$, $B\leq\zeta_{3}B'$. From Lemma \ref{lemma:proofPolyCostMonotonicityDiff},
we have
\begin{align}
 & \frac{\partial w_{kl,t}}{\partial\Upsilon}(y_{k}+B,y_{l}+B,z_{kl}+B)\nonumber \\
 & \leq\frac{\partial w_{kl,t}}{\partial\Upsilon}(y_{k}+\zeta_{1}B',y_{l}+\zeta_{2}B',z_{kl}+\zeta_{3}B')\label{eq:polyCostProof_bound1-1-1}
\end{align}

Combining (\ref{eq:polyCostProof_bound1-1-1}) with (\ref{eq:polyCostProof_bound1-1})
and noting that $\Omega'\leq\Omega$ and $\Omega''\leq\Omega$, together
with (\ref{eq:polyCostProof_bound_u1}), (\ref{eq:polyCostProof_bound_u2}),
and (\ref{eq:polyCostProof_bound2-1}), we get (\ref{eq:polyCostProof_boundResult1})--(\ref{eq:polyCostProof_boundResult4}).
\end{proof}

\begin{lemma} \label{lemma:proofPolyCostRatioBound} For arbitrary
values $\vartheta_{1,n}\geq0$ and $\vartheta_{2,n}\geq0$ for all
$n=1,2,...,N$, where $\vartheta_{1,n}$ and $\vartheta_{2,n}$ are
either both zero or both non-zero and there exists $n$ such that
$\vartheta_{1,n}$ and $\vartheta_{2,n}$ are non-zero, if the following
bound is satisfied:
\[
\max_{\{n\in\{1,...,N\}:\vartheta_{1,n}\neq0,\vartheta_{2,n}\neq0\}}\frac{\vartheta_{1,n}}{\vartheta_{2,n}}\leq\varTheta
\]
then we have
\[
\frac{\sum_{n=1}^{N}\omega_{n}\vartheta_{1,n}}{\sum_{n=1}^{N}\omega_{n}\vartheta_{2,n}}\leq\varTheta
\]
for any $\omega_{n}\geq0$. \end{lemma} \begin{proof} Because $\vartheta_{1,n}\leq\varTheta\vartheta_{2,n}$
for all $n$, we have 
\[
\sum_{n=1}^{N}\omega_{n}\vartheta_{1,n}\leq\sum_{n=1}^{N}\omega_{n}\varTheta\vartheta_{2,n}
\]
yielding the result. \end{proof}

\begin{lemma} \label{lemma:existenceOfBoundInPolyCost} When Assumption
\ref{condition:boundedArrDept} is satisfied and the window size $T$
is a constant, there exists a constant $B\geq0$ such that 

\begin{equation}
\left(\mathbf{y}_{\textrm{max}}+\mathbf{a}_{i\boldsymbol{\lambda}},\mathbf{z}_{\textrm{max}}+\mathbf{b}_{i\boldsymbol{\lambda}}\right)-\left(\mathbf{y},\mathbf{z}\right)\leq B\mathbf{e}\label{eq:polyCostCompetLoadDifferenceCondition}
\end{equation}
for any $i$ and any $\boldsymbol{\lambda}\in\Lambda_{i}$, where
$\mathbf{e}\triangleq[1,...,1]$ is a vector of all ones that has
the same dimension as $\left(\mathbf{y},\mathbf{z}\right)$. \end{lemma}

\begin{proof} We note that 
\begin{align}
 & \left(\mathbf{y}_{\textrm{max}}+\mathbf{a}_{i\boldsymbol{\lambda}},\mathbf{z}_{\textrm{max}}+\mathbf{b}_{i\boldsymbol{\lambda}}\right)-\left(\mathbf{y},\mathbf{z}\right)\nonumber \\
 & \leq\left(\mathbf{y}_{\textrm{max}}+a_{\textrm{max}}\mathbf{e_{y}},\mathbf{z}_{\textrm{max}}+b_{\textrm{max}}\mathbf{e_{z}}\right)-\left(\mathbf{y},\mathbf{z}\right)\label{eq:proofConstBExist1}\\
 & \leq\left(a_{\textrm{max}}\left(B_{d}T+1\right)\mathbf{e_{y}},b_{\textrm{max}}\left(B_{d}T+1\right)\mathbf{e_{z}}\right)\label{eq:proofConstBExist2}\\
 & \leq\max\left\{ a_{\textrm{max}}\left(B_{d}T+1\right);b_{\textrm{max}}\left(B_{d}T+1\right)\right\} \cdot\mathbf{e}\label{eq:proofConstBExist3}
\end{align}
where $\mathbf{e_{y}}\triangleq[1,...,1]$ and $\mathbf{e_{z}}\triangleq[1,...,1]$
are vectors of all ones that respectively have the same dimensions
as $\mathbf{y}$ and $\mathbf{z}$. Inequality (\ref{eq:proofConstBExist1})
follows from the boundedness assumption in Assumption \ref{condition:boundedArrDept}.
Inequality (\ref{eq:proofConstBExist2}) follows by noting that the
gap between $\left(\mathbf{y}_{\textrm{max}},\mathbf{z}_{\textrm{max}}\right)$
and $\left(\mathbf{y},\mathbf{z}\right)$ is because of instances
unpredictably leaving the system before their maximum lifetime, and
that there are at most $T$ slots, at most $B_{d}$ instances unpredictably
leave the system in each slot (according to Assumption \ref{condition:boundedArrDept}).
Inequality (\ref{eq:proofConstBExist3}) is obvious (note that the
maximum is taken element-wise). 

By setting $B=\max\left\{ a_{\textrm{max}}\left(B_{d}T+1\right);b_{\textrm{max}}\left(B_{d}T+1\right)\right\} $,
we prove the result. \end{proof}

\vspace{0.2in}

We now proof Proposition \ref{prop:polyCostCompetitive}.

\vspace{0.2in}

\begin{proof} \textbf{(Proposition \ref{prop:polyCostCompetitive})}
We note that $\widetilde{D}\left(\mathbf{y},\mathbf{z}\right)$ sums
up $u_{k,t}(y_{k})$ and $w_{kl,t}\left(y_{k},y_{l},z_{kl}\right)$
over $t,k,l$, as defined in (\ref{eq:costDef_T_xi}). 

The numerator in the RHS of (\ref{eq:alphaDef}) can be expanded into
a sum containing terms of either 
\[
\frac{du_{k,t}}{dy}\left(\left(\mathbf{y}_{\textrm{max}}+\mathbf{a}_{i\boldsymbol{\lambda}}\right)_{kt}\right)
\]
 or 
\[
\frac{\partial w_{kl,t}}{\partial\Upsilon}\left(\left(\mathbf{y}_{\textrm{max}}+\mathbf{a}_{i\boldsymbol{\lambda}}\right)_{kt},\left(\mathbf{y}_{\textrm{max}}+\mathbf{a}_{i\boldsymbol{\lambda}}\right)_{lt},\left(\mathbf{z}_{\textrm{max}}+\mathbf{b}_{i\boldsymbol{\lambda}}\right)_{klt}\right)
\]
where $\Upsilon$ stands for either $y_{k}(t)$, $y_{l}(t)$, or $z_{kl}(t)$,
with either $a_{i\boldsymbol{\lambda}k}(t)$ or $b_{i\boldsymbol{\lambda}kl}(t)$
as weights. Because Assumption \ref{condition:boundedArrDept} is
satisfied, according to (\ref{eq:polyCostCompetLoadDifferenceCondition})
in Lemma \ref{lemma:existenceOfBoundInPolyCost}, we have
\begin{align*}
\left(\mathbf{y}_{\textrm{max}}+\mathbf{a}_{i\boldsymbol{\lambda}}\right)_{kt} & \leq y_{k}(t)+B\\
\left(\mathbf{z}_{\textrm{max}}+\mathbf{b}_{i\boldsymbol{\lambda}}\right)_{klt} & \leq z_{kl}(t)+B
\end{align*}
for all $k,l,t$. From Lemma \ref{lemma:proofPolyCostMonotonicityDiff},
we have
\begin{align*}
\frac{du_{k,t}}{dy}\left(y_{k}(t)+B\right) & \geq\frac{du_{k,t}}{dy}\left(\left(\mathbf{y}_{\textrm{max}}+\mathbf{a}_{i\boldsymbol{\lambda}}\right)_{kt}\right)
\end{align*}
and
\begin{align*}
 & \frac{\partial w_{kl,t}}{\partial\Upsilon}\left(y_{k}(t)+B,y_{l}(t)+B,z_{kl}(t)+B\right)\\
 & \geq\frac{\partial w_{kl,t}}{\partial\Upsilon}\left(\left(\mathbf{y}_{\textrm{max}}+\mathbf{a}_{i\boldsymbol{\lambda}}\right)_{kt},\left(\mathbf{y}_{\textrm{max}}+\mathbf{a}_{i\boldsymbol{\lambda}}\right)_{lt},\left(\mathbf{z}_{\textrm{max}}+\mathbf{b}_{i\boldsymbol{\lambda}}\right)_{klt}\right)
\end{align*}
Therefore, if
\begin{equation}
\phi\geq\frac{\nabla_{\mathbf{y,z}}\widetilde{D}\left(\left(\mathbf{y},\mathbf{z}\right)+B\mathbf{e}\right)\cdot\left(\mathbf{a}_{i\boldsymbol{\lambda}},\mathbf{b}_{i\boldsymbol{\lambda}}\right)}{\nabla_{\mathbf{y,z}}\widetilde{D}\left(\mathbf{y},\mathbf{z}\right)\cdot\left(\mathbf{a}_{i\boldsymbol{\lambda}},\mathbf{b}_{i\boldsymbol{\lambda}}\right)}\label{eq:polyCostProofalphaDefLoose}
\end{equation}
then (\ref{eq:alphaDef}) is always satisfied. Similarly to the above,
the numerator in the RHS of (\ref{eq:polyCostProofalphaDefLoose})
can be expanded into a sum containing terms of either $\frac{du_{k,t}}{dy}\left(y_{k}(t)+B\right)$
and $\frac{\partial w_{kl,t}}{\partial\Upsilon}\left(y_{k}(t)+B,y_{l}(t)+B,z_{kl}(t)+B\right)$
with either $a_{i\boldsymbol{\lambda}k}(t)$ or $b_{i\boldsymbol{\lambda}kl}(t)$
as weights.

Again, the denominator in the RHS of (\ref{eq:alphaDef}) (or equivalently,
(\ref{eq:polyCostProofalphaDefLoose})) can be expanded into a sum
containing terms of either $\frac{du_{k,t}}{dy}(y(t))$ or $\frac{\partial w_{kl,t}}{\partial\Upsilon}(y_{k}(t),y_{l}(t),z_{kl}(t))$,
with either $a_{i\boldsymbol{\lambda}k}(t)$ or $b_{i\boldsymbol{\lambda}kl}(t)$
as weights. 

For any given $i,\boldsymbol{\lambda}$, the terms $\frac{du_{k,t}}{dy}(y_{k}(t)+B)$
and $\frac{du_{k,t}}{dy}(y_{k}(t))$ have the same weight $a_{i\boldsymbol{\lambda}k}(t)$,
and $\frac{\partial w_{kl,t}}{\partial\Upsilon}(y_{k}(t)+B,y_{l}(t)+B,z_{kl}(t)+B)$
and $\frac{\partial w_{kl,t}}{\partial\Upsilon}(y_{k}(t),y_{l}(t),z_{kl}(t))$
have the same weight $b_{i\boldsymbol{\lambda}kl}(t)$. According
to Lemmas \ref{lemma:proofPolyCostBound} and \ref{lemma:proofPolyCostRatioBound},
for any $\delta>0$, there exist sufficiently large values of $\mathbf{y}$
and $\mathbf{z}$, such that 
\[
\textrm{RHS of (\ref{eq:polyCostProofalphaDefLoose})}\leq1+\delta
\]
Following a similar reasoning, we know that, for any $\delta>0$,
there exist sufficiently large values of $\mathbf{y}$ and $\mathbf{z}$,
such that
\[
\textrm{RHS of (\ref{eq:betaDef})}\leq\Omega+\delta
\]

We assume sufficiently large $\mathbf{y},\mathbf{z}$ in the following,
in which case we can set $\phi=1+\delta$ and $\psi=\Omega+\delta$
while satisfying (\ref{eq:polyCostProofalphaDefLoose}) (thus (\ref{eq:alphaDef}))
and (\ref{eq:betaDef}).

We then note that from (\ref{eq:polyLocalCost}), (\ref{eq:polyMigCost}),
and the definition of $\Omega$, we have 
\begin{align*}
\widetilde{D}(\phi\psi\mathbf{x}^{*}) & \leq(\phi\psi)^{\Omega}\widetilde{D}(\mathbf{x}^{*})\\
 & =((1+\delta)(\Omega+\delta))^{\Omega}\widetilde{D}(\mathbf{x}^{*})\\
 & =\left(\Omega^{\Omega}+\delta'\right)\widetilde{D}(\mathbf{x}^{*})
\end{align*}
where $\delta'\triangleq\delta\Omega+\delta+\delta^{2}>0$ is an arbitrary
constant (because $\delta$ is an arbitrary constant). The first inequality
is because of $\phi,\psi\geq1$ and $\widetilde{D}(\phi\psi\mathbf{x}^{*})$
is a polynomial of $\phi\psi\mathbf{x}^{*}$ with maximum order of
$\Omega$, where we note that $\mathbf{y}$ and $\mathbf{z}$ are
both linear in $\mathbf{x}$. 

We then have 
\begin{equation}
\frac{\widetilde{D}(\mathbf{x})}{\widetilde{D}(\mathbf{x}^{*})}\leq\frac{\widetilde{D}(\phi\psi\mathbf{x}^{*})}{\widetilde{D}(\mathbf{x}^{*})}=\Omega^{\Omega}+\delta'\label{eq:polyCostProofBeforeFinalResult}
\end{equation}

Until now, we have shown that (\ref{eq:polyCostProofBeforeFinalResult}) holds for sufficiently large
$\mathbf{y}$ and $\mathbf{z}$. According to Assumption \ref{condition:boundedArrDept},
the number of instances that unpredictably leave the system in each
slot is upper bounded by a constant $B_{d}$. It follows that $\mathbf{y}$
and $\mathbf{z}$ increases with $M$ when $M$ is larger than a certain
threshold. Therefore, there exists a sufficiently large $M$, so that
we have a sufficiently large $\mathbf{y}$ and $\mathbf{z}$ that
satisfies (\ref{eq:polyCostProofBeforeFinalResult}). 

Hence, the competitive
ratio upper bound can be expressed as
\begin{equation}
\Gamma\triangleq\max_{\mathcal{I}(M)}\Gamma(\mathcal{I}(M))\leq\Omega^{\Omega}+\delta'\label{eq:polyCostProofFinalResult}
\end{equation}
for sufficiently large $M$. 

According to the definition of the big-$O$
notation, we can also write 
\begin{equation}
\Gamma=O(1)\label{eq:polyCostProofFinalResultBigO}
\end{equation}
because $\Omega$ and $\delta'$ are both constants in $M$.
\end{proof}

\section{Proof of Proposition \ref{prop:costDifferenceBound}}

\label{sec:costDifferenceBoundProof}

Define $T_{\textrm{max}}>1$ as an arbitrarily large timeslot index.
We note that there are $\left\lfloor \frac{T_{\textrm{max}}}{T}\right\rfloor $
\emph{full} look-ahead windows of size $T$ within timeslots from
$1$ to $T_{\textrm{max}}$, where $\left\lfloor x\right\rfloor $
denotes the integral part of $x$. In the last window, there are $T_{\textrm{max}}-T\cdot\left\lfloor \frac{T_{\textrm{max}}}{T}\right\rfloor $
slots. We have 
\begin{equation}
F\left(T_{\textrm{max}}-T\cdot\left\lfloor \frac{T_{\textrm{max}}}{T}\right\rfloor \right)\leq\frac{T_{\textrm{max}}-T\cdot\left\lfloor \frac{T_{\textrm{max}}}{T}\right\rfloor }{T}F(T)\label{eq:proofFBoundConvex}
\end{equation}
because $F(T)$ is convex non-decreasing and $F(0)=0$. 

For the true optimal configuration $\boldsymbol{\pi}^{*}$, according
to the definitions of $\epsilon(\tau)$ and $F(T)$, the difference
in the predicted and actual sum-costs satisfies 
\begin{align}
 & \sum_{t=1}^{T_{\textrm{max}}}D_{\boldsymbol{\pi}^{*}}(t)-\sum_{t=1}^{T_{\textrm{max}}}A_{\boldsymbol{\pi}^{*}}(t)\nonumber \\
 & \leq\left\lfloor \frac{T_{\textrm{max}}}{T}\right\rfloor F(T)+F\left(T_{\textrm{max}}-T\cdot\left\lfloor \frac{T_{\textrm{max}}}{T}\right\rfloor \right)\nonumber \\
 & \leq\frac{T_{\textrm{max}}}{T}F(T)\label{eq:boundProof1}
\end{align}
where the last inequality follows from (\ref{eq:proofFBoundConvex}).
Similarly, for the configuration $\boldsymbol{\pi}_{p}$ obtained
from predicted costs, we have
\begin{equation}
\sum_{t=1}^{T_{\textrm{max}}}A_{\boldsymbol{\pi}_{p}}(t)-\sum_{t=1}^{T_{\textrm{max}}}D_{\boldsymbol{\pi}_{p}}(t)\leq\frac{T_{\textrm{max}}}{T}F(T)\label{eq:boundProof2}
\end{equation}

In the following, we establish the relationship between $\boldsymbol{\pi}^{*}$
and $\boldsymbol{\pi}_{p}$. Assume that, in (\ref{eq:objFuncFrame}),
we neglect the migration cost at the beginning of each look-ahead
window, i.e. we consider each window independently and there is no
migration cost in the first timeslot of each window, then we have
\[
\sum_{t=1}^{T_{\textrm{max}}}D_{\boldsymbol{\pi}_{p}}(t)\leq\Gamma\sum_{t=1}^{T_{\textrm{max}}}D_{\boldsymbol{\pi}^{*}}(t)
\]
where the constant $\Gamma\geq1$ is the competitive ratio of solving
(\ref{eq:objFuncFrame}). This holds because there is no connection
between different windows, thus the optimal sequences (considering
predicted costs) obtained from (\ref{eq:objFuncFrame}) constitute
the optimal sequence up to a factor $\Gamma$ for all timeslots $\left[1,T_{\textrm{max}}\right]$.
Now we relax the assumption and consider the existence of migration
cost in the first slot of each window. Note that we cannot have more
than $\left\lfloor \frac{T_{\textrm{max}}}{T}\right\rfloor +1$ windows
and the first timeslot $t=1$ does not have migration cost. Thus,
\begin{equation}
\sum_{t=1}^{T_{\textrm{max}}}D_{\boldsymbol{\pi}_{p}}(t)\leq\Gamma\sum_{t=1}^{T_{\textrm{max}}}D_{\boldsymbol{\pi}^{*}}(t)+\frac{T_{\textrm{max}}}{T}\sigma\label{eq:boundProof3}
\end{equation}
The bound holds because regardless of the configuration in slot $t_{0}-1$,
the migration cost in slot $t_{0}$ cannot exceed $\sigma$.

By multiplying $\Gamma$ on both sides of (\ref{eq:boundProof1})
and summing up the result with (\ref{eq:boundProof3}), we get
\begin{equation}
\sum_{t=1}^{T_{\textrm{max}}}D_{\boldsymbol{\pi}_{p}}(t)-\Gamma\sum_{t=1}^{T_{\textrm{max}}}A_{\boldsymbol{\pi}^{*}}(t)\leq\frac{T_{\textrm{max}}}{T}\left(\Gamma F(T)+\sigma\right)\label{eq:boundProof4}
\end{equation}
Summing up (\ref{eq:boundProof2}) with (\ref{eq:boundProof4}), dividing
both sides by $T_{\textrm{max}}$, and taking the limit on both sides
yields the proposition.

\section{Proof of Proposition \ref{prop:optT} and Corollary \ref{prop:optTcorr}}

\label{sec:costOptTProof}

Taking the derivative of $\Phi(T)$, we get
\begin{align}
\frac{d\Phi}{dT} & =(\Gamma+1)T\frac{d^{2}G(T)}{dT{}^{2}}\geq0
\end{align}
where the last inequality is because $G(T)$ is convex. This implies
that $\Phi(T)$ is non-decreasing with $T$. Hence, there is at most
one consecutive interval of $T$ (the interval may only contain one
value) such that (\ref{eq:diffTEqu0}) is satisfied. We denote this
interval by $\left[T_{-},T_{+}\right]$, and a specific solution to
(\ref{eq:diffTEqu0}) is $T_{0}\in\left[T_{-},T_{+}\right]$. 

We note that $\frac{d\ln\theta}{dT}$ and $\Phi(T)$ have the same
sign, because $\frac{d\ln\theta}{dT}\lessgtr0$ yields $\Phi(T)\lessgtr0$
and vice versa, which can be seen from (\ref{eq:diffLnErrorCost})
and (\ref{eq:diffTEqu0}). When $T<T_{-}$, we have $\Phi(T)<0$ and
hence $\frac{d\ln\theta}{dT}<0$; when $T>T_{+}$, we have $\Phi(T)>0$
and hence $\frac{d\ln\theta}{dT}>0$. This implies that $\ln\theta$,
thus $\theta(T)$, keeps decreasing with $T$ until the optimal solution
is reached, and afterwards it keeps increasing with $T$. It follows
that the minimum value of $\theta(T)$ is attained at $T\in\left[T_{-},T_{+}\right]$.
Because $T_{0}\in\left[T_{-},T_{+}\right]$ and $T^{*}$ is a discrete
variable, we complete the proof of the proposition. 

Noting that we do not consider the convexity of $\theta(T)$ in the
above analysis, we
can also conclude the corollary.

\section{Additional Simulation Results}
\label{sec:additionalSimResult}

\begin{table*}[!t]
\renewcommand{\arraystretch}{1.3}
\caption{Statistics of computation time and FLOP count}
\label{table:StatisticsOfComputResource}
\centering
\begin{tabular}{c | c c c c c}
\hline
Performance measure & Approach & Sum & Mean & Std. dev. & Maximum \\
\hline
\multirow{ 2}{*}{Computation time  (seconds)} & Prec. fut. knowledge (D)  & $4.46\times 10^{3}$ & $4.39$ & $4.78$ & $50.9$ \\
& Proposed (E)  & $6.45\times 10^{3}$  & $1.84$ & $0.72$ & $3.22$\\
\hline
\multirow{ 2}{*}{FLOP count} & Prec. fut. knowledge (D)  & $1.71\times 10^{10}$  & $1.68\times 10^{7}$ & $1.85\times 10^{7}$  & $2.03\times 10^{8}$\\
& Proposed (E) & $2.14\times 10^{10}$  & $6.08\times 10^{6}$  & $2.43\times 10^{6}$  & $7.56\times 10^{6}$\\
\hline
\end{tabular}
\end{table*}

\begin{figure}
\centering \subfigure[]{\includegraphics[width=1\columnwidth]{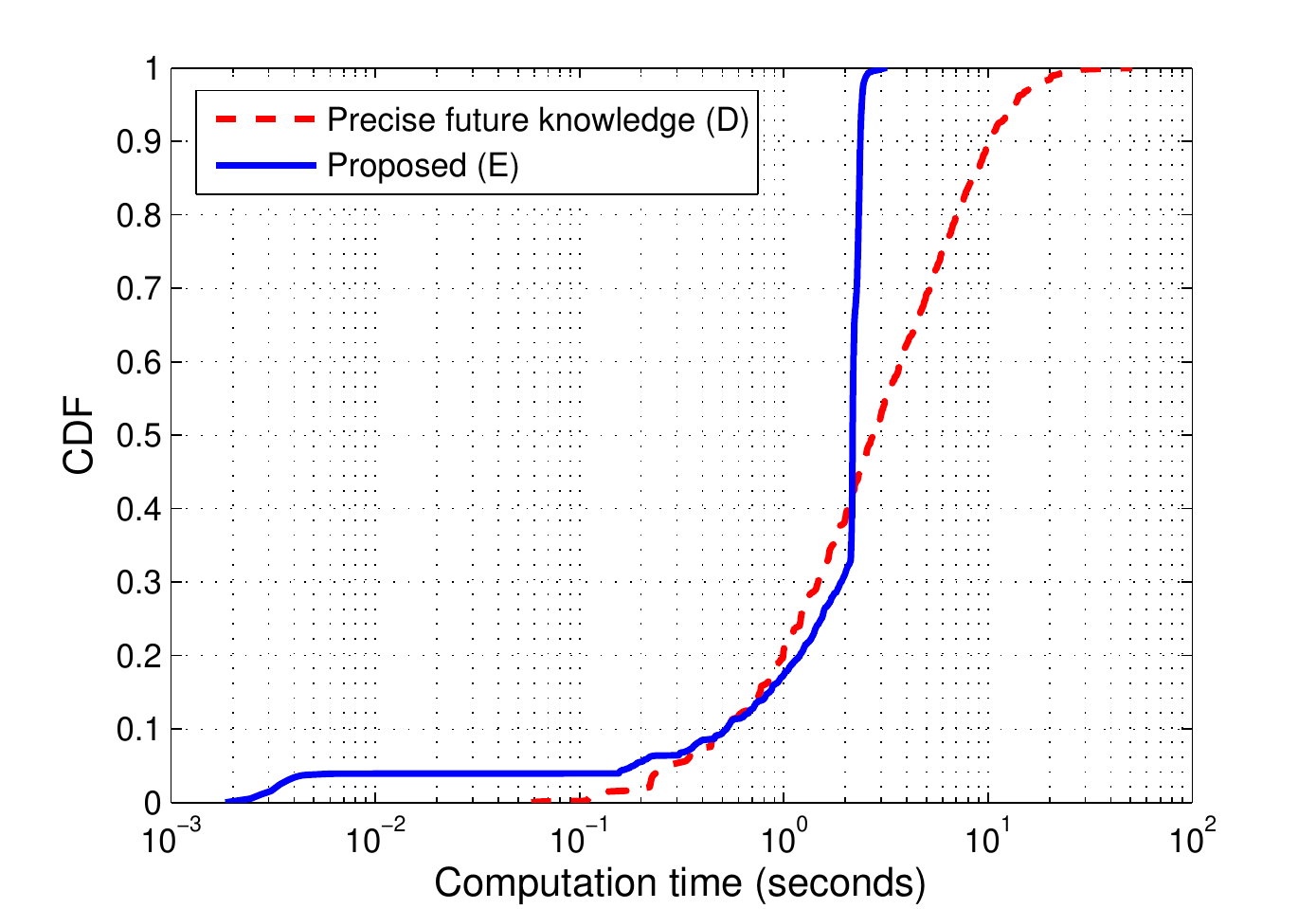}}
\subfigure[]{\includegraphics[width=1\columnwidth]{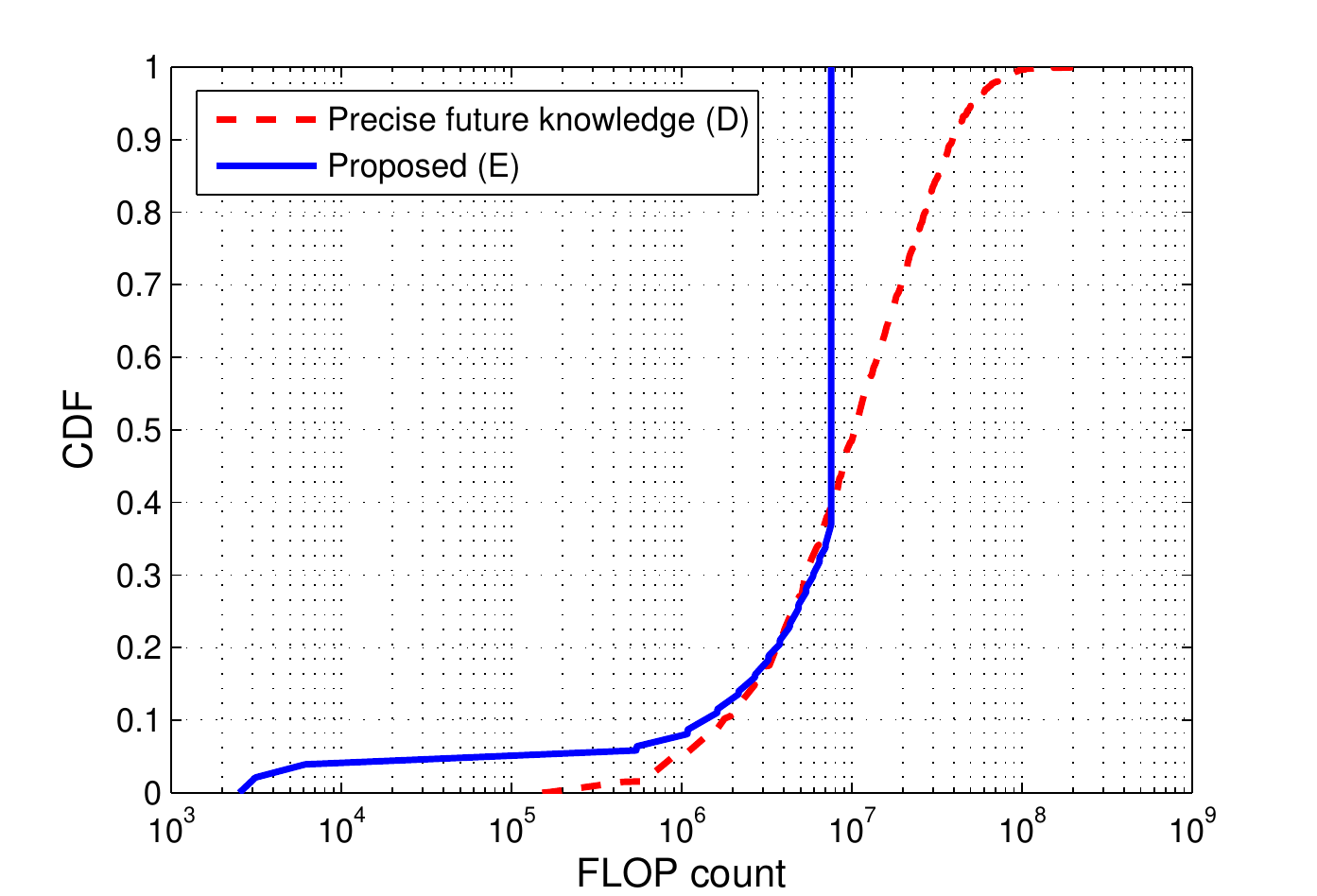}}

\protect\caption{Cumulative distribution functions (CDF) of computation time and FLOP count.}
\label{fig:simComputResource}
\end{figure}

We study the computational overhead of the proposed algorithm for obtaining service configurations corresponding to the results shown in Fig. \ref{fig:simRealWorld}(a). The simulations were run in MATLAB on a laptop with Intel(R) Core(TM) i5-4300U CPU, 4~GB memory, and 64-bit Windows 10.
We focus on the computation time and amount of floating-point operations (FLOP) for every computation of the next $T$-slot configuration for one instance. We compare the performance between the proposed approach (E) and the precise future knowledge scenario (D). We do not compare against MMC-only and backend cloud-only approaches (A, B, and C) because those approaches are much simpler but produce a significantly higher cost than approaches D and E (see Fig. \ref{fig:simRealWorld}(a)) thus it is not fair to compare their computational overheads.
The FLOP counts are obtained using the toolbox in \ref{appReference1} (references for appendices are listed at the end of this document).

We collected statistics of the computation time and FLOP count for each execution of the algorithm to compute the next $T$-slot configuration when using approaches D and E.
Fig.~\ref{fig:simComputResource} shows the cumulative distribution function (CDF), where $\textrm{CDF}(x)$ is defined as the percentage of observations (i.e., computation time or FLOP count) that has a value smaller than or equal to $x$. Table~\ref{table:StatisticsOfComputResource} lists the sum (for all executions), mean, standard deviation, and maximum values.

We see that in terms of computation time, the proposed approach (E) can find the configuration within $1.84$~s on average. Recalling that each timeslot corresponds to a physical time of $60$~s and the configuration for multiple timeslots is usually found by a single execution of the algorithm, the time needed for running the algorithm is relatively short. 

We note that because we ran the simulation in MATLAB and our code is not optimized for performance, our results are pessimistic and the algorithm should run much faster with a better optimized code written in C, for instance. 
This can be seen by the fact that the average FLOP count of the proposed approach is only $6.08 \times 10^6$ while modern personal computers can easily complete over $10^9$ floating-point operations per second (FLOPS),  see the experimental results in \ref{appReference2} for example.
Server machines will be even more powerful. 
This means that if we fully utilize the resource of a personal computer, the proposed algorithm should be able to find the configuration (for a single instance) within $0.01$~s. Such a timescale will be sufficient for deploying newly arrived instances on the cloud in real time.
From this FLOP count comparison, we can conclude that the proposed algorithm only consumes a small amount of processing capability, so it is applicable in practice and scalable to a reasonably large amount of users. 

We further note that the proposed algorithm may benefit from parallelization and performing the computation on GPUs. For example, parallelization is possible for the ``for all'' loops between Lines~\ref{alg:shortestPath:forall1:start} and \ref{alg:shortestPath:forall1:end} and also between Lines~\ref{alg:shortestPath:forall2:start} and \ref{alg:shortestPath:forall2:end} in Algorithm~\ref{alg:shortestPath}. This can further expedite the process of finding configurations.

Comparing approach D with the proposed approach (E), Table~\ref{table:StatisticsOfComputResource} shows that the mean, standard deviation, and maximum values of the proposed approach is lower than those of approach D. This is because the proposed approach uses the optimal look-ahead window size, which is usually much smaller than the total number of timeslots in which the instance remains active, whereas approach D is assumed to have precise knowledge of the future and therefore takes into account the entire duration in which the instance is active. According to the complexity analysis in Section~\ref{sub:AlgorithmOffline}, considering more timeslots in the optimization causes higher algorithmic complexity, thus explaining the result. 

When looking at the sum values in Table~\ref{table:StatisticsOfComputResource}, the proposed approach has a similar but slightly larger computation time and FLOP count than approach D. This is because approach D is assumed to know exactly when an instance departs from the system, whereas the proposed approach does not have this knowledge and may consider additional slots after instance departure as part of the optimization process. However, the gap is not large because the proposed approach has a relatively small look-ahead window; the instance departure becomes known to the proposed approach at the beginning of a new look-ahead window.

The fact that the mean and maximum values (in Table~\ref{table:StatisticsOfComputResource}) of the proposed approach are much smaller than those of approach D also makes the proposed approach more suitable for real-time requests, because it can apply a configuration quickly after instance arrival. More importantly, we recall that approach~D is impractical because it assumes precise prediction of future costs and instance departure times.


\section*{References for Appendices}
\begin{enumerate}[label={[A\arabic*] }]
\item  H. Qian, ``Counting the Floating Point Operations (FLOPS),'' MATLAB Central File Exchange, No. 50608, available at \url{http://www.mathworks.com/matlabcentral/fileexchange/50608} \label{appReference1}
\item MaxxPI$^2$ -- the System Bench, \url{http://www.maxxpi.net/pages/result-browser/top15---flops.php} \label{appReference2}
\end{enumerate}

\end{document}